\documentclass[acmsmall,screen]{acmart}
\pdfoutput=1

\usepackage{booktabs}   
\usepackage{subcaption} 

\usepackage{amssymb}
\usepackage{amsmath}
\usepackage{cmll} 
\usepackage{stmaryrd} 
\usepackage{enumitem} 
\usepackage{proof} 
\usepackage[disable]
{todonotes}
\newcommand\todom[2][]{\todo[color=blue!20,#1]{#2}{}} 
\newcommand\todot[2][]{\todo[color=red,#1]{#2}{}} 
\newcommand\todott[2][]{\todo[color=yellow,#1]{#2}{}} 
\newcommand\todoc[2][]{\todo[color=pink,#1]{#2}{}} 

\usepackage[normalem]{ ulem }
\usetikzlibrary{calc} %
\usetikzlibrary{arrows}

\usepackage[utf8]{inputenc}
\newtheorem{remark}{Remark} 

\newcommand\ass{::=}   

\newcommand\real{\ensuremath{\mathbb R}}   
\newcommand\preal{\ensuremath{\mathbb R^+}}   

\newcommand\ppcf{\ensuremath{\mathrm{PPCF}}}   
\newcommand\var{\mathcal{V}}   
\newcommand\Terms[1][]{\ensuremath{\Lambda^{#1}}}   
\newcommand\TermsFree[2][]{\ensuremath{\Lambda^{#1}_{#2}}}   
\newcommand\typea{A} 
\newcommand\typeb{B}
\newcommand\typec{C}
\newcommand\terma{M} 
\newcommand\termb{N}
\newcommand\termc{L}
\newcommand\fterma{S} 
\newcommand\ftermb{T}
\newcommand\ftermc{T'}

\newcommand\varreal{z}   
\newcommand\subst[1]{\mathrm{Subst}_{#1}} 

\newcommand\bernoulli{\mathtt{bernoulli}} 
\newcommand\normal{\mathtt{normal}} 
\newcommand\gaussian{\mathtt{gauss}} 
\newcommand\observe[1]{\ensuremath{\mathtt{observe}(#1)}}   

\newcommand\Const{\ensuremath{\mathcal C}} 
\newcommand\Red{\ensuremath{\mathrm{Red}}} 
\newcommand\realsub[1][]{\sigma_{#1}}   

\newcommand\treal{\ensuremath{\mathcal R}}   
\newcommand\num[1]{\underline{#1}}   
\newcommand\realfun[1]{\ensuremath{\underline{#1}}}   
\newcommand\fix{\ensuremath{\mathtt Y}}   
\newcommand\ifterm[4]{\ensuremath{\mathtt{if}(#1\in #2,#3,#4) }}   
\newcommand\ifz[3]{\ensuremath{\mathtt{ifz}(#1,#2,#3) }}   
\newcommand\oracle{\ensuremath{\mathtt{sample}}}   
\newcommand\letterm[3]{\ensuremath{\mathtt{let}(#1,#2,#3) }}   

\newcommand\semtype[1]{\llbracket #1 \rrbracket} 
\newcommand\semterm[2][]{\llbracket #2\rrbracket^{#1}} 

\newcommand\Kern{\ensuremath{\mathbf{Kern}}}   
\newcommand\Meas{\ensuremath{\mathbf{Meas}}}   

\newcommand{\Endproof}{
  \ifmmode 
  \else \leavevmode\unskip\penalty9999 \hbox{}\nobreak\hfill
  \fi
  \quad\hbox{$\Box$}
  \par\medskip}

\newcommand\Eqref[1]{(\ref{#1})}

\newcommand\IE{\textsl{i.e.}~}
\newcommand\Eg{\textsl{e.g.}~}

\renewcommand{\phi}{\varphi}
\renewcommand\epsilon{\varepsilon}

\newcommand{\Implies}{\Rightarrow}

\newcommand{\St}{\mid}

\renewcommand{\Bot}{{\mathord{\perp}}}

\newcommand\cF{\mathcal{F}}

\newcommand\cM{\mathcal{M}}

\newcommand\cP{\mathcal{P}}

\newcommand\cX{\mathcal{X}}

\newcommand\Myleft{}
\newcommand\Myright{}

\newcommand\Web[1]{\Myleft|{#1}\Myright|}
\newcommand\Orth[2][]{#2^{\Bot_{#1}}}

\newcommand\Pair[2]{\langle{#1},{#2}\rangle}

\newcommand\Biorth[1]{#1^{\Bot\Bot}}

\newcommand\Card[1]{\#{#1}}

\newcommand\Locun[1]{1^J}

\newcommand\Isom\simeq

\newcommand\Comp{\mathrel\circ}

\newcommand\Nat{{\mathbb{N}}}

\newcommand\Natnz{{\Nat^+}}

\newcommand\Biind[2]{\genfrac{}{}{0pt}{1}{#1}{#2}}

\newcommand\List[3]{#1_{#2},\dots,#1_{#3}}

\newcommand\Real{\mathbb{R}}
\newcommand\Realp{\mathbb{R}^+}
\newcommand\Realpto[1]{(\Realp)^{#1}}

\newcommand\Izu{[0,1]}

\newcommand\Dirac[1]{\delta_{#1}}

\newcommand\Ev{\operatorname{\mathsf{Ev}}}

\newcommand\Redst[1]{\mathop{\mathsf{Red}}}

\newcommand\Tuple[1]{\langle{#1}\rangle}

\newcommand\Msetofsubst[1]{\bar F}

\newcommand\Pcoh[1]{\mathsf P{#1}}

\newcommand\Retri\zeta
\newcommand\Retrp\rho

\newcommand\Impl[2]{{#1}\Rightarrow{#2}}

\newcommand\Tnat\iota

\newcommand\Loop\Omega

\newcommand\Tseq[3]{{#1}\vdash{#2}:{#3}}

\newcommand\Timpl\Impl
\newcommand\Simpl\Impl

\newcommand\Id{\operatorname{\mathrm{Id}}}

\newcommand\Proj[1]{{\mathrm{pr}}_{#1}}

\newcommand\Relincl\eta
\newcommand\Relrestr\rho

\newcommand\Vect[1]{\vec{#1}}

\newcommand\Ccarrier[1]{#1}
\newcommand\Ccarrierb[1]{#1}
\newcommand\Cnorm[2]{\|#2\|_{#1}}
\newcommand\Fnorm[1]{\|#1\|}
\newcommand\Cuball[1]{\mathcal B#1}
\newcommand\Cuballp[1]{\mathcal B(#1)}
\newcommand\Cuballr[2]{\Cuball{\Crel{#1}{#2}}}

\newcommand\Fdiffp[4]{\Delta^{#1}#2(#3;#4)}
\newcommand\Fdiff[3]{\Delta#1(#2;#3)}

\newcommand\Cocard[2]{\cP_{#1}(#2)}

\newcommand\Scone[2]{\mathsf S^{#1}(#2)}

\newcommand\INFSCOTT{\mathbf{Cstab}}
\newcommand\INFSCOTTM{\INFSCOTT_{\mathsf m}}

\newcommand\Crel[2]{{#1}_{#2}}

\newcommand\Tribu[1]{\Sigma_{#1}}
\newcommand\Bmes[1]{\mathsf{Meas}(#1)}

\newcommand\Diracf{\delta}

\newcommand\Dual[1]{{#1}'}

\newcommand\Mtest[2]{\mathsf M^{#1}(#2)}

\newcommand\Mtestapp[2]{{#1}*{#2}}

\newcommand\Mpath[2]{\mathsf{Path}^{#1}(#2)}
\newcommand\Mpathu[2]{\mathsf{Path}^{#1}_1(#2)}

\newcommand\Mapp[1]{\epsilon_{#1}}

\newcommand\Implm[2]{{#1}\Rightarrow_{\mathsf m}{#2}}

\newcommand\Mfun[1]{\cM^{#1}}

\newcommand\Mtestfun[2]{{#1}\triangleright{#2}}

\newcommand\Mchar[1]{\chi_{#1}}
\newcommand\Mext[1]{{#1}^\dagger}

\newcommand\Cofpcoh[1]{\widehat{#1}}

\newcommand\Wpor{\mathsf{wpor}}

\setcopyright{acmlicensed}
\acmPrice{}
\acmDOI{10.1145/3158147}
\acmYear{2018}
\copyrightyear{2018}
\acmJournal{PACMPL}
\acmVolume{2}
\acmNumber{POPL}
\acmArticle{59}
\acmMonth{1}

\begin{document}
\title[Measurable Cones and Stable, Measurable Functions]{Measurable Cones and Stable, Measurable Functions}         
\subtitle{A Model for Probabilistic Higher-Order Programming}                     


\author{Thomas Ehrhard}
\affiliation{
  \institution{IRIF UMR 8243, Universit\'e Paris Diderot, Sorbonne Paris Cit\'e, CNRS}            
  \city{F-75205 Paris}
  \country{France}
}
\email{ehrhard@irif.fr}
\author{Michele Pagani}
\affiliation{
  \institution{IRIF UMR 8243, Universit\'e Paris Diderot, Sorbonne Paris Cit\'e, CNRS}            
  \city{F-75205 Paris}
  \country{France}
}
\email{pagani@irif.fr}
\author{Christine Tasson}
\affiliation{
  \institution{IRIF UMR 8243, Universit\'e Paris Diderot, Sorbonne Paris Cit\'e, CNRS}            
  \city{F-75205 Paris}
  \country{France}
}
\email{tasson@irif.fr}


\begin{abstract}
We define a notion of stable and measurable map between cones endowed with measurability tests and show that it forms a cpo-enriched cartesian closed category. This category gives a
denotational model of an extension of PCF supporting the main primitives of probabilistic functional programming, like continuous and discrete probabilistic distributions, sampling, conditioning and full recursion. We prove the soundness and adequacy of this model with respect to a call-by-name operational semantics and give some examples of its denotations. 
\end{abstract}


\setcopyright{acmlicensed}
\acmJournal{PACMPL}
\acmYear{2018} \acmVolume{2} \acmNumber{POPL} \acmArticle{59} \acmMonth{1} \acmPrice{}\acmDOI{10.1145/3158147}


\maketitle

\section{Introduction}

%

Around the 80's, people started to apply formal methods to the
analysis and design of probabilistic programming languages. In
particular, \citet{Kozen81} defined a denotational semantics for a first-order
\texttt{while}-language endowed with a random real number generator. In that setting, programs can be seen as stochastic
kernels between measurable spaces: the possible configurations of the
memory are described by measurable spaces, with the measurable sets
expressing the observables, while kernels define the probabilistic
transformation of the memory induced by program execution. For
example, a \texttt{while}-program using $n$ variables taking values in
the set $\Real$ of real numbers is a stochastic kernel $K$ over the
Lebesgue $\sigma$-algebra on $\Real^n$ (see 
%
Compendium of Measures and Kernels, 
Section~\ref{section:compendium}) --- \IE~$K$ is a function
taking a sequence $\vec r\in\Real^n$ and a measurable set $U\subseteq
\Real^n$ and giving a real number $K(\vec r, U)\in[0,1]$, which is the
probability of having the memory (\IE~the values of the $n$ variables)
within $U$ after having executed the program with the memory
initialized as $\vec r$.

Kozen's approach cannot be trivially extended to higher-order types,
because there is no clear notion of measurable subset for a functional
space, e.g. we do not know which measurable space can describe values
of type, say, $\Real\rightarrow\Real$ (see~\cite{aumann1961} for
details).

 \citet{PANANGADEN1999} reframed the work by Kozen in a categorical setting, using the category \Kern{} of stochastic kernels. This
 category has been presented as the Kleisli category of the so-called Giry's
 monad \citep{Giry1982} over the category \Meas{} of measurable spaces and
 measurable functions. 
 One can precisely state the
 issue for higher-order types in this framework --- both \Meas{}
 and \Kern{} are cartesian categories but not closed.

The quest for a formal syntactic semantics of higher-order
probabilistic programming had more success.
We mention in particular
\citet{Park:2008}, 
proposing a probabilistic functional language $\lambda_\bigcirc$ based
on sampling functions. This language
has a type $\treal$
of sub-probabilistic distributions over the set of real
numbers\footnote{In \citep{Park:2008} $\treal$ is written
$\bigcirc\texttt{real}$. One should consider \emph{sub}-probabilistic
distributions because program evaluation may diverge.}, \IE~ measures
over the Lebesgue $\sigma$-algebra on $\real$ with total mass at most
$1$. 
Using the usual functional primitives (in particular recursion)
together with the uniform distribution over $[0,1]$ and a sampling
construct, the authors 
encode various methods for generating
distributions (like the inverse transform method and rejection
sampling) and computing properties about them (such as approximations
for expectation values, variances, etc). The amazing feature of
$\lambda_\bigcirc$ is its rich expressiveness as witnessed by the
number of examples and applications detailed in \citep{Park:2008},
showing the relevance of the functional paradigm for probabilistic
programming.

Until now, $\lambda_\bigcirc$ lacked a denotation model, 
\citep{Park:2008} sketching only an operational semantics. In
particular, the correctness proof of the encodings follows 
a syntactic reasoning which is not compositional. Our paper fills this
gap, giving a denotational model to a variant of
$\lambda_\bigcirc$. As a byproduct, we can check program correctness in a
straight way by applying to program denotations the standard laws of calculus
(Example~\ref{ex:sem_distributions},\ref{ex:sem_expect}), even for recursive
programs (Example~\ref{ex:semantics_observe})
. This method is justified by the
Adequacy Theorem~\ref{th:adequacy}
stating the correspondence
between the operational and the denotational semantics. 

If we restrict the language to \emph{countable} data types (like booleans and
natural numbers, excluding the real numbers), then the situation is much simpler.
 Indeed, any distribution over a countable set is \emph{discrete}, i.e. it can be described as a linear combination of its possible outcomes and there is no need of a notion of measurable space. In previous papers \cite{EhrPagTas11,EhrPagTas14,EhrhardTasson16}, we have shown that the category $\mathbf{PCoh}_\oc$ of \emph{probabilistic coherence spaces} and entire functions  gives fully abstract denotational models of functional languages extended with a random \emph{natural number} generator. The main goal of this work is to generalize these models in order to account for \emph{continuous} data types also. 
 
 
The major difficulty for such a generalization is that a probabilistic coherence space is defined with respect to a kind of canonical basis (called \emph{web}) that, at the level of ground types, corresponds to the possible samples of a distribution. For continuous data types, these webs should be replaced by measurable spaces, and then one is stuck on the already mentioned impossibility of associating a measurable space with a  functional type --  both \Meas{} and \Kern{} being not cartesian closed.


Our solution 
is to replace 
probabilistic coherence spaces with 
cones~\cite{ando1962fundamental}, already used by \citet{Selinger2004}, allowing for an axiomatic presentation not
referring to a web.
A cone is similar to a normed vector space,
but with non-negative real scalars  (Definition~\ref{def:cone}).
 Any probabilistic coherence space can be seen as a cone (Example~\ref{ex:pcoh}) as well as the set $\Bmes{X}$ of all bounded measures
over a measurable space $X$ (Example~\ref{ex:measures-cone}). In particular, the cone $\Bmes\real$ associated with the Lebesgue $\sigma$-algebra on $\real$ 
will be our interpretation of the ground type $\treal$.

What about functional types, e.g.~$\treal\rightarrow\treal$? \citet{Selinger2004} studied the notion of Scott continuous maps
between cones, \IE~monotone non-decreasing bounded maps which commute with the lub of
non-decreasing sequences\footnote{Actually, Selinger considers lubs of directed
sets, but non-decreasing chains are enough for our purposes. Moreover, because
we need to use the monotone convergence theorem for guaranteeing the
measurability of these lubs in function spaces, completeness wrt.~arbitrary
directed sets would be a too strong requirement (see
Section~\ref{sec:pre-stable-fun-space}): a crucial feature of measurable sets is
that they are closed under \emph{countable} (and not arbitrary)
unions.}.
 The set of these functions also forms a cone with the algebraic operations defined pointwise. However, this
cone construction does not yield a cartesian closed category, namely 
the currying of a Scott
continuous map can fail to be monotone non-decreasing, hence Scott continuous (see
discussion in Section~\ref{sub:Scott_not_ccc}). The first relevant contribution of our paper is then to introduce a notion of \emph{stable} map, meaning Scott
continuous \emph{and} ``absolutely monotonic'' (Definition~\ref{def:Sn}), which solves the problem about currying and gives a
cartesian closed category.\todom{I moved all discussions about the order-completeness and the fact that we consider the algebraic order to the section on the cones, since it is more technical and linked to the wpor example}

We borrow the term of ``stable function'' from Berry's analysis of sequential computation \citep{stability}. In fact, our definition is deeply related with a notion of ``probabilistic'' sequentiality, as we briefly mention in Section~\ref{sub:Scott_not_ccc} showing that it rejects the ``parallel or'' (but not the ``Gustave function''). 

The notion of stability is however not enough to interpret all primitives of
probabilistic functional programming. One should be able to integrate at least
first-order functions in order to sample programs denoting probabilistic
distributions (e.g. see the denotation of the $\texttt{let}$ construct in
Figure~\ref{fig:interpretation}). The problem is that there are stable
functions which are not measurable, so not Lebesgue integrable
(Section~\ref{sect:measurability}). We therefore equip the cones with a
notion of measurability tests (Definition~\ref{def:meas-tests}), inducing a
notion of measurable paths (Definition~\ref{def:measurable_path}) in a cone. In
the case the cone is associated with a standard measurable space $X$, \IE~it
is of the form $\Bmes{X}$, then the measurability tests are the
measurable sets of $X$. However, at higher-order types the definition is less
immediate. The crucial point is that the measurable paths in $\Bmes{X}$ are
Lebesgue integrable, as expected (Section~\ref{subsect:integr_meas_paths}). We
then call \emph{measurable} a stable map preserving measurable paths and we
prove that it gives a cartesian closed category, denoted $\INFSCOTTM$
(Figure~\ref{fig:ccc} and Theorem~\ref{th:stab-CCC}).

To illustrate the expressiveness of $\INFSCOTTM$ we consider a variant of Scott and
Plotkin's PCF \citep{plotPCF} with numerals for real
numbers, a
constant \texttt{sample} denoting the uniform distribution over $[0,1]$ and a
\texttt{let} construct over the ground type. This language is as expressive as $\lambda_\bigcirc$ of \citet{Park:2008} (namely, the \texttt{let} construct corresponds to the sampling of $\lambda_\bigcirc$). The only notable difference lies in the  call-by-name operational semantics (Figure~\ref{fig:beta}) that we adopt,  while~\citep{Park:2008} 
follows a call-by-value strategy.\footnote{Let us underline that our \texttt{let} construct does not allow to encode the call-by-value strategy at higher-order types, since it is restricted to the ground type $\treal$. See Section~\ref{sect:ppcf} for more details.} Our choice is motivated by the fact that the call-by-name model is 
simpler to present than the call-by-value one. We plan to detail this latter in a forthcoming paper.  

We also decided not to consider the so-called soft-constraints, which are
implemented in e.g.~\citep{BorgstromLGS16,StatonYWHK16,Staton17} with a
construct called \texttt{score}. This can be added to our language by using a
kind of exception monad in order to account for the possible failure of
normalization, as detailed in \citep{Staton17} (see Remark~\ref{rk:score}). Also
in this case we prefer to omit this feature for focussing on the true novelties
of our approach --- the notions of stability and measurability.

Let us underline that although the definition of $\INFSCOTTM$ and the proof of its cartesian closeness are not trivial, the denotation of the programs (Figure~\ref{fig:interpretation}) is completely standard, extending the usual interpretation of PCF programs as Scott continuous functions \citep{plotPCF}. We prove the soundness (Proposition~\ref{prop:soundmess}) and the adequacy (Theorem~\ref{th:adequacy}) of $\INFSCOTTM$. A major byproduct of this result is then to make it possible to reason about higher-order programs as functions between cones, which is quite convenient when working with programs acting on measures.  

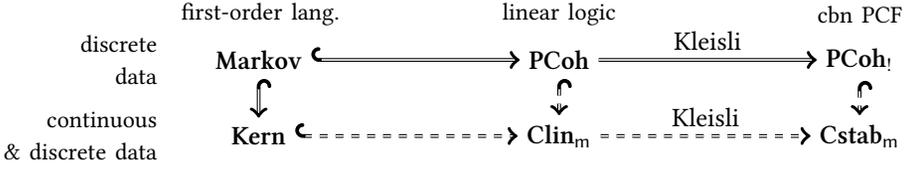
\begin{figure}
\centering
\begin{tikzpicture}
\node[text width=2.5cm, align=right] (d) at (-2.6,0) {\small discrete\\ data};
\node[text width=2.5cm, align=right] (d) at ($(d)-(0,1)$) {\small continuous\\ \& discrete data};
\node (markov) at (0,0) {$\mathbf{Markov}$};
\node (pcoh) at (4,0) {$\mathbf{PCoh}$};
\node (pcoh_oc) at (8,0) {$\mathbf{PCoh}_\oc$};
\node[text width=2.5cm, text centered] (fo) at ($(markov)+(0,.6)$) {\small first-order lang.};
\node[text width=2cm, text centered] (ll) at ($(pcoh)+(0,.6)$) {\small linear logic};
\node[text width=1.5cm, text centered] (lambda) at ($(pcoh_oc)+(0,.6)$) {\small cbn PCF};
\node (kernel) at ($(markov)-(0,1)$) {$\Kern{}$};
\node (clin) at ($(pcoh)-(0,1)$) {$\mathbf{Clin}_{\mathsf m}$};
\node (cstab) at ($(pcoh_oc)-(0,1)$) {$\INFSCOTTM$};
\draw[double, right hook->] (markov) -- (pcoh);
\draw[double, ->] (pcoh) --node[above] {Kleisli} (pcoh_oc);
\draw[double, right hook->] (markov) -- (kernel);
\draw[double, dashed, right hook->] (pcoh) -- (clin);
\draw[double, dashed, right hook->] (pcoh_oc) -- (cstab);
\draw[double, dashed, right hook->] (kernel) -- (clin);
\draw[double, dashed, ->] (clin) --node[above] {Kleisli} (cstab);
\end{tikzpicture}
\vspace{-.4cm}
\caption{Relationship between probabilistic coherence categories and  $\INFSCOTTM$. Dashed arrows are conjectures.}\label{fig:the_picture}
\end{figure}
To conclude, let us comment Figure~\ref{fig:the_picture}, sketching the relations 
between the category  $\INFSCOTTM$ achieved here and the category  
$\mathbf{PCoh}_\oc$ of probabilistic coherence spaces
and entire functions which has been the starting point of our approach. 
The two categories give models of the functional primitives (PCF-like languages), but
$\mathbf{PCoh}_\oc$ is restricted to discrete data types, while $\INFSCOTTM$ extends the model 
to continuous types. We guess this extension to be conservative, hence the arrow is hooked but just dashed. 
%
%
%
 We are even convinced that $\INFSCOTTM$ is the result of a Kleisli construction from a more fundamental model $\mathbf{Clin}_{\mathrm m}$ of (intuitionistic) linear logic, based on positive cones and measurable, Scott continuous and linear functions. We plan to study $\mathbf{Clin}_{\mathrm m}$ in an extended version of this paper as a category extending the category \Kern{} of measurable spaces and stochastic kernels. This would close the loop and further confirm the analogy with $\mathbf{PCoh}_\oc$, which is the Kleisli category associated with the exponential comonad of the model based on the category $\mathbf{PCoh}$ of Scott continuous and linear functions between probabilistic coherence spaces, this latter containing the category $\mathbf{Markov}$ of Markov chains as a full sub-category.

\paragraph{Contents.} This paper needs a basic knowledge of measure theory: we briefly recall in Section~\ref{section:compendium} the main notions and notations used. Section~\ref{sect:ppcf} presents the programming language \ppcf --- the probabilistic variant of PCF we use for showing the expressiveness of our model. Figure~\ref{fig_types} gives the grammar of terms and the typing rules, while Equation~\eqref{eq:red} and Figure~\ref{fig:beta} define the kernel $\Red$ describing the stochastic operational semantics. Our first main contribution is presented in Section~\ref{sec:cones}: after having recalled Selinger's definition of cone (Definition~\ref{def:cone}) we study our notion of absolutely monotonic map (Definition~\ref{def:Sn}), or equivalently pre-stable map (Definition~\ref{def:n-pre-stable} and Theorem~\ref{th:non-decreasing-class-equiv}) and we prove that it composes (Theorem~\ref{th:pre-stable_compose}). Stable maps are absolutely monotonic  and Scott-continuous (Definition~\ref{def:stable}). Section~\ref{sect:measurability} introduces our second main contribution, which is the notion of measurability test (Definition~\ref{def:meas-tests}) and measurable map (Definition~\ref{def:mes-pre-stable-function}), giving the category $\INFSCOTTM$ (Definition~\ref{def:mes-pre-stable-function}). Section~\ref{sect:ccc} presents the cartesian closed structure of $\INFSCOTTM$, summarized in Figure~\ref{fig:ccc}. Finally, Section~\ref{sect:sound_adeq} details the model of \ppcf{} given by $\INFSCOTTM$ (Figure~\ref{fig:interpretation}) and states soundness (Proposition~\ref{prop:soundmess}) and adequacy (Theorem~\ref{th:adequacy}). Section~\ref{sect:related_work} discusses the previous literature. Because of space limits, many proofs are omitted and 
postponed and in the technical appendix~\ref{appendix:technical}.

\paragraph{Notations} We use $\Card I$ for the cardinality of a set $I$. The
set of non-negative real numbers is $\Realp$ and its elements are denoted
$\alpha,\beta\dots$. General real numbers are denoted $r,s,t\dots$. The set of
non-zero natural numbers is $\Natnz$. The greek letter $\lambda$ will denote the Lebesgue
measure over $\Real$, $\lambda_{[0,1]}$ being its restriction to the unit
interval. Given a measurable space $X$ and an $x\in X$, we use $\Dirac x$ for
the Dirac measure over $X$: $\Dirac x(U)$ is equal to $1$ if $x\in U$ and to
$0$ otherwise. We also use $\chi_U$ to denote the characteristic function of
$U$ which is defined as $\chi_U(x)$ is equal to $1$ if $x\in U$ and to $0$
otherwise. We use $\Mfun n$ for the set of measurable functions
$\Real^n\to\Realp$. We use $F(\_)$ to denote the map $x\mapsto F(x)$.

%

\section{Compendium of measures and kernels} 
\label{section:compendium}
\newcommand\textdef[1]{\emph{#1}}

A \textdef{$\sigma$-algebra} $\Tribu X$ on a set $X$ is a  family of
subsets of $X$ that is nonempty and closed under complements and countable
unions, so that  $\emptyset,\ X\in\Tribu X$. A \textdef{measurable space} is a pair $(X, \Tribu X)$ of a set $X$ equipped with  a $\sigma$-algebra $\Tribu X$. A \textdef{measurable set} of $(X, \Tribu X)$ is an element of $\Tribu X$. From now on, we will denote a measurable space  $(X, \Tribu X)$ simply by its underlying set $X$, whenever the $\sigma$-algebra is clear or irrelevant. We consider $\Real$ and $\Realp$ as measurable spaces equipped with the Lebesgue $\sigma$-algebra,  
generated by the open intervals. 
 A \textdef{bounded measure} on a measurable space $X$ is a map $\mu : \Tribu X
 \to \Realp$ satisfying 
 $\mu(\biguplus_{i\in I} S_i ) = \sum_{i\in I}\mu(S_i)$ for any countable family $\{S_i\}_{i\in I}$ of
disjoint sets in $\Tribu X$. We call $\mu$ a \textdef{probability} (resp.~\textdef{subprobability}) \textdef{measure}, whenever $\mu(X) =1$ (resp.~$\mu(X)\le 1$).
When $\mu$ is a measure on $\Real^n$, we often call it a \textdef{distribution}.

A \textdef{measurable function} $f : (X,\Tribu X) \to (Y,\Tribu Y)$ is a function $f:X\to Y$ such that
$f^{ -1} (U )\in\Tribu X$ for every $U \in \Tribu Y$.
The \textdef{pushforward measure} $f_\ast \mu$ from a measure $\mu$ on $X$ along a measurable map $f$ is defined as $(f_\ast\mu)(U ) = \mu(f^{-1} (U ))$, for every $U\in\Tribu Y$.

These notions have been introduced in order to define the \textdef{Lebesgue integral} $\int_X f(x)\mu(dx)$ of a generic measurable function $f:X\to \Real$ with respect to a measure $\mu$ over $X$. 
This paper uses only basic facts about the Lebesgue integral which we do not detail here.

Measures are special cases of kernels. A \textdef{bounded kernel} $K$ from $X$ to $Y$ is a function 
$K : X \times\Tribu Y \to \Realp$ such that: (i) for every $x\in X$,  $K(x,\_)$ is a bounded measure over $Y$; (ii) for every $U\in\Tribu Y$, $K(\_,U)$ is a measurable map from $X$ to $\Realp$. A \textdef{stochastic kernel} $K$ is a kernel such that $K(x,\_)$ is a sub-probability measure for every $x\in X$. Notice that a bounded measure (resp.~sub-probability measure) $\mu$ over $X$ can be seen as a particular bounded kernel (resp.~stochastic kernel) from the singleton measurable space $(\{\star\}, \{\emptyset, \{\star\}\})$ to $X$.

\paragraph{Categorical approach.} We use two categories having measurable spaces as objects, denoted respectively \Meas{} and \Kern. 

The category \textdef{\Meas{}} has measurable functions as morphisms.
This category is cartesian (but not cartesian closed), the \textdef{cartesian product $(X, \Tribu X )\times (Y,\Tribu Y )$} of $(X, \Tribu X )$ and $(Y,\Tribu Y )$ is $(X\times Y,\Tribu X\otimes\Tribu Y)$, where $X\times Y$ is the set-theoretic product and $\Tribu X\otimes\Tribu Y$ is the $\sigma$-algebra generated by the rectangles $U\times V$, where $U\in\Tribu X$ and $V\in\Tribu Y$. It is easy to check that the usual projections are measurable maps, as well as that the set-theoretic pairing $\langle f,g\rangle$ of two functions $f: Z \to   X$, $g:Z\to Y$ is a measurable map from $Z$ to $X\times Y$, whenever $f$, $g$ are measurable. 

The  category \textdef{\Kern{}} has stochastic kernels as morphisms\footnote{One can well define the category of bounded kernels also, but this is not used in this paper.}. Given a stochastic kernel $H$ from $X$ to $Y$ and $K$ from $Y$ to $Z$, the \textdef{kernel composition} $K\Comp H$ is a stochastic kernel from $X$ to $Z$ defined as, for every $x\in X$ and $U\in\Tribu Z$:
\begin{equation}\label{eq:kern_composition}
(K\Comp H)(x,U) = \int_Y K(y, U)H(x,dy).
\end{equation}
Notice that the above integral is well-defined because $H(x,\_)$ is a stochastic measure from condition (i) on kernels and $K(\_, U)$ is a measurable function from condition (ii). A simple application of Fubini's theorem gives the associativity of the kernel composition. The \textdef{identity kernel} is the function mapping $(x,U)$ to $1$ if $x\in U$ and to $0$ otherwise.

Unlike \Meas, we consider a tensor product $\otimes$ in $\Kern$ which is a symmetric monoidal product but not the cartesian product\footnote{Indeed, \Kern{} has cartesian products, but we will not use them.}.  The action of $\otimes$ over the objects $X,X'$ is defined as the cartesian product in \Meas, so that we still denote it as $X\times X'$. The tensor of a kernel $K$ from $X$ to $Y$ and $K'$ from $X'$ to $Y'$ is the kernel $K\otimes K'$ given as follows, for $(x,x')\in X\times X'$ and $U\in\Tribu Y$, $U'\in\Tribu{Y'}$ :
\begin{equation}\label{eq:tensor_kernel}
K\otimes K' ((x,x'),U\times U') = K(x,U)K'(x',U')
\end{equation}
Notice that \Kern{} is not closed with respect to $\otimes$. Recall that a measure can be seen as a kernel from the singleton measurable space, so that Equation~\eqref{eq:tensor_kernel} defines also a \textdef{tensor product $\mu\otimes\mu'$} between measures over resp. $Y$ and $Y'$.


The category $\Kern$ has also \textdef{countable coproducts}. Given a countable family $\{(X_i , \Tribu i )\}_{i\in I}$ of measurable spaces, the coproduct $\coprod_{i\in I} (X_i , \Tribu i )$ has as underlining set the disjoint union $\cup_{i\in I} X_i\times\{i\}$ of the $X_i$'s, and as the $\sigma$-algebra the one generated by $\cup_{i\in I} U_i\times \{i\}$ disjoint union of $U_i\in\Tribu i$. The injections $\iota_j$ from $X_j$ to $\coprod_{i\in I} X_i$ are defined as $\iota_i(x,\cup_{j\in I}U_j\times\{j\})=\chi_{U_i}(x)$. 
 Given a family $K_i$ from $X_i$ to $Y$, the copairing $[K_i]_{i\in I}$ from $\coprod_{i\in I} X_i$ to $Y$ is defined by $[K_i]_{i\in I}((x,j),U)=K_j(x,U)$.

Actually, the categories $\Meas$ and $\Kern$ can be related in a very similar way  as the relation between the categories {\bf Set} (of sets and functions) and {\bf Rel} (of sets and relations). In fact, $\Kern$ corresponds to the Kleisli category of the so-called Giry's monad over \Meas{} \citep{Giry1982}, exactly has the category {\bf Rel} of relations is the Kleisli category of the powerset monad over {\bf Set} (see \citep{PANANGADEN1999}). Since this paper does not use this construction, we do not detail it. 

\section{The probabilistic language \ppcf}\label{sect:ppcf}

%
\subsection{Types and Terms}
\begin{figure}
\centering
\[
\infer{\Delta,x:\typea\vdash x:\typea }{}
\qquad
\infer{\Delta\vdash\lambda x^{\typea} .\terma :\typea \rightarrow \typeb }{\Delta,x:\typea \vdash \terma :\typeb }
\qquad
\infer{\Delta\vdash (\terma  \termb ) :\typeb }{
  \Delta\vdash \terma  :\typea \rightarrow \typeb 
  &
  \Delta\vdash \termb  :\typea 
}
\qquad
\infer{\Delta\vdash (\fix \terma ):\typea }{\Delta\vdash \terma :\typea \rightarrow \typea }
\]

\[
\infer{\Delta\vdash \num r:\treal}{r\in\mathbb R}
\qquad
\infer{\Delta\vdash \realfun{f}(\terma _1,\dots, \terma _n) :\treal}
	{f \text{ meas. map } \real^n\rightarrow \real
	&
	\Delta\vdash \terma _i:\treal,
		\forall i\leq n
	}
\qquad
\infer{\Delta\vdash \ifz \termc  \terma  \termb : \treal }{
	\Delta\vdash \termc  :\treal
	&
	\Delta\vdash \terma  :\treal
	&
	\Delta\vdash \termb :\treal 
}
\]

\[
\infer{\Delta\vdash \oracle:\treal}
	{}
\qquad
\infer{\Delta\vdash \letterm x \terma  \termb : \treal }{
	\Delta\vdash \terma :\treal
	&
	\Delta,x:\treal\vdash \termb :\treal
	}
\]
\caption{The grammar of terms of \ppcf{} and their typing rules. The variable $x$ belongs to a fixed countable set of variables $\var$. The metavariable $\num f$ ranges over a fixed countable set $\Const$ of functional identifiers, while the metavariable $\num r$ may range on the whole $\Real$.}
\label{fig_types}
\end{figure}
We give in Figure~\ref{fig_types} the grammar of our probabilistic extension of PCF, briefly \ppcf, together with the typing rules. The types are generated by $\typea ,\typeb \ass\treal\;\vert\; \typea \rightarrow \typeb $, where the constant $\treal$ is the ground type for the set of real numbers. We denote by $\Terms[\Gamma\vdash \typea ]$ the set of terms typeable within the sequent $\Gamma\vdash \typea $. We write simply $\Terms$ if the typing sequent is not important or clear from the context.

The first line of Figure~\ref{fig_types}  contains the usual constructs of the simply typed $\lambda$-calculus extended to the fix-point combinator $\fix$ for any type $\typea $. The second line describes the primitives dealing with the ground type $\treal$. Our goal is to show the expressiveness of the category $\INFSCOTTM$ introduced in the next section, therefore \ppcf{} is an ideal language and does not deal with the issues about a realistic implementation of computations over real numbers. We refer the interested reader to e.g.~\citep{Vuillemin:1988,ESCARDO1996}.
We will suppose that the meta-variable $\realfun{f}$ ranges over a fixed countable set $\Const$ of basic measurable functions over real numbers. Examples of these functions include addition $+$, comparison $>$, and equality $=$; they are often written in infix notation. When clear from the context, we sometimes write $f$ for $\realfun f$. To be concise, we consider only the ground type $\treal$, the boolean operators (like $>$ or $=$) then evaluate to $1$ or $0$, representing resp.~ true and false.

The third line of Figure~\ref{fig_types} gives the ``probabilistic core'' of \ppcf. The constant $\oracle$ stands for the uniform distribution over $[0,1]$, i.e.\ the Lebesgue measure $\lambda_{[0,1]}$ restricted to the unit interval. The fact that \ppcf{} has only this distribution as a primitive is not limiting, in fact many other probabilistic measures (like binomial, geometric, gaussian or exponential distribution) can be defined from $\lambda_{[0,1]}$ and the other constructs of the language, see e.g.\ \citep{Park:2008} and Example~\ref{ex:gaussian}. The {\tt let} construction allows a call-by-value discipline over the ground type $\treal$: the execution of $\letterm x\terma \termb $ will sample a value (i.e.\ a real number $r$) from a probabilistic distribution $\terma $ and will pass it to $\termb $ by replacing \emph{every} free occurrence of $x$ in $\termb $ with $\num r$. This primitive\footnote{Notice that this primitive corresponds to the {\tt sample} construction in~\cite{Park:2008}.} is essential for the expressiveness of \ppcf{} and will be discussed both operationally and semantically in the next sections. 

\ppcf{} has a limited number of constructs, but it is known that many probabilistic primitives can be introduced as syntactic sugar from the ones in \ppcf, as shown in the following examples. We will prove the correctness of these encodings using the denotational semantics (Section~\ref{sect:sound_adeq}), this latter corresponding to the program operational behavior by the adequacy property (Theorem~\ref{th:adequacy}).

\begin{example}[Extended branching]\label{ex:if}
Let $U$ be a measurable set of real numbers whose characteristic function $\chi_U $ is in $\Const$, let $\termc\in\Terms[\Gamma\vdash\treal]$  and $\terma,\termb\in\Terms[\Gamma\vdash\typea]$ for $\typea=\typeb_1\rightarrow\dots\typeb_n\rightarrow\treal$. Then the term $\Gamma\vdash \ifterm \termc U\terma \termb:\typea$, branching between $\terma$ and $\termb$ according to the outcome of $\termc$ being in $U$, is a syntactic sugar for $\lambda x^{\typeb_1}\dots\lambda x^{\typeb_n}.\ifz{\realfun{\chi_U}(\termc )}{\termb x_1\dots x_n }{\terma x_1\dots x_n}$.\footnote{The swap between $\terma $ and $\termb $ is due to fact that $\mathtt{ifz}$ is the test to zero.} 
\end{example}
\begin{example}[Extended let]\label{ex:let}
  Similarly, the $\mathtt{let}$ constructor can be extended to any output type
  $\typea=\typeb_1\rightarrow\dots\typeb_n\rightarrow\treal$. Given
  $\terma\in\Terms[\Gamma\vdash\treal]$ and
  $\termb\in\Terms[\Gamma,x:\treal\vdash\typea]$, we denote by
  $\letterm x\terma\termb$ the term
  $\lambda x^{\typeb_1}\dots\lambda x^{\typeb_n}.\letterm x\terma{\termb
    x_1\dots x_n}$
  which is in $\Terms[\Gamma\vdash\typea]$. However we do not know in general
  how to extend the type of the bound variable $x$ to higher types in this model. The issue
  is clear at the denotational level, where the $\mathtt{let}$ construction is
  expressed with an integral (see Figure~\ref{fig:interpretation}). With each
  ground type, we associate a positive cone $\Bmes X$ which is generated by a
  measurable space $X$. At higher types, the associated cones do not have to be
  generated by measurable
  spaces.
  
  Notice that, because of this restriction on the type of the bound variable $x$, our $\mathtt{let}$ constructor does not allow to embed into our language the full call-by-value 
  PCF.
\end{example}


\begin{example}[Distributions]\label{ex:gaussian}

The Bernoulli distribution takes the value $1$ with some probability $p$ and the value $0$ with probability $1-p$. It can be expressed as the term $\bernoulli$ of type $\treal\to\treal$, taking the parameter $p$ as argument and testing whether $\oracle$ draws a value within the interval $[0,p]$, i.e.~$\lambda p.\letterm x\oracle{x\realfun{\le} p}$. 

The exponential distribution at rate $1$ is specified by its density $\mathrm e^{-x}$. It can be implemented as the term  $\exp$ of type $\treal$ by the inversion sampling method: $\letterm x\oracle{{\realfun{-\log}}(x)}$. 

The standard normal distribution (gaussian with mean $0$ and variance $1$) is
defined by its density $\tfrac1{\sqrt{2\pi}}{e^{-\tfrac 12x^2}}$. We use the
Box Muller method  to encode the normal distribution $\normal=\letterm
x{\oracle}{\letterm y{\oracle}{\realfun{(-2\log} (x)\realfun{)^{\tfrac 12}\, \cos(2\pi} y\realfun{)}}}$.

We can encode the Gaussian distribution as a function of the expected value $x$ and standard deviation $\sigma$ by $\gaussian =\lambda x\lambda\sigma\,\letterm y\normal{(\sigma y)\realfun{+}x}$.
\end{example}

\begin{example}[Conditioning]\label{ex:observe}
Let $U$ be a measurable set of real numbers such that $\chi_U\in\Const$, we
define a term $\observe U {}$ of type $\treal\rightarrow\treal$, taking a term
$\terma $ and returning the renormalization of the distribution of $\terma $ on the only samples that
satisfy $U$
: $\observe U {} = \lambda m.\fix(\lambda y. \letterm x m
{\ifterm xUxy})$. This corresponds to the usual way of implementing
conditioning by rejection sampling: the evaluation of $\observe U \terma $ will
sample a real $r$ from $\terma $, if $r\in U$ holds then the program returns
$\num r$, otherwise it iterates the procedure. Notice that $\observe U \terma $
makes a crucial use of sampling. The program $\lambda m.\fix(\lambda y.\ifterm
mUmy)$
 has a different behavior, because the two occurrences of $m$ correspond
in this case to two independent random variables (see Example~\ref{rk:letVSabs}
below).

%
\end{example}

\begin{example}[Monte Carlo Simulation]\label{ex:expectation}
An example using the possibility of performing independent copies of a random variable is the encoding of the $n$-th estimate of an expectation query. The expected value of a measurable function $f$ with respect to distribution $\mu$ is defined as $\int_\Real f(x)\mu(dx)$. The Monte Carlo method relies on the laws of large number: if $\bf x_1,\dots,\bf x_n$ are independent and identically distributed random variables of equal probability distribution $\mu$, then the $n$-th estimate $\frac{f(\bf x_1)+\dots+f(\bf x_n)}{n}$ converges almost surely to $\int_\Real f(x)\mu(dx)$. For any integer $n$, we can then encode the $n$-th estimate combinator by $\mathtt{expectation}_n=\lambda f.\lambda m.(\overbrace{f(m)+\dots+f(m)}^{n\text{ times}}/\num n)$ of type $(\treal\rightarrow\treal)\rightarrow\treal\rightarrow\treal$. Notice that it is crucial here that the variable $m$ has $n$ occurrences representing $n$ independent random variables, this being in contrast with  Example~\ref{ex:observe} (see also Example~\ref{rk:letVSabs}). 

\end{example}
%
\todom{Here commented metropolis hasting}
\subsection{Operational Semantics}\label{subsect:operational_semantics}

\begin{figure}
\begin{subfigure}{\linewidth}
\begin{align*}
(\lambda x.\terma )\termb &\rightarrow \terma \{\termb /x\}
&
\realfun{f}(\num{r_1},\dots,\num{r_n})&\rightarrow \num{f(r_1,\dots,r_n)}
\\[5pt]
\ifz{\num r}\terma \termb &\rightarrow 
	\begin{cases}
	\terma  &\text{if $r=0$,}\\
	\termb & \text{otherwise.}
	\end{cases}
&
\letterm{x}{\num r}{\terma }&\rightarrow \terma \{\num r/x\}
\\[5pt]
\fix \terma &\rightarrow \terma (\fix \terma )
&
\oracle&\rightarrow \num r \text{\qquad for any $r\in[0,1]$.}
\end{align*}
\caption{Reduction of a \ppcf{} redex.}\label{subfig:redex}

\medskip
\end{subfigure}
\begin{subfigure}{\linewidth}
\[
	E[\;]\ass [\;] \;\vert\; E[\;] \terma  \;\vert\; \ifz{E[\;]}\terma \termb  \;\vert\; \letterm{x}{E[\;]}{\termb }
	\;\vert\;\realfun f(\num r_1,\dots,\num r_{i-1}, E[\;],\terma_{i+1},\dots,\terma_{n})
\]

\[
	E[\terma ]\rightarrow E[\termb ]\text{, whenever } \terma  \rightarrow \termb 
\]
\caption{Grammar of the evaluation contexts and context closure of the reduction.}\label{subfig:evaluation_context}
\end{subfigure}
\caption{One-step operational semantics of \ppcf.}\label{fig:beta}
\end{figure}
The operational semantics of \ppcf{} is a Markov process defined starting from
the rewriting rules of Figure~\ref{fig:beta}, extending the standard
call-by-name reduction of PCF \citep{plotPCF}. 
The probabilistic primitive $\oracle$ draws a possible value from $[0,1]$, like in~\citep{Park:2008}.
The fact that we are sampling from the uniform
distribution and not from other distributions with equal support appears in the definition of the stochastic kernel $\Red$ (Equation~\eqref{eq:red}). 
In order to define this kernel, 
we equip $\Terms$ with a structure of measurable space (Equation~\eqref{eq:measurable_set_terms}). This defines a $\sigma$-algebra $\Sigma_{\Terms}$ of sets of terms equivalent to the one given in e.g.~\citep{BorgstromLGS16,StatonYWHK16} for slightly different languages. Similarly to~\citep{StatonYWHK16}, our definition is explicitly given by a countable coproduct of copies of the Lebesgue $\sigma$-algebra over $\real^n$ (for $n\in\mathbb N$, see Equations~\eqref{eq:sigma_algebra_terms}), while in~\citep{BorgstromLGS16} the definition is based on a notion of distance between $\lambda$-terms. The two definitions give the same measurable space, but the one adopted here allows to take advantage of the categorical structure of $\Kern$.

\begin{remark}\label{rk:sampling}
The operational semantics associates with a program $M$ a probabilistic distribution $\mathcal D$ of values describing the possible outcomes of the evaluation of $M$. There are actually two different``styles'' for giving $\mathcal D$: one based on samplings and another one, adopted here, based on stochastic kernels. \citet{BorgstromLGS16} proved that the two semantics are equivalent, giving the same distribution $\mathcal D$. 

The ``sampling semantics'' associates with $M$ a function mapping a trace of random samples to a weight, expressing the likelihood of getting that trace of samples from $M$. The final distribution $\mathcal D$ is then calculated by integrating this function over the space of the possible traces, equipped with a suitable measure. This approach is usually adopted when one wants to underline an implementation of the probabilistic primitives of the language via a sampling algorithm, e.g.~\cite{Park:2008}. 

The ``kernel-based semantics'' instead describes program evaluation as a discrete-time Markov process over a measurable space of states given by the set of programs ($(\Terms,\Sigma_{\Terms})$ in our case). The transition of the process is given by a stochastic kernel (here $\Red$ defined in Equation~\eqref{eq:red}) and then the probabilistic distribution $\mathcal D$ of values associated with a term is given by the supremum of the family of all finite iterations of the kernel ($\Red^\infty$, Equation~\eqref{eq:red_infty}).  This latter approach is more suitable when comparing the operational semantics with a denotational model (in order to prove soundness and adequacy for example) and it is then the one adopted  in this paper. 
\end{remark}

A \emph{redex} is a term in one of the forms at left-hand side of the $\rightarrow$ defined in Figure~\ref{subfig:redex}. A \emph{normal form} is a term $\terma $ which is no more reducible under $\rightarrow$. Notice that the closed normal forms of ground type $\treal$ are the real numerals
. The definition of the evaluation context (Figure~\ref{subfig:evaluation_context}) is the usual one defining the deterministic lazy call-by-name strategy: we do not reduce under an abstraction and there is always at most one redex to reduce, as stated by the following standard lemma. 
\begin{lemma}\label{lemma:evaluation_decomposition}
For any term $\terma $, either $\terma $ is a normal form or there exists a unique redex $R$ and an evaluation context $E[\;]$ such that $\terma =E[R]$.
\end{lemma}
It is standard to check that the property of subject reduction holds (if $\Gamma\vdash \terma :\typea $ and $\terma \rightarrow \termb $, then $\Gamma\vdash \termb :\typea $).

From now on, let us fix an enumeration without repetitions $(\varreal_i)_{i\in\mathbb N}$
of 
variables of type $\treal$.  Notice that any term
$\terma \in\Terms[\Gamma\vdash \typea ]$ with $n$ different occurrences of real
numerals, can be decomposed \emph{univocally} into a term
$\varreal_1:\treal, \dots, \varreal_n:\treal, \Gamma\vdash \fterma :\typea $ without real numerals
and a substitution
$\realsub=\{\num{r_1}/\varreal_1,\dots, \num{r_n}/\varreal_n\}$, such that: (i)
$\terma=\fterma\realsub$; (ii) each $\varreal_i$ occurs exactly once in
$\fterma$; (iii) $\varreal_i$ occurs before $\varreal_{i+1}$ reading the term
from left to right. Because of this latter condition, we can omit the name of
the substituted variables, writing simply $\fterma\vec r$ with
$\vec r=(r_1,\dots,r_n)$.
We denote by $\TermsFree[\Gamma\vdash \typea ]{n}$ the set of terms in
$\Terms[\varreal_1:\treal, \dots, \varreal_n:\treal,\Gamma\vdash \typea ]$
with no occurrence of numerals and respecting conditions (ii) and (iii) above. We let $\fterma, \ftermb$
vary over such real-numeral-free terms. 

Given $\fterma \in\TermsFree[\Gamma\vdash \typea ]{n}$ we then define the set $\Terms[\Gamma\vdash \typea ]_{\fterma } = \{\terma \in \Terms[\Gamma\vdash \typea ]\text{ s.t. }\exists \vec r\in\real^n, \terma =\fterma\vec r\}$. The bijection $s: \Terms[\Gamma\vdash \typea ]_{\fterma }\rightarrow \real^n$ given by $s(\fterma\vec r)=\vec r$ endows $\Terms[\Gamma\vdash \typea ]_{\fterma }$ with a $\sigma$-algebra  isomorphic to $\Sigma_{\real^n}$:
$U\in \Sigma_{\Terms[\Gamma\vdash \typea ]_{\fterma }}$  iff $s(U) \in \Sigma_{\real^n}$.
 The fact that $\TermsFree[\Gamma\vdash \typea ]{n}$ is countable and that $\Kern$ has countable coproducts (see  Section~\ref{section:compendium}), 
  allows us to define the measurable space of \ppcf{} terms of type $\Gamma\vdash \typea $ as the coproduct:
\begin{equation}\label{eq:sigma_algebra_terms}
(\Terms[\Gamma\vdash \typea ], \Sigma_{\Terms[\Gamma\vdash \typea ]}) 
= \coprod_{
	\substack{
		n\in\mathbb N, 
		\fterma \in \TermsFree[\Gamma\vdash \typea ]{n}				
	}
}
(\Terms[\Gamma\vdash \typea ]_{\fterma }, \Sigma_{\Terms[\Gamma\vdash \typea ]_{\fterma }})
\end{equation}
Spelling out the definition, a subset $U\subseteq \Terms[\Gamma\vdash \typea ]$ is measurable if and only if:
\begin{equation}\label{eq:measurable_set_terms}
	\forall n,\forall \fterma \in \TermsFree[\Gamma\vdash \typea ]{n}, \left\{\vec r\;\text{ s.t. }\; \fterma \vec r\in U\right\}\in\Sigma_{\real^n}
\end{equation}
Given a set $U\subseteq\Real$, we denote by $\num U$ the set of numerals associated with the real numbers in $U$. Of course $U$ is measurable iff $\num U$ is measurable.  The following lemma allows us to define $\Red$ and $\Red^\infty$.
\begin{lemma}\label{lemma:subst_measurable}
Given $\Gamma,x:\typeb\vdash\terma:\typea$ the function $\subst{x,\terma}$ mapping $\termb\in\Terms[\Gamma\vdash\typeb]$ to $\terma\{\termb/x\}\in\Terms[\Gamma\vdash\typea]$ is measurable. 
\end{lemma}

Given a term $\terma \in\Terms$ and a measurable set $U\subseteq\Terms$ we define $\Red(\terma ,U)\in[0,1]$ depending on the form of $\terma $, according to Lemma~\ref{lemma:evaluation_decomposition}:
\begin{equation}\label{eq:red}
\Red (\terma , U) = 
	\begin{cases}
	\delta_{E[\termb ]}(U)&\text{if $\terma =E[R]$, $R\rightarrow \termb $ and $R\neq\oracle$,}\\
	\lambda\{r\in[0,1]\text{ s.t. } E[\num r]\in U\}&\text{if $\terma =E[\oracle]$,}\\
	\delta_{\terma }(U)&\text{if $\terma $ normal form.}
	\end{cases}
\end{equation}
The last case sets the normal forms as accumulation points of $\Red$, so that $\Red(\terma ,U)$ gives the probability that we observe $U$ after \emph{at most one} reduction step applied to $\terma $. The definition in the case of $E[\oracle]$ specifies that $\oracle$ is drawing from the uniform distribution over $[0,1]$. Notice that, if $U\subseteq\real$ is measurable, then the set $\{r\in[0,1]\text{ s.t. } E[\num r]\in U\}$ is measurable by Lemma~\ref{lemma:subst_measurable}.
The definition of $\Red$ extends to a continuous setting the operational semantics Markov chain of~\citep{danosehrhard,EhrPagTas14}.

\begin{proposition}\label{prop:red_kernel}
For any sequent $\Gamma\vdash A$, the map $\Red$ is a stochastic kernel from $\Terms[\Gamma\vdash A]$ to $\Terms[\Gamma\vdash A]$.
\end{proposition}
\begin{proof}[Proof (Sketch).]
The fact that $\Red(\terma ,\_)$ is a measure is an immediate consequence of the definition of $\Red$ and the fact that any evaluation context $E[\;]$ defines a measurable map (Lemma~\ref{lemma:subst_measurable}).

Given a measurable set $U\subseteq\Terms[\Gamma\vdash A]$, we must prove that $\Red(\_,U)$ is a measurable function from $\Terms[\Gamma\vdash A]$ to $[0,1]$. Since $\Terms[\Gamma\vdash A]$ can be written as the coproduct in Equation~\eqref{eq:sigma_algebra_terms}, it is sufficient to prove that for any $n$ and $\fterma\in\TermsFree[\Gamma\vdash\typea]n$, $\Red_\fterma(\_,U):\Terms[\Gamma\vdash A]_\fterma\rightarrow[0,1]$ is a measurable function. One reasons by case study on the shape of $\fterma$, using Lemma~\ref{lemma:evaluation_decomposition} and the definition of a redex.
\end{proof}

We can then iterate $\Red$ using the composition of stochastic kernels (Equation~\eqref{eq:kern_composition}): 
\[
\Red^{n+1}(\terma ,U) = (\Red\circ\Red^{n})(\terma ,U) =\int_{\Terms} \Red(t,U)\,\Red^n(\terma , d t),
\]
this giving the probability that we observe $U$ after at most $n+1$ reduction
steps from $\terma $. Because  the normal forms are accumulation
points, one can prove by induction on $n$ that:
\begin{lemma}
Let $\Gamma\vdash\terma:\typea$ and let $U$ be a measurable set of normal forms in $\Terms[\Gamma\vdash\typea]$. The sequence $(\Red^n(\terma ,U))_n$ is monotone non-decreasing.
\end{lemma}
We can then define, for $\terma \in \Terms$ and $U$ a measurable set of normal forms, the limit
\begin{equation}\label{eq:red_infty}
	\Red^\infty(\terma ,U) = \sup_n\left( \Red^n(\terma ,U)\right).
\end{equation}
In particular, if $\terma $ is a closed term of ground type $\treal$, the only normal forms that $\terma $ can reach are numerals, in this case $\Red^\infty(\terma ,\_)$ corresponds to the probabilistic distribution over $\real$ which is computed by $\terma $ according to the operational semantics of \ppcf{} (Remark~\ref{rk:sampling}).

\begin{example}\label{rk:letVSabs}
In order to make clear the difference between a call-by-value and a call-by-name reduction in a probabilistic setting, let us consider the following two terms:
\begin{align*}
    \terma&= (\lambda x. (x = x))\oracle,&
    \termb&= \letterm x\oracle{x = x}.
\end{align*}
Both are closed terms of type $\treal$, ``applying'' the uniform distribution to the diagonal function $x \mapsto x = x$. However,  $\terma$ implements a call-by-name application, whose reduction duplicates the probabilistic primitive before sampling the distribution, while the evaluation of $\termb$ first samples a real number $\num r$ and then duplicates it:
  \begin{align*}
    \terma &\rightarrow \oracle =\oracle \rightarrow \num r=\num s&&\text{for any $r$ and $s$,}\\
    \termb &\rightarrow \letterm x {\num r} {x=x}\rightarrow \num{r}=\num r&& \text{for any $r$.}
  \end{align*}
  
The distribution associated with $\terma$ by $\Red^{\infty}$ is the Dirac $\delta_0$, because $\Red^{\infty}(\terma,\num U)=\Red^3(\terma,\num U)=\lambda(\{(r,r)\text{ s.t. } r\in [0,1]\})\times\delta_1(U)+  
\lambda(\{(r,s)\text{ s.t. } r\neq s, r,s\in
[0,1]\})\times\delta_0(U)=\delta_0(U)$, the last equality is because the
diagonal set $\{(r,r)\text{ s.t. } r\in [0,1]\}$ has Lebesgue measure
zero. This expresses that $\terma$ evaluates to $\num 0$ (i.e.\ ``false'') with
probability $1$, although there are an uncountable 
number of reduction paths reaching $\num 1$.  On the contrary, the distribution associated with $\termb$ is $\delta_1$: $\Red^{\infty}(\termb,\num U)=\Red^3(\terma,\num U)=\lambda([0,1])\times\delta_1(U)=\delta_1(U)$, expressing that $\termb$ always evaluates to $\num 1$ (i.e.\ true). 
\end{example}

\begin{remark}[Score]\label{rk:score}
  Some probabilistic programming languages have a primitive 
  \texttt{score} (e.g. \citep{BorgstromLGS16,StatonYWHK16}) or \texttt{factor}
  (e.g.~\cite{GoodmanT2014}), allowing 
  to express a probabilistic distribution from 
  a density function. A map $f$ is the probabilistic density function of a
  distribution $\mu$ with respect to another measure, say the Lebesgue measure
  $\lambda$, whenever $\mu(U)=\int_U f(x)\,\lambda(dx)$, for every measurable
  $U$. Intuitively, $f(x)$ gives a ``score'' expressing the likelihood of
  sampling the value $x$ from $\mu$. 

In our setting, the primitive $\mathtt{score}_x$ would be a term like
$\Gamma\vdash \mathtt{score}_x(\terma):\treal$, with
$\Gamma,x:\treal\vdash\terma:\treal$ defining $f$. 
The reduction of
$\mathtt{score}_x(\terma)$ outputs any numeral $\num r$ (a possible sample
from the distribution $\mu$), while the value $f(r)$ is used to multiply
$\Red$, like:
\begin{equation}\label{eq:score}
	\Red(\mathtt{score}_x(\terma), U)=\int_{\Real}\chi_{U}(r)f(r)\,\lambda(dr).
\end{equation}

This primitive allows to implement a distribution 
 in a more efficient way than rejection sampling, this latter based on a loop (Example~\ref{ex:gaussian}). However, $\mathtt{score}_x(\terma)$ suffers a major drawback: there is no static way of characterizing whether a term $\terma$ is implementing a probabilistic density function or rather a generic measurable map. The integral in  \eqref{eq:score} can have a value greater than one or even to be infinite or undefined for general $f$, in particular $\Red$ would fail to be a stochastic kernel for all terms.

This problem can be overcome by modifying the output type of a program, see
e.g.~\citep{StatonYWHK16}. We decided however to avoid these issues, convinced
that \ppcf{} is already expressive enough to test the
category $\INFSCOTTM$, which is the true object of study of this article. 
\end{remark}

%
\section{Cones}\label{sec:cones}
We now study the central semantical concept of this paper: cones and stable functions between cones. 
Before entering into technicalities, let us provide some intuitions and motivations.  A \emph{complete cone} $P$ is an $\Realp$-semimodule together with a norm $\Cnorm P{\underline\;}$ satisfying some natural axioms (Definition~\ref{def:cone}) and such that the unit ball $\Cuball P$ defined by the norm is complete with respect to the cone order $\leq_P$ (Definition~\ref{def:cone_order}). A type $\typea$ of $\ppcf$ will be associated with a cone $\semtype\typea$ and a closed program of type $\typea$ will be denoted as an element in the unit ball $\Cuball{\semtype\typea}$. The order completeness of $\Cuball{\semtype\typea}$ is crucial for defining the interpretation of the recursive programs (Section~\ref{sect:interp}), as usual.

There are various notions of cone in the literature and we are
following \citet{Selinger2004}, who uses cones similar to the
ones already presented in e.g. ~\cite{ando1962fundamental}. Let us
stress two of its crucial features.  (1) The cone order $\leq_P$ is
\emph{defined by} the algebraic structure and not given independently
from it --- this is in accordance with what happens in the category of
probabilistic coherence spaces \cite{danosehrhard}. 
(2) The
completeness of $\Cuball P$ is defined with respect to the cone order,
in the spirit of domain theory, rather than with respect to the norm,
as it is usual in the theory of Banach spaces.

A program taking inputs of type $\typea$ and giving outputs of type $\typeb$ will be denoted as a map from $\Cuball{\semtype\typea}$ to $\Cuball{\semtype\typeb}$. The goal of Section~\ref{sect:pre-stable} is to find the right properties enjoyed by such functions in order to get a cartesian closed category, namely that the set of these functions generates a complete cone compatible with the cartesian structure (which will be the denotation of the type $\typea\rightarrow\typeb$). It turns out that the usual notion of Scott continuity (Definition~\ref{def:linear_scott}) is too weak a condition for ensuring cartesian closeness (Section~\ref{sub:Scott_not_ccc}). A precise analysis of this point led us to the conclusion that these functions
have also to be \emph{absolutely monotonic} (Definition~\ref{def:Sn}). This latter condition is usually expressed by saying that all derivatives 
are everywhere non-negative, however we define it here as the non-negativity of iterated differences. Such non-differential definitions of absolute
monotonicity have already been considered various times in classical analysis, see for instance~\citep{McMillan54}. 

We call \emph{stable functions} the Scott continuous and absolutely monotonic functions (Definition~\ref{def:stable}), allowing for a cpo-enriched cartesian closed structure over the category of cones. The model of \ppcf{} needs however a further notion, that of measurability, which will be discussed in Section~\ref{sect:measurability}.

\begin{definition}\label{def:cone}
  A \emph{cone} $P$ is an $\Realp$-semimodule given together with an
  $\Realp$-valued function $\Cnorm P\_$ such that
  the following conditions hold for all $x,x',y,y'\in P$ and $\alpha\in\Realp$
\begin{center}
  \begin{minipage}{0.3\textwidth}
    \begin{itemize}
    \item $x+y=x+y'\Implies y=y'$
    \end{itemize}
  \end{minipage}
  \begin{minipage}{0.3\textwidth}
  \begin{itemize}
\item $\Cnorm P{\alpha x}=\alpha\Cnorm Px$
  \item $\Cnorm Px=0\Implies x=0$
  \end{itemize}
\end{minipage}
\begin{minipage}{0.38\textwidth}
  \begin{itemize}
\item $\Cnorm P{x+x'}\leq\Cnorm Px+\Cnorm P{x'}$
  \item $\Cnorm Px\leq\Cnorm P{x+x'}$. 
  \end{itemize}
\end{minipage}
\end{center}
  For $\alpha\in\Realp$ the set
  $\Cuball P(\alpha)=\{x\in\Ccarrier P\St\Cnorm Px\leq \alpha\}$ is called
  \emph{the ball of $P$ of radius $\alpha$}. The \emph{unit ball} is
  $\Cuball P=\Cuball P(1)$. A subset $S$ of $\Ccarrier P$ is \emph{bounded} if
  $S\subseteq\Cuball P(\alpha)$ for some $\alpha\in\Realp$.
\end{definition}
Observe that $\Cnorm P0=0$ by the second condition (homogeneity of the
norm) and that if $x+x'=0$ then $x'=x=0$ by the last condition
(monotonicity of the norm).

\begin{definition}\label{def:cone_order}
  Let $x,x'\in\Ccarrier P$, one writes $x\leq_P x'$ if there is a
  $y\in\Ccarrier P$ such that $x'=x+y$. This $y$ is then unique, and we set
  $x'-x=y$. The relation $\leq_P$ is easily seen to be an order relation on
  $\Ccarrier P$ and will be called \emph{the cone order relation} of $P$.
  
   A cone $P$ is \emph{complete} if any non-decreasing sequence
  $(x_n)_{n\in\Nat}$ of elements of $\Cuball P$ has a least upper bound
  $\sup_{n\in\Nat}x_n\in\Cuball P$.\todom{I moved def of completeness here, to save space and to make closer all definitions about the cones}
\end{definition}

The usual laws of calculus using subtraction hold (under the restriction that
all usages of subtraction must be well-defined). For instance, if
$x,y,z\in\Ccarrier P$ satisfy $z\leq_P y\leq_P x$ then we have
$x-z=(x-y)+(y-z)$. Indeed, it suffices to observe that $(x-y)+(y-z)+z=x$.



There are many examples of cones.
\begin{example}\label{ex:cone_ex}
  The prototypical example is $\Realp$ with the usual algebraic operations and  the norm given by $\Cnorm\Realp x=x$.
 %
  The cone $\ell_\infty$ is defined by taking as carrier set the set of all bounded elements of $\Realpto\Nat$, defining the algebraic laws
  pointwise, and equipping it with the norm
  $\Cnorm{} u=\sup_{n\in\Nat}u_n$.
  The cone $\ell_1$ instead is given by taking as carrier set the set of all
  elements $u$ of $\Realpto\Nat$ such that
  $\sum_{n=0}^\infty u_n<\infty$, defining the algebraic laws
  pointwise, and equipping it with the norm
  $\Cnorm{} u=\sum_{n=0}^\infty u_n<\infty$.
\end{example}

\begin{example}\label{ex:pcoh}
  Let $\cX=(\Web\cX,\Pcoh\cX)$ be a probablistic coherence space
  (see~\cite{danosehrhard}). Remember that this means that $\Web\cX$
  is a countable set (called web) and $\Pcoh\cX\subseteq\Realpto{\Web\cX}$ satisfies
  $\Pcoh\cX=\Biorth{\Pcoh\cX}$ (where, given
  $\cF\subseteq\Realpto{\Web\cX}$, the set
  $\Orth\cF\subseteq\Realpto{\Web\cX}$ is
  $\Orth{\cF}=\{u'\in\Realpto{\Web\cX}\St\forall u\in\cF\
  \sum_{a\in\Web\cX}u_au'_a\leq 1\}$)\footnote{There
    are actually two additional conditions which are not essential
    here.}.  Then we define a cone $\Cofpcoh\cX$ by setting
  $\Ccarrier{\Cofpcoh\cX}=\{u\in\Realpto{\Web\cX}\St\exists\epsilon>0\
  \epsilon u\in\Pcoh\cX\}$,
  defining algebraic operations in the usual componentwise way and
  setting
  $\Cnorm{\Cofpcoh\cX}u=\inf\{\alpha>0\St\frac 1\alpha u\in\Pcoh
  \cX\}=\sup\{\sum_{a\in\Web\cX}u_au'_a\St
  u'\in\Orth{\Pcoh\cX}\}$.
  
  The cones in Example~\ref{ex:cone_ex} are instances of this one.
\end{example}
\begin{example}\label{ex:non-complete-cone}
The set of all $u\in\Realpto\Nat$ such that $u_n=0$ for all but a finite
  number of indices $n$, is a cone $P_0$ when setting
  $\Cnorm{P_0}u=\sum_{n\in\Nat}u_n$.
\end{example}
\begin{example}\label{ex:measures-cone}
  Let $X$ be a measurable space. The set of all $\Realp$-valued
  measures\footnote{So we consider only ``bounded'' measures,
    which satisfy that the measure of the total space is finite, which
    is not the case of the Lebesgue measure on the whole $\Real$.} on $X$ is a
  cone $\Bmes X$, algebraic operations being defined in the
  usual ``pointwise'' way (\Eg~$(\mu+\nu)(U)=\mu(U)+\nu(U)$) and norm
  given by $\Cnorm{\Bmes X}\mu=\mu(X)$. This is the main motivating
  example for the present paper. Observe that such a cone is not of
  the shape $\Cofpcoh\cX$ in general.  
\end{example}
In all these examples, the cone order can be described in a pointwise
way. For instance, when $\cX$ is a probabilistic coherence space, one
has $u\leq_{\Cofpcoh\cX}v$ iff $\forall a\in\Web\cX\ u_a\leq
v_a$.
Similarly when $X$ is a measurable space, one has
$\mu\leq_{\Bmes\cX}\nu$ iff $\forall U\in\Tribu X\
\mu(U)\leq\nu(U)$.
This is due to the fact that when this condition holds, the function
$U\mapsto\nu(U)-\mu(U)$ is easily seen to be an $\Realp$-valued measure. 


All the examples above, but Example~\ref{ex:non-complete-cone}, are
examples of complete cones.

\begin{lemma}
  $P$ is complete iff any bounded non-decreasing sequence $(x_n)_{n\in\Nat}$
  has a least upper bound $\sup_{n\in\Nat}x_n$ which satisfies
  $\Cnorm P{\sup_{n\in\Nat}x_n}=\sup_{n\in\Nat}\Cnorm P{x_n}$.
\end{lemma}

\begin{definition}\label{def:bounded}
  Let $P$ and $Q$ be cones. A \emph{bounded map} from $P$ to $Q$ is a
  function $f:\Cuball P\to\Ccarrier Q$ such that $f(\Cuball P)\subseteq\Cuball
  Q(\alpha)$ for some $\alpha\in\Realp$; the greatest lower bound of these
  $\alpha$'s is called \emph{the norm of $f$} and is denoted as $\Fnorm f$.
\end{definition}

\begin{lemma}
  Let $f$ be a bounded map from $P$ to $Q$, then $\Fnorm f=\sup_{x\in\Cuball
    P}\Cnorm Q{f(x)}$ and $f(\Cuball P)\subseteq\Cuball Q(\Fnorm f)$.
\end{lemma}

\begin{definition}\label{def:linear_scott}
  A function $f:\Ccarrier P\to\Ccarrier Q$ is \emph{linear} if it commutes with sums
  and scalar multiplication.
  A \emph{Scott-continuous function} from
  a complete cone $P$ to a complete cone $Q$ is a bounded map\footnote{Remember that then $f:\Cuball P\to Q$, according with Definition~\ref{def:bounded}.} from $P$ to $Q$ which is non-decreasing and
  commutes with the lubs of non-decreasing sequences. 
\end{definition}

%

\begin{lemma}\label{lemma:add-Scott}
  Let $P$ be a complete cone. Addition is Scott-continuous
  $P\times P\to P$ and scalar multiplication is Scott-continuous
  $\Realp\times P\to P$.
\end{lemma}

Proofs are easy, see~\cite{Selinger2004}. The cartesian product of cones is 
defined in the obvious way (see Figure~\ref{fig:cartesian_product}).


\begin{definition}
Let $P$ be a cone and let $u\in\Cuball P$. We define a new cone
$\Crel Pu$ (\emph{the local cone of $P$ at $u$}) as follows. We set
$\Ccarrier{\Crel Pu}=\{x\in\Ccarrier P\St\exists\epsilon>0\ \epsilon
                       x+u\in\Cuball P\}$ and
\begin{align*}
  \Cnorm{\Crel Pu}x&=\inf\{1/\epsilon\St\epsilon >0\text{ and }\epsilon
                     x+u\in\Cuball P\}
                     =(\sup\{\epsilon\St\epsilon >0\text{ and }\epsilon
                     x+u\in\Cuball P\})^{-1}\,.
\end{align*}  
  
Given a sequence $\Vect u=(\List u1n)$ of elements of a cone
$\Ccarrier P$ s.t.~$u=\sum_{i=1}^nu_i\in\Cuball P$, we set
$\Crel P{\Vect u}=\Crel
Pu$.

\end{definition}

\begin{lemma}\label{lemma:local-cone}
  For any cone $P$ and any $u\in\Cuball P$, $\Crel Pu$ is a
  cone. Moreover
  $\Cuball{\Crel Pu}=\{x\in\Ccarrier P\St x+u\in\Cuball P\}$ and, for
  any $x\in\Ccarrier{\Crel Pu}$, one has
  $\Cnorm Px\leq\Cnorm{\Crel Pu}x$. If $P$ is complete then
  $\Crel Pu$ is complete.
\end{lemma}

\subsection{Pre-stable, aka.~Absolutely Monotonic, Functions}\label{sect:pre-stable}

We want now to introduce a notion of morphisms between cones such that
the resulting category will be cartesian closed. Given two cones $P$
and $Q$, a morphism from $P$ to $Q$ will be a Scott-continuous
function from $\Cuball P$ to $Q$ (because we need our morphisms to have least
fix-points in order to interpret general recursion) such that
$f(\Cuball P)\subseteq\Cuball Q$.  

\subsubsection{Failure of the straightforward attempt}\label{sub:Scott_not_ccc}
Is this Scott-continuity condition sufficient for guaranteeing cartesian
closeness? 
We argue that this not the case. Assume the opposite. 
Then it is easy to check that the cartesian product in our category is defined in the obvious way (algebraic operations defined pointwise, and supremum norm).

Given two cones $P$ and $Q$, we will need to define a new cone $R=(\Impl PQ)$
such that $\Cuball R$ will coincide with the set of morphisms from $P$ to $Q$
that is, under our assumption, of all Scott-continuous functions
$\Cuball P\to\Cuball Q$. In this 
cone $R$ (whose elements are all the
Scott-continuous functions $f:\Cuball P\to\Ccarrier Q$), the algebraic
operations are defined pointwise\footnote{Because the evaluation function
  should be linear in its functional argument, in accordance with the
  call-by-name evaluation strategy of our target programming language, see
  Lemma~\ref{lemma:linearity_eval}.}
and so the addition of $R$ induces the following order relation on
Scott-continuous functions: $f\leq g$ if $\forall x\in\Cuball P\ f(x)\leq g(x)$
and, moreover, the function $x\mapsto g(x)-f(x)$ is Scott-continuous.

Consider now the function $\Wpor:\Izu\times\Izu\to\Izu$ defined by
$\Wpor(s,t)=s+t-st=(1-s)t+s=(1-t)s+t$ and considered for instance
in~\cite{escardo_hofmann_streicher_2004}. It is clearly a Scott-continuous
function, so it is a morphism $\Realp\times\Realp\to\Realp$ in our category of
 cones and Scott-continuous functions. Let
$\Wpor':\Izu\to\Cuball{(\Impl\Realp\Realp)}$ be the
curryfied version of $\Wpor$ defined by
$\Wpor'(s)=\Wpor_s$ where $\Wpor_s$ is the Scott-continuous function
$\Izu\to\Izu$ defined by $\Wpor_s(t)=\Wpor(s,t)$. Then, by cartesian closeness,
$\Wpor'$ should be Scott-continuous and hence non-decreasing.
  But $\Wpor_0(t)=t$ and $\Wpor_1(t)=1$ and the function
$t\mapsto \Wpor_1(t)-\Wpor_0(t)=1-t$ is not non-decreasing. So we do not
have $\Wpor'(0)\leq \Wpor'(1)$ in the cone
$\Impl\Realp\Realp$ and our category is not cartesian closed.
\begin{quote}
  Our methodological principle is to stick to the cone order,
  natural wrt.~the algebraic structure, and adapt the notion
  of morphism so as to obtain a cartesian closed category.
\end{quote}
It turns out that this is perfectly possible and leads to an
interesting new notion of morphisms between cones, deeply
related with stability~\citep{stability} and the category $\mathbf{PCoh}_\oc$ of probabilistic coherence spaces already mentioned in the Introduction, as we will show in a further
paper. This connection with stability is already suggested by the
$\Wpor$ example: this function is a ``probabilistic version'' of the
well known parallel-or function~\citep{plotPCF}: we have
$\Wpor(1,0)=\Wpor(0,1)=1$ and $\Wpor(0,0)=0$. Stability has been
introduced for rejecting such functions.

\subsubsection{Absolutely monotonic functions} 
Pushing further the above analysis of the constraints imposed by cartesian
closeness on the monotonicity of morphisms, one arrives naturally to the
following definition.  Given a function $f:\Cuball P\to\Ccarrier Q$ which is
non-decreasing, we use the notation
\begin{align*}
  \Fdiff fxu=f(x+u)-f(x)
\end{align*}
for all $x\in\Cuball P$ and $u\in\Ccarrier P$ such that $x+u\in\Cuball P$. It
is clear that $\Fdiff fxu\in\Cuball Q$.
\begin{definition}\label{def:Sn}
  An \emph{$n$-non-decreasing} function from $P$ to $Q$ is a function
  $f:\Cuball P\to\Ccarrier Q$ such that
  \begin{itemize}
  \item $n=0$ and $f$ is non-decreasing
  \item or $n>0$, $f$ is non-decreasing and, for all $u\in\Cuball P$,
    the function $\Fdiff f\_u$ is $n-1$-non-decreasing from
    $\Crel Pu$ to $Q$.
  \end{itemize}
  One says that $f$ is \emph{absolutely monotonic} if it is $n$-non-decreasing
  for all $n\in\Nat$.
\end{definition}

\begin{example}\label{ex:pres_stable}
  Take $P=Q=\Realp$. A function $f:\Cuball P=\Izu\to Q=\Realp$ is
  $0$-non-decreasing if it is non-decreasing. It is $1$-non-decreasing
  if, moreover, for all $u\in\Izu$, the function
  $\Fdiff f\_u:[0,1-u]\to\Realp$ defined by $\Fdiff fxu=f(x+u)-f(x)$
  is non-decreasing. It is $2$-non-decreasing if moreover, for all
  $u_1,u_2\in\Izu$ such that $u_1+u_2\in\Izu$, the function
  $\Fdiff f\_{u_1,u_2}:[0,1-u_1-u_2]\to\Realp$ defined by
  $\Fdiff fx{u_1,u_2}=\Fdiff f{x+u_2}{u_1}-\Fdiff
  fx{u_1}=f(x+u_2+u_1)-f(x+u_2)-(f(x+u_1)-f(x))
  =f(x+u_1+u_2)-f(x+u_1)-f(x+u_2)+f(x)$
  is non-decreasing, etc. Typical examples of such
  $n$-non-decreasing functions for all $n$ are the polynomial functions
  with non-negative coefficients. 
\end{example}

\begin{example}
  To illustrate this definition further, consider the $\Wpor$ function
  introduced in Section~\ref{sub:Scott_not_ccc}. It is clearly
  $0$-non-decreasing. For $s,t,u,v\in\Realp$ such that $s+u,t+v\in\Izu$ we have
  $\Fdiff\Wpor{(s,t)}{(u,v)}=(s+u)+(t+v)-(s+u)(t+v)-(s+t-st)=u+v-sv-tu+st=(1-t)u+(1-s)v+st$. This
  function is not non-decreasing in $s$ and $t$, so $\Wpor$ is not
  $1$-non-decreasing.
\end{example}
\subsubsection{Pre-stable functions}
In most cases, Definition~\ref{def:Sn} is hard to manipulate because it is 
inductive and uses explicitly subtraction, an operation which is only
partially defined. We thus give an equivalent notion as follows. 

Let $n\in\Nat$, we use $\Cocard+n$ (resp.~$\Cocard-n$) for the set of
all subsets $I$ of $\{1,\dots,n\}$ such that $n-\Card I$ is even
(resp~odd). Given a map $f:\Cuball P\to\Ccarrier Q$,
$\Vect u\in\Ccarrier P^n$ such that $\sum_{i=1}^nu_i\in\Cuball P$ and
$x\in\Cuballr P{\Vect u}$, we define
\begin{align*}
  \Fdiffp\epsilon fx{\Vect u} &=
             \sum_{I\in\Cocard\epsilon n}f(x+\sum_{i\in I}u_i)\in\Ccarrier Q
\end{align*}
for $\epsilon\in\{+,-\}$. For instance
$\Fdiffp-fx{u_1,u_2,u_3}=f(x+u_1+u_2)+f(x+u_2+u_3)+f(x+u_1+u_3)+f(x)$ and
$\Fdiffp+fx{u_1,u_2,u_3}=f(x+u_1+u_2+u_3)+f(x+u_1)+f(x+u_2)+f(x+u_3)$.

Observe that when $n=0$, we have $\Fdiffp+ fx{}=f(x)$ and $\Fdiffp- fx{}=0$.

\begin{definition}\label{def:n-pre-stable}
  An \emph{$n$-pre-stable} function from $P$ to $Q$ is a function
  $f:\Cuball P\to\Ccarrier Q$ such that, for all
  $k\in\{1,\dots,n+1\}$, all $\Vect u\in\Ccarrier P^{k}$ such that
  $\sum_{i=1}^nu_i\in\Cuball P$, and all $x\in\Cuballr P{\Vect u}$, one
  has
    $\Fdiffp-fx{\Vect u}\leq\Fdiffp+fx{\Vect u}$.
  One says that $f$ is \emph{pre-stable} if it is $n$-pre-stable for all $n$.
\end{definition}
Observe $f$ is $0$-pre-stable iff for all $x\in\Cuball P$
and all $u\in\Ccarrier P$ such that $x+u\in\Cuball P$, one has $f(x)\leq
f(x+u)$, that is, $f$ is non-decreasing. 

By generalizing the computation in Example~\ref{ex:pres_stable}, one can prove by induction the following.
\begin{theorem}\label{th:non-decreasing-class-equiv}
  For all $n\in\Nat$, a function $f:\Cuball P\to\Ccarrier Q$ is
  $n$-non-decreasing iff it is $n$-pre-stable. Therefore, $f$ is absolutely
  monotonic iff it is pre-stable.
\end{theorem}


\begin{lemma}\label{lemma:fdiff-iter}
  Let $f$ be an absolutely monotonic function from $P$ to $Q$ (so
  that $f:\Cuball P\to\Ccarrier Q$). Let $n\in\Nat$,
  $\Vect u\in\Cuball P^n$
  with $\sum_{i=1}^nu_i\in\Cuball P $and $x\in\Cuballr P{\Vect
    u}$.
  Let $\List f0n$ be the functions defined by $f_0(x)=f(x)$ and
  $f_{i+1}(x)=\Fdiff{f_i}{x}{u_{i+1}}$. Then
  \begin{align*}
    f_n(x)=\Fdiffp+fx{\Vect u}-\Fdiffp-fx{\Vect u}\,.
  \end{align*}
  We set $\Fdiff fx{\Vect u}=f_n(x)$. The operation $\Delta$ is linear
  in the function:
  $\Fdiff{(\sum_{j=1}^p\alpha_jg_j)}{x}{\Vect
    u}=\sum_{j=1}^p\alpha_j\Fdiff{g_j}{x}{\Vect u}$
  for $\List g1p$ absolutely monotonic from $P$ to  $Q$.
\end{lemma}
As an immediate consequence we have that $\Fdiff fx{\List u1n}$ is symmetric in $\List u1n$, that is:   $\Fdiff fx{\List u1n}=\Fdiff fx{u_{\sigma(1)},\dots,u_{\sigma(n)}}$ for all
  permutation $\sigma$ on $\{1,\dots,n\}$.

\subsection{Composing Pre-stable Functions}
It is not completely obvious that pre-stable functions are closed under
composition (Theorem~\ref{th:pre-stable_compose}). The situation is a bit similar in categories of smooth
functions where composability derives from the chain rule.
Theorem~\ref{th:pre-stable_compose} is an immediate consequence of Lemma~\ref{lemma:fdiff-comp}, the proof of this latter needing the following auxiliary lemmas.

\begin{lemma}\label{lemma:fdiff-iter-pre-stable}
  Let $f:\Cuball P\to\Ccarrier Q$ be a pre-stable function from $P$ to $Q$.
  For all $\Vect u\in\Cuball P^n$, the functions
  $\Fdiffp- f\_{\Vect u}$, $\Fdiffp+ f\_{\Vect u}$ and
  $\Fdiff f\_{\Vect u}$ are pre-stable from $\Crel P{\Vect u}$ to $Q$.
\end{lemma}

\begin{lemma}\label{lemma:fdiff-sommes1}
  Let $f:\Cuball P\to\Ccarrier Q$ be a pre-stable function from $P$ to $Q$. Let
  $n\in\Nat$, $x,u,v\in\Cuball P$ and $\Vect u\in\Cuball P^n$, and
  assume that $x+u+v+\sum_{i=1}^nu_i\in\Cuball P$. Then
\begin{align*}
  \Fdiff f{x+u}{\Vect u}&=\Fdiff fx{\Vect u}+\Fdiff fx{u,\Vect u},&
  \Fdiff fx{u+v,\Vect u}&=\Fdiff fx{u,\Vect u}+\Fdiff f{x+u}{v,\Vect u}\,.
\end{align*}
\end{lemma}

\begin{lemma}\label{lemma:fdiff-sommes}
  Let $f:\Cuball P\to\Ccarrier Q$ be a function which is pre-stable from
  $P$ to $Q$. Let $n\in\Nat$, $x,u\in\Cuball P$ and
  $\Vect u,\Vect v\in\Cuball P^n$, and assume that
  $x+u+\sum_{i=1}^n(u_i+v_i)\in\Cuball P$. Then
  \begin{align*}
    \Fdiff f{x+u}{\Vect u+\Vect v}
    &= \Fdiff fx{\Vect u}+\Fdiff fx{u,\Vect u+\Vect v}
    +\Fdiff f{x+u_1}{v_1,u_2+v_2,\dots,u_n+v_n}\\
    &+\Fdiff f{x+u_2}{u_1,v_2,u_3+v_3,\dots,u_n+v_n}
    +\cdots+\Fdiff f{x+u_n}{u_1,\dots,u_{n-1},v_n}\,.
  \end{align*}
\end{lemma}
\begin{proof}
Simple computation using Lemma~\ref{lemma:fdiff-sommes1}.
\end{proof}

Let $P$ be a cone and let $p\in\Nat$. Let
$\Scone pP=(\Ccarrier P^{p+1},\Cnorm{\Scone pP}\_)$ where
$\Cnorm{\Scone pP}{(x,u_1,\dots,u_{p})}=\Cnorm P{x+\sum_{i=1}^pu_i}$. Then,
with algebraic laws defined componentwise, it is easy to check that $\Scone pP$
is a cone which is complete if $P$ is.

\begin{lemma}\label{lemma:fdiff-pre-stable-gen}
  Let $f:\Cuball P\to\Ccarrier Q$ be a pre-stable function from $P$ to
  $Q$. Then the map $g:\Cuball{\Scone pP}\to\Ccarrier Q$ defined by
  $g(x,\Vect u)=\Fdiff fx{\Vect u}$ is non-decreasing, for all
  $p\in\Natnz$.
\end{lemma}

\begin{lemma}\label{lemma:inf-pre-stable-sums}
  Let $f,g:\Cuball P\to\Ccarrier Q$ be two pre-stable functions from $P$ to $Q$.
  The function $f+g$ (sum defined pointwise) is pre-stable.
\end{lemma}

\begin{lemma}\label{lemma:fdiff-comp}
  Let $p\in\Nat$, $f,\List h1p:\Cuball P\to\Ccarrier Q$ be
  pre-stable functions from $P$ to $Q$ and $g:\Cuball Q\to\Ccarrier R$ be
  pre-stable functions from $Q$ to $R$ such that
  $\forall x\in\Cuball P\ f(x)+\sum_{i=1}^ph_i(x)\in\Cuball Q$. Then the
  function $k:\Cuball P\to\Ccarrier R$ defined by
$
    k(x)=\Fdiff g{f(x)}{h_1(x),\dots,h_p(x)}
 $
  is pre-stable from $P$ to $R$.
\end{lemma}
\begin{proof}
Observe that our hypotheses imply that, for all $x\in\Cuball P$, one has
$f(x)\in\Cuball{\Crel Q{h_1(x),\dots,h_n(x)}}$

With the notations and conventions of the statement, we prove by induction on
$n$ that, for all $n\in\Nat$, for all $p\in\Nat$, for all $f,\List h1p,g$
which are pre-stable and satisfy
$\forall x\in\Cuball P\ f(x)+\sum_{i=1}^ph_i(x)\in\Cuball Q$, the function $k$
is $n$-pre-stable.

For $n=0$, the property results from Lemma~\ref{lemma:fdiff-pre-stable-gen}.

We assume the property for $n$ and prove it for $n+1$. Let $u\in\Cuball P$ we
have to prove that the function $\Fdiff k\_u$ is $n$-pre-stable from $\Crel Pu$
to $R$. Let $x\in\Cuballr Pu$, we have
\begin{align*}
  \Fdiff kxu&=\Fdiff g{f(x+u)}{h_1(x+u),\dots,h_p(x+u)}-
              \Fdiff g{f(x)}{h_1(x),\dots,h_p(x)}\\
            &=\Fdiff g{f(x)+\Fdiff fxu}{h_1(x)+\Fdiff{h_1}{x}{u},
              \dots,h_p(x)+\Fdiff{h_1}{x}{u}}\\
            &\quad-\Fdiff g{f(x)}{h_1(x),\dots,h_p(x)}\\
            &=\Fdiff g{f(x)}{\Fdiff fxu,h_1(x)+\Fdiff{h_1}{x}{u},
              \dots,h_p(x)+\Fdiff{h_p}{x}{u}}\\
            &\quad+\Fdiff g{f(x)+h_1(x)}{\Fdiff{h_1}{x}{u},
              h_2(x)+\Fdiff{h_2}xu,\dots,h_p(x)+\Fdiff{h_p}xu}\\
            &\quad+\Fdiff g{f(x)+h_2(x)}{h_1(x),\Fdiff{h_2}xu,
              h_3(x)+\Fdiff{h_3}xu,\dots,h_p(x)+\Fdiff{h_p}xu}+\cdots\\
            &\quad+\Fdiff g{f(x)+h_p(x)}{h_1(x),
              \dots,h_{p-1}(x),\Fdiff{h_p}{x}{u}}
\end{align*}
by Lemma~\ref{lemma:fdiff-sommes}. We can apply the inductive hypothesis to
each of the terms of this sum. Let us consider for instance the first
of these expressions. We know that the functions
$h'_1,\dots,h'_{p+1}$ defined by $h'_1(x)=\Fdiff fxu$,
$h'_2(x)=h_1(x)+\Fdiff{h_1}{x}{u}=h_1(x+u)$,\dots,
$h'_{p+1}(x)=h_p(x)+\Fdiff{h_p}{x}{u}=h_p(x+u)$ are pre-stable from
$\Crel Pu$ to $Q$: this results from
Lemmas~\ref{lemma:inf-pre-stable-sums}
and~\ref{lemma:fdiff-iter-pre-stable}. Moreover we have
$\forall x\in\Cuball P\
f(x)+\sum_{i=1}^{p+1}h'_i(x)=f(x+u)+\sum_{i=1}^ph_i(x+u)\in\Cuball Q$.
Therefore the inductive hypothesis applies and we know that the function
$x\mapsto\Fdiff g{f(x)}{\Fdiff fxu,h_1(x)+\Fdiff{h_1}{x}{u},
  \dots,h_p(x)+\Fdiff{h_p}{x}{u}}$
is $n$-pre-stable. The same reasoning applies to all terms and, by
Lemma~\ref{lemma:inf-pre-stable-sums}, we know that the function $\Fdiff k\_u$ is
$n$-pre-stable from $\Crel Pu$ to $Q$.
\end{proof}

\begin{theorem}\label{th:pre-stable_compose}
  Let $f$ be a pre-stable function from $P$ to $Q$ and $g$ be a pre-stable
  function from $Q$ to $R$. If $f(\Cuball P)\subseteq\Cuball Q$ then
  $g\Comp f$ is a pre-stable function from $P$ to $R$.
\end{theorem}
\begin{proof}
This is the case $p=0$ of Lemma~\ref{lemma:fdiff-comp}.  
\end{proof}

\begin{definition}\label{def:stable}
  A \emph{stable} function from $P$ to $Q$ is a pre-stable (or,
  equivalently, an absolutely monotonic) function from $P$ to $Q$
  which is Scott continuous.  We use $\INFSCOTT$ for the category of
  complete cones and stable functions. More explicitly,
  $\INFSCOTT(P,Q)$ is the set of all functions
  $f:\Cuball P\to\Ccarrier Q$ which are pre-stable from $P$ to $Q$,
  Scott-continuous and satisfy $f(\Cuball P)\subseteq\Cuball Q$.
\end{definition}

\section{Measurability}\label{sect:measurability}
The cone $\Bmes \Real$ of $\Realp$-valued
measures on $\Real$ (Example~\ref{ex:measures-cone}) is the natural candidate to model the ground type $\treal$. 
In particular, a real numeral will be interpreted as the Dirac measure $\Dirac r$. Consider now a closed term 
$\letterm x \terma  \termb$ of type $\treal$, so that $\vdash\terma:\treal$ and $x:\treal\vdash\termb:\treal$. The term $\terma$ 
will be associated with a measure $\mu$ in $\Cuball(\Bmes \Real)$, while $\termb$ will be a stable 
function $f$ from the whole $\Cuball(\Bmes \Real)$ to $\Cuball(\Bmes \Real)$. However, according to the operational semantics (Figure~\ref{fig:beta}), $\termb$ is supposed to
get a real number $r$ for $x$, and not a generic measure. This means that one has to compose $f$ with a map $\Diracf:\Real\to\Cuball(\Bmes\Realp)$ defined by
$\Diracf(r)=\Dirac r$, so that $f\Comp \Diracf:\Real\to\Cuball(\Bmes \Real)$. Now, the natural way to pass $\mu$ to $f\Comp \Diracf$ is then by the integral $\int_\Real (f\Comp \Diracf)(r)\mu(dr)$. However, this would be meaningful only in case $f\Comp \Diracf$ is measurable, and this is not the true of all stable functions $f$.\footnote{Indeed, by Lebesgue decomposition theorem
we can write $\Bmes \Real=\cM_0\oplus\cM_1$, a co-product of
cones (it is easily checked that the category of complete
   cones of linear and Scott-continuous functions has
  co-products), where the elements of $\cM_0$ are the discrete
measures, that is, the countable linear combinations of Dirac measures
$\sum_{i=1}^\infty \alpha_n\Dirac{r_n}$ with
$\forall n\ \alpha_n\in\Realp$ and $\sum_n\alpha_n<\infty$, and
$\cM_1$ is the cone of measures $\mu$ such that $\mu(\{r\})=0$ for all
$r\in\Real$. Let $U\subseteq\Real$ be a non-measurable set and let
$f:\cM\to\Realp$ be the linear (and hence pre-stable) and Scott-continuous
function defined on this co-product, by: $f(\mu)=0$ if $\mu\in\cM_1$
and $f(\Dirac r)=\Mchar U(r)$. Then $f\Comp\Diracf=\Mchar U$ is
not measurable. We thank Jean-Louis Krivine for this example.
}

We have then to slightly refine our model, endowing our cones
with a structure allowing to formulate a convenient measurability
property for our morphisms. This is the goal of this section.

\subsection{Measurability Tests}
If $P$ is a cone, we use $\Dual P$ for the topological dual of $P$,
which is the set of all functions $l:P\to\Realp$ which are linear
(that is, commute with linear combinations) and Scott-continuous. Such
a function, when restricted to $\Cuball P$, clearly defines a
stable function from $P$ to $\Realp$.

\begin{definition}\label{def:meas-tests}
We consider cones $P$ equipped with a collection $(\Mtest nP)_{n\in\Nat}$ where
$\Mtest nP\subseteq(\Dual P)^{\Real^n}$ satisfies the following properties.
\begin{itemize}
\item $0\in\Mtest nP$
\item if $l\in\Mtest nP$ and $h:\Real^p\to\Real^n$ is measurable then $l\Comp
  h\in\Mtest pP$
\item and for any $l\in\Mtest nP$ any $x\in P$, the function
  $\Real^n\to\Realp$ which maps $\Vect r$ to $l(\Vect r)(x)$ is in
  $\Mfun{n}$, \IE~is a measurable map $\Real^{n}\to\Realp$.

\end{itemize}  
A cone $P$ equipped with a family $(\Mtest nP)_{n\in\Nat}$ satisfying
the above conditions will be called a \emph{measurable cone}.
The elements of the sets $\Mtest nP$ will be called the \emph{measurablility
  tests} of $P$. 
\end{definition}
Measurability tests have parameters in $\Real^n$ and
are not simply Scott-continuous linear forms for making it possible to prove
that the evaluation function of the cartesian closed structure is well
behaved. This will be explained in Remark~\ref{rk:answ-mtest-real-param}.

\begin{definition}\label{def:measurable_path}
  Let $P$ be a cone and let $n\in\Nat$. A \emph{measurable path} of arity $n$ in
  $P$ is a function $\gamma:\Real^n\to P$ such that
  \begin{itemize}
  \item $\gamma(\Real^n)$ is bounded in $P$
  \item and, for all $k\in\Nat$ and all $l\in\Mtest kP$, the function
    $\Mtestapp l\gamma:\Real^{k+n}\to\Realp$ defined by
    $(\Mtestapp l\gamma)(\Vect r,\Vect s)=l(\Vect r)(\gamma(\Vect s))$ is in $\Mfun{k+n}$, \IE~is a measurable map $\Real^{k+n}\to\Realp$.
  \end{itemize}
  We use $\Mpath nP$ for the set of measurable paths of $P$ and
  $\Mpathu nP$ for the set of measurable paths which take their values in
  $\Cuball P$. 
\end{definition}

\begin{lemma}\label{lemma:mes-path-cst}
  For any $x\in P$ and $n\in\Nat$, the function $\gamma:\Real^n\to P$
  defined by $\gamma(\Vect r)=x$ belongs to $\Mpath nP$.  If
  $\gamma\in\Mpath nP$ and $h:\Real^p\to\Real^n$ is measurable then
  $\gamma\circ h\in\Mpath pP$.
\end{lemma}

\begin{example}\label{ex:meas_cone_with_meas}
  Let $X$ be a measurable space. We equip the cone $\Bmes X$ with the
  following notion of measurability tests. For each $n\in\Nat$, we
  set $\Mtest n{\Bmes X}=\{\Mapp U\St U\in\Tribu X\}$ where
  $\Mapp U(\Vect r)(\mu)=\mu(U)$. Observe that indeed $\Mapp U$ is
  linear and Scott-continuous, see Example~\ref{ex:measures-cone} and
  the observation that, in the complete cone $\Bmes X$, lubs are
  computed pointwise: given a non-decreasing and bounded sequence
  $\mu_n$ of elements of $\Bmes X$, one has
  $(\sup_{n\in\Nat}\mu_n)(U)=\sup_{n\in\Nat}\mu_n(U)$. Therefore an
  element of $\Mpathu n{\Bmes X}$ is a stochastic 
  kernel from
  $\Real^n$ to $X$. This 
  example justifies our terminology of
  ``measurable cone'' because, in $\Bmes X$, the measurable
  tests coincide with the measurable sets of $X$.
\end{example}

\begin{definition}\label{def:mes-pre-stable-function}
  Let $P$ and $Q$ be measurable complete cones.  A stable function
  from $P$ to $Q$ (remember that then $f$ is actually a function
  $\Cuball P\to Q$) is \emph{measurable} if, for all $n\in\Nat$ and
  all $\gamma\in\Mpathu nP$, one has $f\Comp\gamma\in\Mpath nQ$.  We
  use $\INFSCOTTM$ for the subcategory of $\INFSCOTT$ whose morphisms
  are measurable.
\end{definition}
This definition makes sense because if $f:\Cuball P\to\Cuball Q$ and
$g:\Cuball Q\to\Cuball R$ are stable and measurable, then $g\Comp f$
has the same properties. 

\section{The cartesian closed structure of {$\INFSCOTTM$}}\label{sect:ccc}
\begin{figure}
\centering
\begin{subfigure}{\linewidth}
\begin{align*}
\prod_{i\in I} P_i&=\{(x_i)_{i\in I}\text{ s.t. } \forall i\in I, x_i\in P_i\},
&
\Cnorm{\prod_{i\in I} P_i}{(x_i)_{i\in I}}&=\sup_{i\in I}\Cnorm{P_i}{x_i}
\\
\Mtest n{\prod_{i\in I}P_i}&=\{\bigoplus_{i\in I}l_i\St\forall i\in I\ l_i\in\Mtest n{P_i}\},
&
\hspace{-.6cm}\text{with } \bigoplus_{i\in I}l_i(\vec r)(x_i)_{n\in I}&=\sum_{i\in I}l_i(\vec r)(x_i)&
\end{align*}
\caption{finite cartesian product ($I$ finite set). We can simply write $P_1\times P_2$ for the binary product.}\label{fig:cartesian_product}
\end{subfigure}

\begin{subfigure}{\linewidth}
{
\begin{align*}
\Implm PQ&=\{ f: \Cuball P\to Q \St \exists\epsilon>0, \epsilon f\in \INFSCOTTM(P,Q)\},
&
\Cnorm{\Implm PQ}{f}&=\sup_{x\in\Cuball P}\Cnorm{Q}{f(x)}
\\
\Mtest n{\Implm PQ}&=\{\Mtestfun{\gamma}{m}\St \gamma\in\Mpath n{P}, m\in \Mtest n Q\},
&
\hspace{-1.5cm}\text{with } \Mtestfun{\gamma}{m}(\vec r)(f)& = m(\vec r)(f(\gamma(\vec r)))
\end{align*}
}
\caption{object of morphisms}\label{fig:exponential_object}
\end{subfigure}
\caption{The CCC structure of $\INFSCOTTM$. The projections, pairing and the evaluation are defined as standard.}\label{fig:ccc}
\end{figure}

\subsection{Cartesian Product}\label{sec:pre-stable-prduct}

The \emph{cartesian product $P=\prod_{i\in I}(P_i)$} of a finite\footnote{We could easily define countable products, this will be done in an extended version of this paper.} family of cones
  $(P_i)_{i\in I}$ is given in Figure~\ref{fig:cartesian_product}, where
  addition and  scalar multiplication are defined componentwise. It is clear that we have defined in that way a complete cone and that
$\Cuball P=\prod_{i\in I}\Cuball{P_i}$. Given a non-decreasing sequence $(x(p))_{p\in\Nat}$ in $\Cuball P$, then $x=\sup_{p\in\Nat}x(p)$ is characterized by $x_i=\sup_{p\in\Nat} x(p)_i$ (this lub being taken in $\Cuball{P_i}$).


The projections $\Proj i:\Cuball P\to\Cuball{P_i}$ are easily seen to
be linear and Scott-continuous and hence stable $P\to P_i$.
Let $f_i\in\INFSCOTTM(Q,P_i)$ for each
$i\in I$. We define a function $f:\Cuball Q\to\Cuball P$ by
$f(y)=(f_i(y))_{i\in I}$. 
It is straightforward to check that this function is stable: 
pre-stability follows from $\Fdiff fyv=(\Fdiff{f_i}yv)_{i\in I}$.

\begin{lemma}\label{lemma:sep-lin-cont}
  Let $f:\Ccarrierb P\times\Cuball Q\to \Ccarrierb R$ be a function such that
  \begin{itemize}
  \item for each $y\in\Cuball Q$, the function
    $f^{(1)}_y:\Ccarrierb P\to\Ccarrierb R$
    defined by $f^{(1)}_y(x)=f(x,y)$ is
    linear (resp.~linear and Scott-continuous)
  \item and for each $x\in\Ccarrierb P$, the function
    $f^{(2)}_x:\Cuball Q\to\Ccarrierb R$ defined by $f^{(2)}_x(y)=f(x,y)$ is
    pre-stable (resp.~pre-stable and Scott-continuous).
  \end{itemize}
  Then the restriction $f:\Cuball P\times\Cuball Q\to\Ccarrierb R$ is
  pre-stable (resp.~pre-stable and Scott-continuous, that is, stable)
  from $P\times Q$ to $R$.
\end{lemma}
In the second line of Figure~\ref{fig:cartesian_product} we endow the cone $\prod_{i\in I}P_i$
with a notion of measurability tests $\Mtest n{\prod_{i\in I}P_i}$ for any $n\in\Nat$. It is
obvious that this notion satisfies the conditions of Definition~\ref{def:meas-tests}. The fact that $(\bigoplus_{i\in I}l_i)(\Vect r)$ is Scott-continuous results from the Scott-continuity of addition (Lemma~\ref{lemma:add-Scott}).  Moreover, for $x\in\prod_{i\in I} P_i$, the map
$\Vect r\mapsto\sum_{i\in I}l_i(\Vect r)(x_i)$ is measurable as a sum of
measurable functions. Given $h:\Real^p\to\Real^n$ measurable, we have
$(\bigoplus_{i\in I}l_i)\Comp h=\bigoplus_{i\in I}(l_i\Comp h)$ and hence the
measurability tests of $\prod_{i\in I}P_i$ are closed under precomposition with
measurable maps.

\begin{lemma}\label{lemma:mes-path-tuple}
  For any $n\in\Nat$ we have
$
    \Mpath n{\prod_{i\in I}P_i}=
    \{\Tuple{\gamma_i}_{i\in I}\St\forall i\in I\ \gamma_i\in\Mpath n{P_i}\}. 
$
\end{lemma}

\begin{theorem}
  The category $\INFSCOTTM$ is cartesian with projections
  $(\Proj i)_{i\in I}$ and tupling defined as in the category of sets
  and functions. The terminal object is the cone $\{0\}$ with
  measurability tests equal to $0$.
\end{theorem}
It suffices to prove that the usual projections are measurable and
that the tupling of measurable stable functions is
measurable, this is straightforward. 

\subsection{Function Space}\label{sec:pre-stable-fun-space}
Let $P$ and $Q$ be two measurable complete cones. We define
$\Ccarrierb{\Implm PQ}$ in Figure~\ref{fig:exponential_object} as the
set of all measurable stable functions from $P$ to $Q$, that is, the
set of all functions $f:\Cuball P\to\Ccarrierb Q$ such that there
exists $\epsilon>0$ such that $\epsilon f\in\INFSCOTTM(P,Q)$. It is
clear that $\Ccarrierb{\Implm PQ}$ is closed under pointwise addition
of functions and pointwise scalar multiplication: this results from
the fact that measurability tests are (parameterized) linear functions
and that measurable functions $\Real^n\to\Realp$ have the same closure
properties. We still have to check that, with the norm
$\Cnorm{\Implm PQ}\_$, $\Implm PQ$ is a complete cone
(Lemma~\ref{lemma:function_cone_complete}).



For this purpose, the next lemma which provides a characterization of the
order relation in the function space similar to Berry's stable
order~\cite{stability} will be essential.

\begin{lemma}\label{lemma:fun-order}
  Let $f,g\in\Ccarrierb{\Implm PQ}$. We have $f\leq g$ in $\Implm PQ$
  iff the following condition holds
  \begin{align*}
    \forall n\in\Nat, \forall\Vect u\in P^n\
    \sum_{i=1}^nu_i\in\Cuball P, \forall
    x\in\Cuballr P{\Vect u}, \Fdiffp+fx{\Vect u}+\Fdiffp-gx{\Vect
    u}\leq\Fdiffp+gx{\Vect u}+\Fdiffp-fx{\Vect u}\,.
  \end{align*}
\end{lemma}

\begin{lemma}\label{lemma:function_cone_complete}
  The cone $\Implm PQ$ is complete and the lubs in
  $\Cuballp{\Implm PQ}$ are computed pointwise.
\end{lemma}

The second line of Figure~\ref{fig:exponential_object} defines a family
$(\Mtest n{\Implm PQ})_{n\in\Nat}$ of sets of measurability tests. Also in this case the conditions of Definition~\ref{def:meas-tests} are respected. The fact that indeed $(\Mtestfun\gamma m)(\Vect r)\in\Dual{(\Implm PQ)}$ clearly follows from the definition of the cone $\Implm PQ$ and from the fact that lubs in that cone are computed pointwise. Since $\Mpathu nP$ is non-empty (it contains at least the $0$-valued constant path $\zeta$) and $0\in\Mtest nQ$, we have $0=\Mtestfun\zeta 0\in\Mtest n{\Implm PQ}$. Let $\gamma\in\Mpathu nP$,
$m\in\Mtest nQ$ and let $h:\Real^p\to\Real^n$ be measurable. We have $(\Mtestfun\gamma m)\Comp h=\Mtestfun{(\gamma\Comp h)}{(m\Comp h)}$ and we know that $\gamma\Comp h\in\Mpathu pP$ (by Lemma~\ref{lemma:mes-path-cst}) and $m\Comp h\in\Mtest nQ$ (by assumption about $Q$) so $(\Mtestfun\gamma m)\Comp h\in\Mtest p{\Implm PQ}$. Last, with the same notations, let $f\in\Implm PQ$. Then the map $\Vect r\mapsto(\Mtestfun\gamma m)(\Vect r)(f)$ is measurable by definition of stable measurable functions. So
we have equipped $\Implm PQ$ with a collection of measurability tests.

Given
$(f,x)\in\Cuballp{(\Implm PQ)\times P} =\Cuballp{\Implm PQ}\times\Cuball
P$
we set $\Ev(f,x)=f(x)\in\Cuball Q$. It is clear that this function is
non-decreasing and Scott-continuous (because lubs of non-decreasing sequences of
functions are computed pointwise).

\begin{lemma}\label{lemma:stab-eval}
  The evaluation map $\Ev$ is stable and measurable, i.e.~ $\Ev\in\INFSCOTTM((\Implm PQ)\times P,Q)$.
\end{lemma}
\begin{proof}[Proof (Sketch)]
  Stability results from
  Lemma~\ref{lemma:sep-lin-cont}, observing that $\Ev$ is linear and
  Scott-continuous in its first argument.
  The proof that $\Ev$ is measurable follows from a simple computation.
\end{proof}

\begin{remark}\label{rk:answ-mtest-real-param}
  The proof of Lemma~\ref{lemma:stab-eval} strongly uses the fact that
  our measurability tests have parameters in $\Real^n$, see
  Definition~\ref{def:meas-tests}.  If the measurability tests were
  just Scott-continuous linear forms (without real parameters), we
  would define\footnote{This is certainly the most natural definition
    in this simplified setting.} the measurability tests of
  $\Implm PQ$ as the $\Mtestfun xm\in\Dual{(\Implm PQ)}$ where
  $x\in\Cuball P$ and $m\in\Mtest{}Q$, defined by
  $(\Mtestfun xm)(f)=m(f(x))$. So, in the proof above we would only
  know that, for all $m\in\Mtest{}Q$ and $x\in\Cuball P$, the function
  $\Real^n\to\Realp$ which maps $\Vect r$ to $m(\phi(\Vect r)(x))$ is
  measurable. But we would have to prove that, for all
  $m\in\Mtest{}Q$, the function
  $f=m\Comp\Ev\Comp\Pair\phi\gamma:\Real^n\to\Realp$ is measurable. We
  have $f(\Vect r)=m(\phi(\Vect r)(\gamma(\Vect r)))$ and the
  measurability of this function does not result from what we know
  about $\phi$.
\end{remark}

\begin{theorem}\label{th:stab-CCC}
  The category $\INFSCOTTM$ is cartesian closed. The pair $(\Implm PQ,\Ev)$ is
  the object of morphisms from $P$ to $Q$ in $\INFSCOTTM$.
\end{theorem}
\begin{proof}
  Let $f\in\INFSCOTTM(R\times P,Q)$. Let $z\in\Cuball R$, consider the
  function $f_z:\Cuball P\to\Cuball Q$ defined by
  $f_z(x)=f(z,x)$. This function is clearly non-decreasing and
  Scott-continuous. Let us prove that it is pre-stable. Let
  $\Vect u\in\Cuball P^n$ be such that $\sum_{i=1}^nu_i\in\Cuball P$
  and let $x\in\Cuballr P{\Vect u}$. We have:
  $\Fdiffp\epsilon{f_z}{x}{\Vect u}= \sum_{I\in\Cocard\epsilon
    n}f_z(x+\sum_{i\in I}u_i) =\sum_{I\in\Cocard\epsilon
    n}f(z,x+\sum_{i\in I}u_i) =\sum_{I\in\Cocard\epsilon
    n}f((z,x)+\sum_{i\in I}(0,u_i)) =\Fdiffp\epsilon
  f{(z,x)}{(0,u_1),\dots,(0,u_n)} $.

So we have $\Fdiffp-{f_z}{x}{\Vect u}\leq\Fdiffp+{f_z}{x}{\Vect u}$ by
the assumption that $f$ is pre-stable. We prove now that $f_z$ is
measurable. Let $n\in\Nat$ and $\gamma\in\Mpathu nP$, we must
prove that $f_z\Comp\gamma\in\Mpath nQ$. So let $k\in\Nat$ and let $m\in\Mtest
kQ$, we must prove that $\Mtestapp m{(f_z\Comp\gamma)}\in\Mfun{k+n}$. Let $\Vect
r\in\Real^k$ and $\Vect s\in\Real^n$, we have
$
  \Mtestapp m{(f_z\Comp\gamma)}(\Vect r,\Vect s)
  = m(\Vect r)(f(z,\gamma(\Vect s)))
  = \Mtestapp m{(f\Comp\Pair\zeta\gamma)}(\Vect r,\Vect s)
$,
where $\zeta\in\Mpathu nR$ is the measurable path defined by
$\zeta(\Vect s)=z$ (using Lemma~\ref{lemma:mes-path-cst}).  We know that
$f\Comp\Pair\zeta\gamma\in\Mpath nQ$ because $f$ is measurable and
$\Pair\zeta\gamma\in\Mpath{n}{P\times Q}$ by Lemma~\ref{lemma:mes-path-tuple},
hence $\Mtestapp m{(f\Comp\Pair\zeta\gamma)}\in\Mfun{k+n}$.  So
$f_z\in\Cuballp{\Implm PQ}$.

Let $g:\Cuball R\to\Cuballp{\Implm PQ}$ be the function defined by
$g(z)=f_z$. We prove that $g$ is pre-stable. Let
$\Vect w\in\Cuball R^p$ be such that $\sum_{j\in J}w_j\in\Cuball R$ and let
$z\in\Cuballr R{\Vect w}$. We have to prove that
\begin{align*}
  h^-=\Fdiffp-gz{\Vect w}\leq\Fdiffp+gz{\Vect w}=h^+
\end{align*}
in $\Cuballp{\Implm PQ}$; we apply Lemma~\ref{lemma:fun-order}. So let
$\Vect u\in\Cuball P^n$ and $x\in\Cuballr P{\Vect u}$, we must prove that
\begin{align}\label{eq:cur-inf-pre-stable}
  \Fdiffp+{h^-}{x}{\Vect u}+\Fdiffp-{h^+}{x}{\Vect u}\leq
  \Fdiffp+{h^+}{x}{\Vect u}+\Fdiffp-{h^-}{x}{\Vect u}
\end{align}
For $\epsilon,\epsilon'\in\{+,-\}$, we have
$
  \Fdiffp\epsilon{h^{\epsilon'}}x{\Vect u}
  =\sum_{I\in\Cocard\epsilon n}h^{\epsilon'}(x+\sum_{i\in I}u_i)
  =\sum_{I\in\Cocard\epsilon n}(\Fdiffp{\epsilon'} gz{\Vect w})(x+\sum_{i\in
    I}u_i)
 =\sum_{I\in\Cocard\epsilon n}
    \sum_{J\in\Cocard{\epsilon'}p}g(z+\sum_{j\in J}w_j)(x+\sum_{i\in I}u_i)
  =\sum_{I\in\Cocard\epsilon n}
    \sum_{J\in\Cocard{\epsilon'}p}f(z+\sum_{j\in J}w_j,x+\sum_{i\in I}u_i)
  =\sum_{I\in\Cocard\epsilon n}
    \sum_{J\in\Cocard{\epsilon'}p}f((z,x)+\sum_{q\in J\cup(p+I)}t_q)
   $, 
where $\Vect t\in\Cuballp{R\times P}^{p+n}$ is defined by
\begin{align*}
  t_q=
  \begin{cases}
    (w_q,0) & \text{ if } q\leq p\\
    (0,u_{q-p}) & \text{ if } p+1\leq q\leq  p+n
  \end{cases}
\end{align*}
and $p+I=\{p+i\St i\in I\}$.  Now observe that the map
$(J,I)\mapsto J\cup(p+I)$ defines a bijection
\begin{itemize}
\item between $(\Cocard+p\times\Cocard-n)\cup(\Cocard-p\times\Cocard+n)$ and
  $\Cocard-{p+n}$
\item and between $(\Cocard+p\times\Cocard+n)\cup(\Cocard-p\times\Cocard-n)$ and
  $\Cocard+{p+n}$.
\end{itemize}
It follows that we have
\begin{align*}
  \Fdiffp+{h^-}{x}{\Vect u}+\Fdiffp-{h^+}{x}{\Vect u}
  &=\Fdiffp-f{(z,x)}{\Vect t},&
  \Fdiffp+{h^+}{x}{\Vect u}+\Fdiffp-{h^-}{x}{\Vect u}
  &=\Fdiffp+f{(z,x)}{\Vect t}
\end{align*}
and hence~\Eqref{eq:cur-inf-pre-stable} holds because $f$ is
pre-stable. It also follows that $g$ is pre-stable, and Scott-continuity
is proven straightforwardly. 

Now we prove that $g$ is measurable. Let $\eta\in\Mpathu nR$, we must prove that
$g\Comp\eta\in\Mpath n{\Implm PQ}$. So let $k\in\Nat$, $\gamma\in\Mpathu kP$
and $m\in\Mtest kQ$, we must prove that
$\Mtestapp{(\Mtestfun\gamma m)}{(g\Comp\eta)}\in\Mfun{k+n}$.

Let $\Vect r\in\Real^k$ and $\Vect s\in\Real^n$, we have
  $(\Mtestapp{(\Mtestfun\gamma m)}{(g\Comp\eta)})(\Vect r,\Vect s)
  = (\Mtestfun\gamma m)(\Vect r)(g(\eta(\Vect s)))
  = m(\Vect r)(g(\eta(\Vect s)(\gamma(\Vect r)))$. By defintion of $g$, this is equal to $m(\Vect r)(f(\eta(\Vect s),\gamma(\Vect r))= \Mtestapp m{(f\Comp(\eta\times\gamma))}(\Vect r,\Vect s,\Vect r)$.

We know that $\eta\times\gamma\in\Mpath{n+k}P$ because $\eta$ and $\gamma$ are
measurable paths (and hence $\eta\Comp\pi_1$ and $\gamma\Comp\pi_2$ are
measurable paths, for $\pi_1$ and $\pi_2$ the projections
$\Real^{n+k}\to\Real^n$ and $\Real^{n+k}\to\Real^k$, and we have
$\eta\times\gamma=\Pair{\eta\Comp\pi_1}{\gamma\Comp\pi_2}$).
Hence $\Mtestapp m{(f\Comp(\eta\times\gamma))}\in\Mfun{k+n+k}$ and therefore
$\Mtestapp{(\Mtestfun\gamma m)}{(g\Comp\eta)}\in\Mfun{k+n}$.

So we have proven that $g\in\INFSCOTTM(R,\Implm PQ)$, which ends the proof of
the Theorem.
\end{proof}

\subsection{Integrating measurable paths}\label{subsect:integr_meas_paths}

The map $\Diracf:\Real\to\Bmes\Real$ such that $\delta(r)=\Dirac r$ belongs to
$\Mpathu 1{\Bmes\Real}$ because, given $n\in\Nat$, $U\in\Tribu\Real$,
$\Vect r\in\Real^n$ and $r\in\Real$ we have
$\Mapp U(\Vect r)(\delta(r))=\Mchar U(r)$ (where $\Mchar U$ is the
characteristic function of $U$). Recall from Example~\ref{ex:meas_cone_with_meas} that the elements of
$\Mtest n{\Bmes\Real}$ are precisely these functions $\Mapp U$.

Therefore, given $f\in\INFSCOTTM(\Bmes\Real,P)$, the function $f\Comp\Diracf$
is a measurable path from $\Real$ to $P$.

We have now to check that such paths are sufficiently regular for being
``integrated''. More precisely, given $\gamma\in\Mpath 1P$, we would like to
define a linear bounded map $\Mext\gamma:\Bmes\Real\to P$ by integrating in $P$
(using its algebraic structure and completeness):
\begin{align*}
  \Mext\gamma(\mu)=\int\gamma(r)\mu(dr)
\end{align*}
but we do not know  yet how to do that in general. 

We will focus instead on the (not so) particular case where
$P=(\Implm Q{\Bmes X})$ for a cone $Q$ and a measurable space $X$, which will be
sufficient for our purpose in this paper (every cone which is the denotation of a \ppcf{} type is isomorphic to a cone of the form $\Implm Q{\Bmes \real}$ for some $Q$).

\begin{lemma}\label{lemma:int-lin-cont-mes}
  Let $X$ be a measurable space and let $f:X\to\Realp$ be a measurable and
  bounded function. Then the function $F:\Bmes X\to\Realp$ defined by
  $F(\mu)=\int f(x)\mu(dx)$
  is linear and Scott-continuous.
\end{lemma}

\begin{lemma}\label{lemma:inf-cont-to-mes}
  Let $Q$ be a cone and $X$ be a measurable space. A function
  $f:\Cuball Q\to\Bmes X$ is stable iff for all $U\in\Tribu X$, the
  function $f_U:\Cuball Q\to\Realp$ defined by $f_U(y)=f(y)(U)$ is stable.
\end{lemma}

\begin{theorem}\label{th:linear-extension-path}
  Let $Q$ be a cone and $X$ be a measurable space.  For any
  $\gamma\in\Mpathu 1{\Implm Q{\Bmes X}}$, there is a measurable
  stable (actually linear) function
  $\Mext\gamma:\Bmes\Real\to(\Implm Q{\Bmes X})$ such that
  $\Mext\gamma\Comp\Diracf=\gamma$. This function is given by
  \begin{align*}
    \Mext\gamma(\mu)(y)(U)=\int\gamma(r)(y)(U)\mu(dr)
  \end{align*}
  for each $\mu\in\Bmes\Real$, $y\in Q$ and $U\in\Tribu X$.
\end{theorem}
\begin{proof}
Let $\gamma\in\Mpathu 1{\Implm Q{\Bmes X}}$, that we prefer to consider 
as a map
$\gamma_0:\Real \times \Cuball Q \times\Tribu X\to\Izu$
with $\gamma(r)(y)(U)=\gamma_0(r,y,U)$.
By Lemma~\ref{lemma:inf-cont-to-mes}, the fact that $\gamma$ is a measurable
path means that the following properties hold.
\begin{itemize}
\item For any $r\in\Real$ and $y\in\Cuball Q$, the map
  $U\mapsto\gamma_0(r,y,U)$ from $\Tribu X$ to $\Izu$ is a sub-probability
  measure;
\item for any $n\in\Nat$, $\eta\in\Mpathu nQ$ and $U\in\Tribu X$, the map
  $(r,\Vect r)\mapsto\gamma_0(r,\eta(\Vect r),U)$ from $\Real^{1+n}$ to
  $\Izu$ belongs to $\Mfun{1+n}$ (that is, is measurable);
\item for any $r\in\Real$ and $U\in\Tribu X$, the map $y\mapsto\gamma_0(r,y,U)$
  from $\Cuball Q$ to $\Realp$ is stable.
\end{itemize}

Therefore (by applying the second condition to $n=0$ and $\eta$ mapping the
empty sequence to $y$) we can define a function
$\phi:{\Bmes\Real}\times\Cuball Q\times\Tribu X\to\Realp$ by
$\phi(\mu,y,U)=\int\gamma_0(r,y,U)\mu(dr)$.

Let $\mu\in\Cuball{\Bmes\Real}$ and $y\in\Cuball Q$. The function
$\Tribu X\to\Izu$ which maps $U$ to $\phi(\mu,y,U)$ is $\sigma$-additive by
linearity and continuity of integration and defines therefore an element of
$\Cuball{\Bmes X}$. We denote by $\phi'$ the function
${\Bmes\Real}\times\Cuball Q\to\Bmes X$ defined by
$\phi'(\mu,y)(U)=\phi(\mu,y,U)\in\Realp$. This function is linear and
Scott-continuous in its first argument by Lemma~\ref{lemma:int-lin-cont-mes}.


Let $\mu\in\Cuball{\Bmes\Real}$ and $U\in\Tribu X$. We prove that the
map $f:\Cuball Q\to\Realp$ defined by $f(y)=\phi(\mu,y,U)$ is
stable. For any $\epsilon\in\{+,-\}$,
$n\in\Nat$, $\Vect u\in \Cuball Q^n$ such that $\sum_{i=1}^nu_i\in\Cuball Q$
and $y\in\Cuballr Q{\Vect u}$ one
has
$
  \Fdiffp\epsilon fy{\Vect u}=\int \Fdiffp\epsilon{f_r}{y}{\Vect u}\mu(dr)
$, 
where $f_r(y)=\gamma_0(r,y,U)$, by linearity of integration. Since the function
$f_r$ is pre-stable, we have
$\Fdiffp-{f_r}y{\Vect u}\leq\Fdiffp+{f_r}y{\Vect u}$ for each $r\in\Real$ and
hence $\Fdiffp-{f}y{\Vect u}\leq\Fdiffp+{f}y{\Vect u}$ as required. Given a
non-decreasing sequence $(y_n)_{n\in\Nat}$ in $\Cuball Q$, we must prove that
$f(\sup_n{y_n})=\sup_nf(y_n)$.  The sequence of measurable functions
$g_n:r\mapsto f_r(y_n)$ from $\Real$ to $\Izu$ is non-decreasing (for the pointwise
order) and satisfies $\sup_{n\in\Nat}g_n(r)=f_r(\sup_ny_n)$ by the last
condition on $\gamma_0$, and therefore, by
the monotone convergence theorem, we have
$f(\sup_{n\in\Nat}y_n)=\sup_{n\in\Nat}f(y_n)$.

So $\phi'$ is stable in its second argument (using the fact that the
order relation on measures in the cone $\Bmes X$ coincides with the
``pointwise order'': $\mu\leq\nu$ iff
$\forall U\in\Tribu X\ \mu(U)\leq\nu(U)$), and linear and
Scott-continuous in its first argument. Therefore, considered as a
function $\Cuball{\Bmes X}\times\Cuball Q\to\Bmes\Real$, $\phi'$ is
stable by Lemma~\ref{lemma:sep-lin-cont}. Now we must prove that this
function is measurable in the sense of
Definition~\ref{def:mes-pre-stable-function}.

So let $U\in\Tribu X$. Let $n\in\Nat$, $\theta\in\Mpathu n{\Bmes\Real}$ and
$\eta\in\Mpathu nQ$. The map $\rho:\Real^n\to\Izu$ defined by
$
  \rho(\Vect r)=\int\gamma_0(r,\eta(\Vect r),U)\theta(\Vect r,dr)
$
is measurable since $\theta$ is a stochastic 
kernel and
$g:\Real^{1+n}\to\Izu$ defined by $g(r,\Vect r)=\gamma_0(r,\eta(\Vect r),U)$ is
measurable by our assumptions about $\gamma$. Therefore
$\phi'\in\INFSCOTTM(\Bmes\Real\times Q,\Bmes X)$.

Let $\Mext\gamma\in\INFSCOTTM(\Bmes\Real,\Implm Q{\Bmes X})$ be the currying of
$\phi'$, that is $\Mext\gamma(\mu)(y)=\phi'(\mu,y)$. By
Theorem~\ref{th:stab-CCC}, $\Mext\gamma$ is stable and
measurable. Observe that this function is actually linear. 
\end{proof}

\section{Soundness and adequacy}\label{sect:sound_adeq}

%

\subsection{The Interpretation of $\ppcf$ into $\INFSCOTTM$}\label{sect:interp}
\begin{figure}
\begin{align*}
\semterm{x}(\vec g, a)&=a
&
\semterm{\num r}\vec g&=\Dirac r\qquad\qquad \semterm{\oracle}\vec g=\lambda_{[0,1]}
\\[4pt]
\semterm{\lambda x^\typea.\terma}\vec g&=a\mapsto\semterm{\terma}(\vec g,a)
&
\semterm{\num{f}(\terma_1,\dots,\terma_n)}\vec g&=U\mapsto (\semterm{\terma_1}\vec g\otimes\dots\otimes\semterm{\terma_n}\vec g)(f^{-1}(U))
\\[4pt]
\semterm{\terma\termb}\vec g&=\semterm{\terma}\vec g(\semterm{\termb}\vec g)
&
\semterm{\ifz \termc  \terma  \termb}\vec g&=(\semterm{\termc}\vec g\{0\})\,\semterm{\terma}\vec g+(\semterm{\termc}\vec g(\real\!\setminus\!\{0\}))\,\semterm{\termb}\vec g
\\[4pt]
\semterm{\fix\terma}\vec g&=\sup_n(({\semterm[]{\terma}}\vec g)^n0)
&
\semterm{\letterm x \terma  \termb}\vec g&=U\mapsto \int_\real \semterm{\termb}(\vec g,\Dirac r)(U) \, \semterm{\terma}\vec g(d r)\\[-30pt]
\end{align*}
\todoc[inline]{dans le point fixe, qui est $0$, vu son type, ce devrait être une mesure sur $\real$, est-ce que c'est la mesure nulle ? $0(U)=0$}
\caption{Interpretation of \ppcf{} in $\INFSCOTTM$. The terms are supposed typed as in Figure~\ref{fig_types}, and $\vec g\in\semtype{\Gamma}$, $a\in\semtype\typea$.  }\label{fig:interpretation}
\end{figure}
The interpretation of $\ppcf$ in $\INFSCOTTM$ extends the standard model of PCF
in a cpo-enriched category.  The ground type $\treal$ is denoted as the
cone $\Bmes\Real$ of  bounded measures over $\real$, the arrow $\typea\rightarrow\typeb$ by the object of morphisms $\Implm{\semtype{\typea}}{\semtype{\typeb}}$ and a sequence $\typea_1,\dots,\typea_n$ by the cartesian product $\prod_{i=1}^n\semtype{\typea_i}$ (recall Figure~\ref{fig:ccc}).
 The denotation of a judgement $\Gamma\vdash\terma:\typea$ is a morphism
$\semterm[\Gamma\vdash\typea]{\terma}\in\INFSCOTTM(\semtype{\Gamma},\semtype\typea)$,
given in Figure~\ref{fig:interpretation} by structural induction on
$\terma$. We omit the type exponent when clear from the
context. Notice that if $\Gamma\vdash\terma:\treal$, then for $\vec g\in\semtype\Gamma$, $\semterm \terma\vec g$ is a measure on $\real$. 

The fact that the definitions of Figure~\ref{fig:interpretation} lead
to morphisms in the category $\INFSCOTTM$ results easily from the
cartesian closeness of this category and from the algebraic and order
theoretic properties of its objects. The only construction which
deserves further comments is the $\texttt{let}$ construction. We use
the notations of Figure~\ref{fig:interpretation}, the typing context
is $\Gamma=(x_1:\typec_1,\dots,x_n:\typec_n)$.
Let
$Q=\semtype\Gamma=\semtype{\typec_1}\times\cdots\times\semtype{\typec_n}$. By
inductive hypothesis we have $\semterm M\in\INFSCOTTM(Q,\Bmes\Real)$ and
$\semterm{N'}\in\INFSCOTTM(\Bmes\Real,\Impl Q{\Bmes\Real})$ where
$N'=\lambda x_1^{C_1}\dots\lambda x_n^{C_n} N$ (up to trivial isos resulting
from the cartesian closeness of $\INFSCOTTM$). Then
$\semterm{N'}\Comp\Diracf\in\Mpathu 1{\Impl Q{\Bmes\Real}}$ because
$\Diracf\in\Mpathu 1{\Bmes\Real}$ (see
Section~\ref{subsect:integr_meas_paths}). Hence we define
$\Mext{(\semterm{N'}\Comp\Diracf)}\in\INFSCOTTM(\Bmes\Real,\Impl
Q{\Bmes\Real})$
by setting
$\Mext{(\semterm{N'}\Comp\Diracf)}(\mu)(\vec g)(U)=\int\semterm{N'}(\Dirac
r)(\vec g)(U)\mu(dr)=\int\semterm N(\vec g,\Dirac r)(U)\mu(dr)$
for $\mu\in\Bmes\Real$, $\vec g\in Q=\semtype\Gamma$ and $U\in\Tribu\Real$, by
Theorem~\ref{th:linear-extension-path} (remember that
$\Tseq{\Gamma,x:\treal}{N}{\treal}$). By cartesian closeness, we define 
 $f$ in $\INFSCOTTM(Q,\Bmes\Real)$ by 
$f(\vec g)(U)=\Mext{(\semterm{N'}\Comp\Diracf)}(\semterm M(\vec g))(\vec
g)(U)=\int\semterm N(\vec g,\Dirac r)(U)\semterm M(\vec g)(dr)$.
Hence, $\semterm{\letterm x \terma \termb}$ belongs to
$\INFSCOTTM(Q,\Bmes\Real)$ as required. Observe moreover that, for $r\in\Real$,
we have
$\semterm{\letterm x {\num r} \termb}\vec g=\semterm\termb(\vec g,\Dirac r)$.

\begin{example}\label{ex:sem_simple_examples}
Numerals are associated with Dirac measures and a functional constant $\num f$
yields the pushforward measure of the product of the measures denoting the arguments of $\num f$. For example, we have: $\semterm[\vdash \treal]{\num +(\num 3, \num 2)}=U\mapsto \Dirac{3}\otimes\Dirac{2}(\{(r_1,r_2)\text{ s.t. } r_1+r_2\in U\})=\Dirac{5}$. 

The construct $\mathtt{ifz}$ sums up the denotation of the two branches according to the probability that the first term evaluates to $\num 0$ or not. Given a measurable set $U\subseteq \real$,  a closed term $\termc$ of ground type and two closed terms $\terma,\termb$ of a type $\typea$, we have that, recalling the notation of Example~\ref{ex:if}, $\semterm[\vdash\typea]{\ifterm \termc U\terma \termb}
=(\semterm[\vdash\treal]{\realfun{\chi_U}(\termc)}(\real\setminus\{0\}))\, \semterm[\vdash\typea]\terma
+(\semterm[\vdash\treal]{\realfun{\chi_U}(\termc)}(\{0\}))\,\semterm[\vdash\typea]\termb
=(\semterm[\vdash\treal]\termc(U))\,\semterm\terma+(\semterm[\vdash\treal]\termc(\real\setminus U))\,\semterm\termb$.
\end{example}

\begin{example}
The two terms implementing the diagonal in Example~\ref{rk:letVSabs} have different semantics: for any measurable $U$ of $\Real$, for any $r,s\in U$, $r=s$ has value $0$ or $1$. Besides, the diagonal   $\{(r,s)\text{ s.t. } r=s\in \{1\}\}$ in $[0,1]^2$ has measure $0$, and its complementary $\{(r,s)\text{ s.t. } r=s\in \{0\}\}$ has measure $1$. Thus,
\begin{align*}
  \semterm[\vdash\treal]{(\lambda x.(x=x))\oracle}(U)
    &= (\lambda_{[0,1]}\otimes \lambda_{[0,1]})\{(r,s)\text{ s.t. } r=s\in U\}=\Dirac 0(U).
\end{align*}
On the contrary,   $\semterm[\vdash\treal]{\letterm x{\oracle}{x=x}}(U)=\Dirac 1(U)$. Indeed, $\semterm[\vdash\treal]{\letterm x{\oracle}{x=x}}(U)$ is
\[
\int_\Real (\Dirac r\otimes\Dirac r)\{(x,y)\text{ s.t. } x=y\in U\}\lambda_{[0,1]}(dr)
  =\int_\Real \Dirac 1(U)\lambda_{[0,1]}(dr)= \Dirac 1(U).
\]
\end{example}

\begin{example}\label{ex:sem_distributions}
Let us compute the semantics of the encodings of the distributions in Example~\ref{ex:gaussian}. Let $p\in[0,1]$, then $\semterm[\vdash\treal]{\bernoulli\, \num p}=p\Dirac 1+(1-p)\Dirac 0$ is given by,  for $U$ measurable: $\semterm[\vdash\treal]{\bernoulli\,\num p}(U)= \int_\Real \Dirac r\otimes\Dirac p(\{(x,y)\text{ s.t. } x\le y \in U\})\lambda_{[0,1]}(dr)$, this latter being equal to $\lambda_{[0,1]}([0,p])\Dirac 1(U)+\lambda_{[0,1]}((p,1])\Dirac 0(U)$.


The exponential distribution $\exp$ computes the probability that an exponential random variable belongs to $U$:
   $\semterm[\vdash\treal]{\exp}(U)=\semterm[\vdash\treal]{\letterm x\oracle{\realfun{-\log}(x)}}(U)= \int_\Real\Dirac r(\{x \text{ s.t. }-\log x\in U\}) \lambda_{[0,1]}(dr)=\int_\Real \chi_U(-\log r)\lambda_{[0,1]}(dr)$, which is equal to $\int_{\preal}\chi_U(s)\mathrm{e}^{-s}\lambda(ds)$ by substitution $r=\mathrm{e}^{-s}$.
We compute the semantics of $\normal$ and check that we get a normal
distribution: 
\begin{align*}
  \semterm[\vdash\treal]{\normal}(U)&=\semterm[\vdash\treal]{\letterm x{\oracle}{\letterm y{\oracle}{\realfun{(-2\log} (x)\realfun{)^{\tfrac 12}\, \cos(2\pi} y\realfun)}}}(U)\\
&= \int_{\Real^2}\chi_U(\sqrt{-2\log u}\, \cos(2\pi v))\lambda_{[0,1]}(du)\lambda_{[0,1]}(dv).
\end{align*}
 By polar substitution with 
$x= \sqrt{-2\log u} \cos(2\pi v),\ y= \sqrt{-2\log u} \cos(2\pi v)$, we then have: $ \semterm[\vdash\treal]{\normal}(U)=\tfrac 1{2\pi}\int_{\real^2}\chi_U(x)\mathrm{e}^{-\tfrac{(x^2+y^2)}{2}}\lambda(dx)\lambda(dy) =\tfrac 1{\sqrt{2\pi}} \int_U\mathrm{e}^{-\tfrac{x^2}2}\lambda(dx)$, which  is what we wanted.
\\
Similarly, $\semterm[\vdash\treal]{\gaussian\ \num r\,\num\sigma}(U)=\tfrac 1{\sqrt{2\pi}} \int_\real \chi_U(\sigma\,y+r)\mathrm e^{-\tfrac{y^2}2}\lambda(dy)=\tfrac 1{\sigma\sqrt{2\pi}} \int_U\mathrm e^{-(\tfrac{z-r}{\sigma})^2}\lambda(dz)$.
\end{example}

\begin{example}\label{ex:sem_expect}
Recall Example~\ref{ex:expectation}, let $\num f\in\Const$ and $M$ be a term of type $\treal$. We want to check that $\mathtt{expectation}_n \num f\, M$ corresponds to the $n$-th estimate of the expectation of $f$ with respect to the measure $\semterm[\vdash\treal]\terma$, meaning that $\semterm[\vdash\treal]{\mathtt{expectation}_n  \num f\, M}$ has the same measure as $\tfrac{f(\mathbf x_1)+\dots+f(\mathbf x_n)}{n}$ where $\mathbf x_i$'s 
are iid random variables of measure $\semterm[\vdash\treal]\terma$. For all $U\subseteq \real$ measurable, $\semterm[\vdash\treal]{\mathtt{expectation}_n  \num f\, M}(U)=\underbrace{\semterm[\vdash\treal]\terma\otimes\,\cdots\,\otimes \semterm[\vdash\treal]\terma}_n(\{(x_1,\cdots,x_n)\text{ s.t. }\tfrac{f(x_1)+\dots+f(x_n)}{n}\in U\})$ which is what we wanted.
\end{example}

The following two lemmas are standard and proven by structural induction.
\begin{lemma}[Substitution property]\label{lemma:substitution}
Given $y:\typeb,\Gamma\vdash\terma:\typea$ and $\Gamma\vdash \termb:\typeb$ we have, for every $\vec g\in\semtype\Gamma$, that $\semterm[y:\typeb,\Gamma\vdash\typea]{\terma}(\semterm[\Gamma\vdash\typeb]\termb\vec g)\vec g=\semterm[\Gamma\vdash\typea]{\terma\{\termb/y\}}\vec g$.
\end{lemma}
%
\begin{lemma}[Linearity evaluation context]\label{lemma:linearity_eval}
Let $y:\typeb,\Gamma\vdash E[y]:\typea$ for $E[\;]$ an evaluation context and $y$ a fresh variable. Then $\semterm[y:\typeb,\Gamma\vdash\typea]{E[y]}\in\INFSCOTTM(\semtype\typeb\times\semtype\Gamma,\semtype\typea)$ is a linear function in its first argument $\semtype\typeb$. 
\end{lemma}

\subsection{Soundness}

The soundness property states that the interpretation is invariant under reduction. In a non-deterministic case, this means that the semantics of a term is the sum of the semantics of all its possible one-step reducts, see e.g.~\cite{LairdMMP13}. In our setting, the reduction is a stochastic kernel, so this sum becomes an integral, i.e.~for all $\terma\in\Terms[\Gamma\vdash\typea]$,
\begin{equation}\label{eq:soundness_explicit}
	\semterm[\Gamma\vdash\typea]{\terma} = \int_{\Terms[\Gamma\vdash\typea]}\semterm[\Gamma\vdash\typea]{t}\Red(M,dt)
\end{equation}
The following lemma actually proves   that the above integral is a meaningful notation for the function mapping $\vec g\in\semtype\Gamma$,  and, supposing $\typea = \typeb_1\rightarrow\dots\rightarrow \typeb_k\rightarrow\treal$, $b_1\in\semtype{\typeb_1}$,\dots, $b_k\in\semtype{\typeb_k}$ and $U\in\Sigma_\real$, to $\int_{\Terms[\Gamma\vdash\typea]}\semterm[\Gamma\vdash\typea]{t}\vec
gb_1\dots b_k(U)\Red(M,dt)$, this latter being well-defined because $\semterm[\Gamma\vdash\typea]{t}\vec
gb_1\dots b_k(U)$ is measurable (Lemma~\ref{lemma:interp_kernel}) and $\Red(M,\_)$ is a measure (Proposition~\ref{prop:red_kernel}).
\begin{lemma}\label{lemma:interp_kernel}
Let $\Gamma\vdash\terma:\typea$, with $\typea=\typeb_1\rightarrow\dots\typeb_k\rightarrow\treal$. For all $i\leq k$, let $b_i\in\semtype{\typeb_i}$ and $\vec g\in\semtype{\Gamma}$, then the map $\terma\mapsto\semterm[\Gamma\vdash\typea]{\terma}\vec gb_1\dots b_k$ is a stochastic kernel from $\Terms[\Gamma\vdash\typea]$ to $\real$.
\end{lemma}
\begin{proof}[Proof (Sketch).]
By \eqref{eq:sigma_algebra_terms}, it is enough to prove that, for any $\fterma\in\TermsFree[\Gamma\vdash\typea]n$, the restriction ${\semterm{\_}}_\fterma\vec g\vec b$
 of ${\semterm{\_}}\vec g\vec b$ to $\TermsFree[\Gamma\vdash\typea]\fterma$ is a kernel. This is done by using the crucial fact that the map $h=\semterm{\fterma}\circ(\Diracf^n\times\vec g)$ is a measurable path in $\Mpathu n{\semtype\typea}$. This implies that $\vec r\mapsto h(\vec r)\vec b$ is in $\Mpathu n{\Bmes\real}$, so it is a stochastic 
kernel from $\real^n$ to $\real$ (Example~\ref{ex:meas_cone_with_meas}). We are done, since $\Real^n$ and $\TermsFree[\Gamma\vdash\typea]\fterma$ are isomorphic.
 \end{proof}
\begin{proposition}[Soundness]\label{prop:soundmess}
Let $\typea=\typeb_1\rightarrow\dots\typeb_k\rightarrow\treal$, for all $i\leq k$, $b_i\in\semtype{\typeb_i}$, and let $\vec g\in\semtype{\Gamma}$, then $(\semterm[\Gamma\vdash\typea]{\_}\vec gb_1\dots b_k) \circ \Red=\semterm[\Gamma\vdash\typea]{\_}\vec gb_1\dots b_k$, i.e.~Equation~\eqref{eq:soundness_explicit} holds for any $\terma\in\Terms[\Gamma\vdash\typea]$.
\end{proposition}
\begin{proof}[Proof (Sketch)]
If $\terma$ is a normal form, then the statement is trivial. Otherwise, let $\terma=E[R]$ with $R$ a redex  (Lemma~\ref{lemma:evaluation_decomposition}). 
If $R\neq\text{\oracle}$, let $R\rightarrow \termb$. By the substitution property (Lemma~\ref{lemma:substitution})) it is sufficient to prove $\semterm R=\semterm N$ to conclude. This is done by cases, depending on the type of $R$. 
The last case is $\terma=E[\oracle]$. This is obtained by using the linearity of the evaluation context $E[\,]$ (Lemma~\ref{lemma:linearity_eval}) and the substitution property (Lemma~\ref{lemma:substitution}).
\end{proof}

\begin{example}\label{ex:semantics_observe}
Suppose $\terma$ a closed term of type $\treal$ and consider $\vdash \observe U\terma:\treal$ introduced in Example~\ref{ex:observe} as an encoding of the conditioning. We compute its semantics by using soundness. 
Since 
$\observe UM \to^* \letterm x\terma  {\ifterm xUx{\observe U\terma}}$, we get by soundness 
that for all $V\subseteq\Real$ measurable,
$
\semterm{\observe U\terma}(V)= \int_\Real \semterm[x:\treal\vdash\treal]{\ifterm xUx{\observe U\terma}}(\Dirac r)(V)\, \semterm\terma(dr)
= \int_\Real(\Dirac r(U)\,\Dirac r(V)+(\Dirac r(\real \setminus U))\,(\semterm{\observe U\terma}(V)))\, \semterm\terma(dr)$. Since $\semterm{\observe U\terma}$ does not depend on $r$, the latter integral can be rewritten to: $\semterm{\observe U\terma}(V)=\int_\Real(\chi_U(r)\,\chi_V(r))\, \semterm\terma(dr)+(\semterm{\observe U\terma}(V))\,\int_\Real\chi_{\real\setminus U}(r)\,\semterm\terma(dr)$. 

Whenever $\terma$ represents a probability distribution, so that $\semterm\terma(U)=1-\semterm\terma(\real\setminus U)$ and if moreover $\semterm\terma(U)\neq 0$, this equation gives the conditional probability:
\begin{equation*}
  \semterm{\observe U\terma}(V) 
  = \frac{\int_\Real(\chi_U(r)\,\chi_V(r))\, \semterm\terma(dr)}{1-\int_\Real\chi_{\real\setminus U}(r)\, \semterm\terma(dr)}
  =\frac{\semterm\terma(V\cap U)}{\semterm\terma(U)}
\end{equation*}
If $\semterm\terma (U)=0$, then as $(\semterm{\lambda y\,\letterm x\terma{\ifterm x U x y}})^n0=0$, the denotation of the fixpoint is $\semterm{\observe U\terma}=0$. By adequacy, the program then loops with probability $1$ when $\semterm\terma (U)=0$.

Now, consider the term $O = \lambda m.\fix(\lambda y.\ifterm mUmy)$
presented in Example~\ref{ex:observe} as a wrong implementation of $\observe
U$. Since $O\terma\to^* \ifterm \terma U\terma{O\terma}$, assuming that $\semterm\terma$ is a
probability distribution and $V$ a measurable set, one gets with a similar
reasoning that, in case $\semterm\terma U\neq 0$,
$\semterm{O\terma}(V)=(\semterm\terma(V)\,\semterm\terma(U))/\semterm\terma(U)=\semterm\terma(V)$. 
As before, if $\semterm\terma (U)=0$, then $\semterm{O\terma}=0$.

\end{example}

\todom{here commented metropolis-hasting:semantics}
\subsection{Adequacy}

Let $\terma$ be a closed term of ground type of \ppcf. Both the operational and the denotational semantics associate with $\terma$ a distribution over $\Real$ --- the adequacy property states that these two distributions are actually the same (Theorem~\ref{th:adequacy}). The proof is standard: the soundness property gives as a corollary that the ``operational'' distribution is bounded by the ``denotational'' one. The converse is obtained by using a suitable logical relation (Definition~\ref{def:logical relation}, Lemma~\ref{lemma:key-lemma}).

\begin{definition}\label{def:logical relation}
By induction on a type $\typea$, we define a relation $\prec^\typea\subseteq\semtype\typea\times\Terms[\vdash\typea]$ as follows:
\begin{align*}
\mu \prec^\treal \terma&\quad\text{ iff }\quad \forall U\in\Sigma_\Real, \mu(U)\leq\Red^{\infty}(\terma,\num U),\\
f \prec^{\typea\rightarrow\typeb} \terma&\quad\text{ iff }\quad \forall u\prec^\typea\termb,  f(u)\prec^\typeb\terma\termb.
\end{align*}
\end{definition}

\begin{lemma}\label{lemma:key-lemma}
Let $x_1:\typeb_1,\dots,x_n:\typeb_n\vdash\terma:\typea$ and $\forall i\leq n, u_i\prec^{\typeb_i}\termb_i$, then: $\semterm\terma\vec u\prec^{\typea}\terma\{{\vec\termb/\vec x}\}$.
\end{lemma}

\begin{theorem}[adequacy]\label{th:adequacy}
Let $\vdash\terma:\treal$, then for every measurable set $U\subseteq\Real$, we have:
\[
	\semterm[\vdash \treal]\terma(U) = \Red^{\infty}(\terma,\num U)
\]
where $\num U$ is the set of numerals corresponding to the real numbers in $U$.
\end{theorem}

\section{Related work and conclusion}\label{sect:related_work}

%
The first denotational models for higher-order probabilistic programming were based on
probabilistic power
domains~\citep{Saheb-Djahromi80,JonesPlotkin89}. This setting follows a
monadic approach, considering a program as a function from inputs to
the probabilistic power domain of its outputs. The major issue here is
to find a cartesian closed category which is also closed under the
probabilistic power domain monad~\citep{JungT98}. Some advances have
been obtained by~\citet{phdBarker}, using a monad based on
random variables inspired by~\citet{goubaultvarraca11}. Besides, \citet{Mislove16} has
  introduced a domain theory of
random variables.
Another approach is based on game semantics, designing models of probabilistic languages with references~\citep{probagames} or concurrent features~\citep{Winskel14}.

The notions of \emph{d-cones}~\citep{TixKP09a} and
\emph{Kegelspitzen}~\citep{KeimelP16} 
are promising 
for getting a
family of models different from 
ours. \citet{Rennela16}
has recently used this approach for studying a probabilistic extension
of FPC. A Kegelspitzen is a convex set of a positive
cone equipped with an order compatible with the algebraic structure of
the cone. Notice that this notion differs from ours because the order
of a Kegelspitzen might be independent from the one induced by its
algebraic structure. It is likely that the two approaches live in two
different but related frameworks as the continuous and the stable
semantics of standard PCF.

The denotational semantics approach to probabilistic programming has been recently relaunched by the increasing importance of continuous distributions and sampling primitives. Indeed, this raises the question of the \emph{measurability of a morphism} as the interpretation of the sampling primitives requires integration. This question has not been investigated yet in the domain theoretic approach and forces to introduce a new line of works which puts the focus on measurability.

The challenge is to define a cartesian closed category in which base types such as reals would be interpreted as measurable spaces. As mentioned in the Introduction, the category $\Meas{}$ of  measurable spaces and functions is cartesian but not closed. To overcome this problem, \citet{StatonYWHK16} embed $\Meas$  in a functor category which is cartesian closed although not well-pointed. Then, to get a more concrete and a well-pointed category, they introduce the category of \emph{quasi-borel spaces}~\cite{StatonYHK17} which are sets endowed with a set of random variables. 
Notice that both categories 
miss the order completeness, and thus the possibility of interpreting higher-order recursion. This is a big difference with our model $\INFSCOTTM$ which is order complete. 

Let us also cite the ongoing efforts presented last year at the
workshop PPS by \citet{HuangM17}, aiming to give a model
based on computable distributions, and by \citet{FaissoleS17}, working on a Coq formalization of a
semantics built on top of the constructions detailed
in~\citep{StatonYWHK16}.

\medskip
In this paper, we have presented $\INFSCOTTM$, a new model of higher-order probabilistic computations with
full recursion, as a cartesian closed category enriched over posets which are
complete for non-decreasing sequences. The objects of $\INFSCOTTM$ are cones
equipped with a notion of measurability tests and morphisms are functions
which are measurable in the sense that they behave well wrt.~this notion of
measurability tests. These functions are also Scott-continuous, but this
is not sufficient for guaranteeing cartesian closeness: they must satisfy an
hereditary monotonicity condition that we call \emph{stability} because, when
adapted to coherence spaces, it coincides with Berry-Girard stability. The
introduction of this notion of ``probabilistic stability'' is a relevant byproduct of our approach.

A typical example of such a cone is the set of $\Realp$-valued measures on the
real line that we use to interpret the type of real numbers, the unique ground
type of \ppcf{}, a probabilistic version of PCF. This language also features a
${\tt sample}$ primitive allowing to sample a real number according to a
prescribed probability measure on the reals (intuitively, a closed \ppcf{} term
of ground type represents a sub-probability measure on the real line). We have
presented the semantics of \ppcf{} in $\INFSCOTTM$ and proven adequacy for a call-by-name operational semantics.

There are many research directions suggested by these new constructions, namely to study the category $\mathbf{Clin}_{\mathsf m}$ of linear and measurable Scott continuous maps mentioned in the Introduction and prove the conjectures sketched in Figure~\ref{fig:the_picture}. Also, full-abstraction will be addressed, following \citep{EhrPagTas14}.

\begin{acks}                            
  
  We would like to thank Ugo Dal Lago and Jean-Louis Krivine for useful discussions. This material is based upon work supported by the ANR grant Elica (No.~\grantnum{ANR-14-CE25-0005}{ANR-14-CE25-0005}).
    
\end{acks}

\newpage

\bibliography{biblio}

\newpage
\appendix
\section{Appendix}\label{appendix:technical}
%
%

\subsection{Proofs of Section~\ref{sect:ppcf}}

\paragraph{Lemma~\ref{lemma:subst_measurable}.} Given $\Gamma,x:\typeb\vdash\terma:\typea$ the function $\subst{x,\terma}$ mapping $\termb\in\Terms[\Gamma\vdash\typeb]$ to $\terma\{\termb/x\}\in\Terms[\Gamma\vdash\typea]$ is measurable. 
\begin{proof}
Since $\Terms[\Gamma\vdash A]$ can be written as the
coproduct~\eqref{eq:sigma_algebra_terms}, it is sufficient to prove that for
any $n$ and $\ftermb\in\TermsFree[\Gamma\vdash\typea]n$,
$\subst{x,\terma}:\Terms[\Gamma\vdash B]_\ftermb\rightarrow \Terms[\Gamma\vdash
A]$ is measurable. Let $\fterma$ and $\vec r\in\real^m$ be such that $\terma=\fterma\vec r$ and let $U\subseteq\Terms[\Gamma\vdash\typea]$. We prove that $\subst{x,M}^{-1}(U)=\{\vec {r'}\in\real^n\text{ s.t. } \fterma\vec r\{\ftermb\vec{r'}\slash x\}\in U\}$ is measurable. Let $k$ be the number of occurrences of $x$ in $\terma$ and let us enumerate these occurrences as $x_1,\dots, x_k$. Then there are $i_1,\dots,i_k$, such that $0\le i_1\le\dots\le i_k\leq m$ such that:
\begin{equation*}
\fterma\vec r\{\ftermb \vec{r'}\slash x\}=\fterma\{\ftermb\slash x_1,\dots,\ftermb\slash x_k\}
	r_1\dots r_{i_1}\vec r'r_{i_1+1}
	\dots
	r_{i_{k-1}}\vec r'r_{i_{k-1}+1}\dots r_{i_{k}}\vec r'
	r_{i_{k}+1}\dots r_{m}
\end{equation*}
with $\fterma\{\ftermb\slash x_1,\dots,\ftermb\slash x_k\}$ a real-free
 term. The decomposition of $\vec r$ into the $k+1$ sections above, depends on the positions of the various occurrences of $x$ in $\fterma$. 
Using~\eqref{eq:measurable_set_terms}, it is sufficient to remark that  $\vec{r'}\mapsto 	
r_1\dots r_{i_1}\vec r'r_{i_1+1}
	\dots
	r_{i_{k-1}}\vec r'r_{i_{k-1}+1}\dots r_{i_{k}}\vec r'
	r_{i_{k}+1}\dots r_{m}
$ is a measurable function $\real^n\to\real^{m+kn}$.
\end{proof}

\paragraph{Proposition~\ref{prop:red_kernel}.} For any sequent $\Gamma\vdash A$, the map $\Red$ is a stochastic kernel from $\Terms[\Gamma\vdash A]$ to $\Terms[\Gamma\vdash A]$.
\begin{proof}
Let $\terma$ be a term. The fact that $\Red(\terma ,\_)$ is a measure from $\Terms[\Gamma\vdash A]$ to $[0,1]$ is an immediate consequence of the definition of $\Red$ and the fact that any evaluation context $E[\;]$ defines a measurable map $\subst{x,E[x]}:\terma  \rightarrow E[\terma ]$ from $\Terms[\Gamma\vdash A]$ to $\Terms[\Gamma'\vdash A']$ (Lemma~\ref{lemma:subst_measurable}).

Given a measurable set $U\subseteq\Terms[\Gamma\vdash A]$, we must prove that $\Red(\_,U)$ is a measurable function from $\Terms[\Gamma\vdash A]$ to $[0,1]$. 
Since $\Terms[\Gamma\vdash A]$ can be written as the coproduct in Equation~\eqref{eq:sigma_algebra_terms}, it is sufficient to prove that for any $n$ and $\fterma\in\TermsFree[\Gamma\vdash\typea]n$, $\Red_\fterma(\_,U):\Terms[\Gamma\vdash A]_\fterma\rightarrow[0,1]$ is a measurable function.  

We reason by case study on the shape of $\fterma$. Notice that by using
Lemma~\ref{lemma:evaluation_decomposition} and the definition of a redex we
have that: either (i) for all $\vec r$, $\fterma\vec r$ is a normal form, or
(ii) $\fterma=E[\ftermb]$ such that for all $\vec r$, $\ftermb\vec r$ is a
redex. In case (i), $\Red_\fterma(\_,U)=\chi_U$ and we are done. Otherwise, we first tackle the case where $\ftermb=\oracle$. Notice that $\Terms[\Gamma\vdash A]_{E[\oracle]}=\{E[\oracle]\}$, so that the constant map $\Red_\fterma(\_,U)=\lambda\{r\in[0,1]\text{ s.t. } E[\num r]\in U\}$ is measurable.\todoc{with the case sample, ok ?}

Now, we focus on the tricky case where $\ftermb\neq \oracle$.   
Notice that it is sufficient to prove that  $\Red_{\fterma}(\_,U)^{-1}(\{1\})=\{E[\ftermb]\vec r\text{ s.t. } \ftermb\vec r\to\termb\text{ and }\ E[\termb]\in U\}$ is measurable, then $\Red_{\fterma}(\_,U)^{-1}(\{0\})=\{E[\ftermb]\vec r\text{ s.t. } \ftermb\vec r\to\termb\text{ and }\ E[\termb]\notin U\}$ is also measurable as the complementary of a measurable set in $\Terms[\Gamma\vdash A]_\fterma$ and finally, $\Red_{\fterma}(\_,U)^{-1}(]0,1[)=\emptyset$ is also measurable. We reason again by case study on the shape of the redex $\ftermb$. If $\ftermb=(\lambda x.\ftermc_0)\ftermc_1$ then $\Red_{\fterma}(\_,U)^{-1}(\{1\})=\{E[\ftermb]\vec r\text{ s.t. } E[\ftermc_0\{\ftermc_1\slash x\}]\vec r\in U\}$ which is measurable thanks to~\eqref{eq:measurable_set_terms} and Lemma~\ref{lemma:subst_measurable}. If $\ftermb=\ifz{z}{\ftermc_0}{\ftermc_1}$, then $\Red_{\fterma}(\_,U)^{-1}(\{1\})=\{E[\ftermb](\vec r,0)\text{ s.t. } E[\ftermc_0]\in U\}\cup \{E[\ftermb](\vec r,r)\text{ s.t. } E[\ftermc_1]\in U,\text{ and }r \in ]0,1]\}$  which is measurable thanks to~\eqref{eq:measurable_set_terms}.
\end{proof}

\subsection{Proofs of Section~\ref{sec:cones}}

\paragraph{Lemma~\ref{lemma:local-cone}}
  For any cone $P$ and any $u\in\Cuball P$, $\Crel Pu$ is a
  cone. Moreover
  $\Cuball{\Crel Pu}=\{x\in\Ccarrier P\St x+u\in\Cuball P\}$ and, for
  any $x\in\Ccarrier{\Crel Pu}$, one has
  $\Cnorm Px\leq\Cnorm{\Crel Pu}x$. If $P$ is complete then
  $\Crel Pu$ is complete.

\begin{proof}
Observe first that $0\in\Ccarrier{\Crel Pu}$ because $u\in\Cuball P$.  Let us
check that $\Ccarrier{\Crel Pu}$ is closed under addition. Let
$x,x'\in\Ccarrier{\Crel Pu}$ and let $\epsilon,\epsilon'$ be such that
$u+\epsilon x,u+\epsilon'x'\in\Cuball P$. Without loss of generality we can
assume that $\epsilon\leq\epsilon'$ and hence we have
$u+\epsilon x,u+\epsilon x'\in\Cuball P$ and therefore
$u+\frac\epsilon 2(x+x')\in \Cuball P$ because $\Cuball P$ is convex. It follows
that $x+x'\in\Ccarrier{\Crel Pu}$. Let $x\in\Ccarrier{\Crel Pu}$, we have
$0x=0\in\Ccarrier{\Crel Pu}$. Let now $\alpha>0$. Let $\epsilon>0$ be such
that $\epsilon x+u\in\Cuball P$. We have therefore
$\frac\epsilon\alpha(\alpha x)+u\in\Cuball P$ and hence $\alpha
x\in\Ccarrier{\Crel Pu}$.

We prove now that $\Cnorm{\Crel Pu}{\_}$ is a norm. The fact that
$\Cnorm{\Crel Pu}{0}=0$ is clear. Let $x\in\Ccarrier{\Crel Pu}\setminus\{0\}$.
Let $\alpha >\Cnorm Px^{-1}$, we have $\alpha x\notin\Cuball P$ and hence
$\alpha x+u\notin\Cuball P$ and therefore
$\Cnorm{\Crel Pu}x\geq\frac 1\alpha$. We have proven that
$\Cnorm{\Crel Pu}x=0\Implies x=0$. Let $x,x'\in\Ccarrier{\Crel Pu}$, we prove
that $\Cnorm{\Crel Pu}{x+x'}\leq\Cnorm{\Crel Pu}{x}+\Cnorm{\Crel Pu}{x'}$. Let
$\alpha>0$. By definition of $\Cnorm{\Crel Pu}{x}$ we can find $\epsilon>0$
such that $\Cnorm P{\epsilon x+u}\leq 1$ and
$\Cnorm{\Crel Pu}{x}\geq\frac 1\epsilon-\alpha$. Similarly we can find
$\epsilon'>0$ such that $\Cnorm P{\epsilon' x'+u}\leq 1$ and
$\Cnorm{\Crel Pu}{x}\geq\frac 1{\epsilon'}-\alpha$. We have
\begin{align*}
  \Cnorm P{x+x'+(\frac 1\epsilon+\frac 1{\epsilon'})u}\leq \frac
  1\epsilon+\frac 1{\epsilon'}
\end{align*}
and hence
$\Cnorm{\Crel Pu}{x+x'}\leq\frac 1\epsilon+\frac
1{\epsilon'}\leq\Cnorm{\Crel Pu}{x}+\Cnorm{\Crel Pu}{x'}+2\alpha$.
Since this holds for all $\alpha>0$, we have
$\Cnorm{\Crel Pu}{x+x'}\leq\Cnorm{\Crel Pu}{x}+\Cnorm{\Crel Pu}{x'}$.
It is straightfoward that
$\Cnorm{\Crel Pu}{x}\leq\Cnorm{\Crel Pu}{x+x'}$ (because
$\Cnorm P{\epsilon x+u}\leq\Cnorm P{\epsilon(x+x')+u}$). A similar
reasoning allows to prove that
$\Cnorm{\Crel Pu}{\alpha x}=\alpha\Cnorm{\Crel Pu}x$ for all
$x\in\Ccarrier{\Crel Pu}$ and $\alpha\in\Realp$ (one has actually to
distinguish two cases: $\alpha=0$ and $\alpha>0$; the first case has
already been dealt with).

Now we prove that $\Cuball{\Crel Pu}=\{x\in\Ccarrier P\St x+u\in\Cuball
P\}$. Let $x\in\Cuball{\Crel Pu}$. There exists a non-decreasing sequence
$(\epsilon_n)_{n\in\Nat}$ such that $\epsilon_n>0$ and $\epsilon_nx+u\in\Cuball
P$ for all $n$, and moreover $\sup_{n\in\Nat}\epsilon_n=1$. Then by closeness
of $P$ we have $x+u\in\Cuball P$. The converse inclusion is obvious.

Let $x\in\Ccarrier{\Crel Pu}$, and let $\alpha > \Cnorm{\Crel Pu}x$. We have
$\Cnorm P{\frac 1\alpha x+u}\leq 1$ and hence $\Cnorm P{\frac 1\alpha x}\leq
1$, that is $\Cnorm Px\leq\alpha$, so that $\Cnorm Px\leq\Cnorm{\Crel Pu}x$.

Last assume that $P$ is complete, let $(x_n)_{n\in\Nat}$ be a
non-decreasing sequence in $\Cuball{\Crel Pu}$ and let $x$ be its lub
(in $\Ccarrier P$, which exists since
$\Cnorm P{x_n}\leq\Cnorm{\Crel Pu}{x_n}\leq 1$ for each $n$). We have
that $x_n+u\in\Cuball P$ for all $n$ and hence $x+u\in\Cuball P$ by
continuity of $+$ and closeness of $P$. It is clear that $x$ is also
the lub of the $x_n$'s in $\Crel Pu$.
\end{proof}

\paragraph{Theorem~\ref{th:non-decreasing-class-equiv}.}
  A function $f:\Cuball P\to\Ccarrier Q$ is $n$-non-decreasing iff it is
  $n$-pre-stable. 
\begin{proof}
Let us first prove the left to right implication, by induction on $n$.

For $n=0$, both notions coincide with the fact of being non-decreasing.

Let now $n$ be $>0$. Let $f:\Cuball P\to\Ccarrier Q$ be
$n$-non-decreasing from $P$ to $Q$ and let us prove that $f$ is
$n$-pre-stable. Due to our inductive hypothesis, we just have to prove
that, for all $\Vect u\in P^n$ such that $\sum_{i=1}^nu_i\in\Cuball
P$ and all $x\in\Cuballr P{\Vect u}$, we have
$\Fdiffp-{f}x{\Vect u}\leq\Fdiffp+{f}x{\Vect u}$. Let $u=u_n$ and let
$\Vect v=(\List u1{n-1})$.

We know that $f$ is non-decreasing and that the function $\Fdiff f\_u$
is $n-1$-non-decreasing from $\Crel Pu$ to $Q$. Therefore, by inductive
hypothesis, we know that this function is $n-1$-pre-stable. This means in
particular that
\begin{align*}
  \Fdiffp-{(\Fdiff f\_u)}x{\Vect v}\leq \Fdiffp+{(\Fdiff f\_u)}x{\Vect v}
\end{align*}
that is
\begin{align*}
  \sum_{I\in\Cocard-{n-1}}\bigg(f(x+u+\sum_{i\in I}v_i)
  &-f(x+\sum_{i\in I}v_i)\bigg)\\
  &\leq
  \sum_{I\in\Cocard+{n-1}}\bigg(f(x+u+\sum_{i\in I}v_i)
  -f(x+\sum_{i\in I}v_i)\bigg)  
\end{align*}
and hence
\begin{align*}
  \sum_{I\in\Cocard-{n-1}}f(x+u+\sum_{i\in I}v_i)
  &+\sum_{I\in\Cocard+{n-1}}f(x+\sum_{i\in I}v_i)\\
  &\leq\sum_{I\in\Cocard+{n-1}}f(x+u+\sum_{i\in I}v_i)
    +\sum_{I\in\Cocard-{n-1}}f(x+\sum_{i\in I}v_i)
\end{align*}
Observe that the left hand expression is equal to
\begin{align*}
  \sum_{\Biind{J\in\Cocard-{n}}{n\in J}}f(x+\sum_{j\in J}u_j)
  +\sum_{\Biind{J\in\Cocard-{n}}{n\notin J}}f(x+\sum_{j\in
  J}u_j)=\Fdiffp-fx{\Vect u}
\end{align*}
and similarly the right hand expression is equal to $\Fdiffp+fx{\Vect u}$, so
we have $\Fdiffp-fx{\Vect u}\leq\Fdiffp+fx{\Vect u}$ as contended.

We prove now the right to left implication, by induction on $n$. For
$n=0$, this is obvious. So assume that $f$ is $n$-pre-stable and let us
prove that it is $n$-non-decreasing. First, $f$ is non-decreasing because
it is $0$-pre-stable. Let $u\in\Cuball P$ and let us prove that the
function $\Fdiff f\_u$ is  $n-1$-non-decreasing. To this end, by
inductive hypothesis, it suffices to prove that this function is
$n-1$-pre-stable. Let $x\in\Cuball P$ and $\Vect u\in\Cuball P^{n-1}$ be
such that $x+u+\sum_{i=1}^{n-1}u_i\in\Cuball P$, we must prove that
\begin{align*}
  \Fdiffp-{(\Fdiff f\_u)}{x}{\Vect u}\leq \Fdiffp+{(\Fdiff f\_u)}{x}{\Vect u}
\end{align*}
which by the same calculation as above amounts to showing that
$\Fdiffp-fx{\Vect u,u}\leq\Fdiffp+fx{\Vect u,u}$, and we know that this latter
holds by our assumption that $f$ is $n$-pre-stable.
\end{proof}

\paragraph{Lemma~\ref{lemma:fdiff-iter}}   
  Let $f$ be an absolutely monotonic function from $P$ to $Q$ (so
  that $f:\Cuball P\to\Ccarrier Q$). Let $n\in\Nat$,
  $\Vect u\in\Cuball P^n$
  with $\sum_{i=1}^nu_i\in\Cuball P $and $x\in\Cuballr P{\Vect
    u}$.
  Let $\List f0n$ be the functions defined by $f_0(x)=f(x)$ and
  $f_{i+1}(x)=\Fdiff{f_i}{x}{u_{i+1}}$. Then
  \begin{align*}
    f_n(x)=\Fdiffp+fx{\Vect u}-\Fdiffp-fx{\Vect u}\,.
  \end{align*}
  We set $\Fdiff fx{\Vect u}=f_n(x)$. The operation $\Delta$ is linear
  in the function:
  $\Fdiff{(\sum_{j=1}^p\alpha_jg_j)}{x}{\Vect
    u}=\sum_{j=1}^p\alpha_j\Fdiff{g_j}{x}{\Vect u}$
  for $\List g1p$ absolutely monotonic from $P$ to  $Q$.

For proving this lemma we need the following auxiliary result:

\begin{lemma}\label{lemma:fdiffp-induction}
  Let $f:\Cuball P\to\Ccarrier Q$, $x,u\in\Cuball P$ and
  $\Vect u\in\Cuball P^n$ be such that
  $u+\sum_{i=1}^nu_i\in\Cuball P$,
  $x\in\Cuballr P{\Vect u,u}$.
  We have
  \begin{align*}
    \Fdiffp+fx{\Vect u,u} &= \Fdiffp+f{x+u}{\Vect u}+\Fdiffp-fx{\Vect u}\\
    \Fdiffp-fx{\Vect u,u} &= \Fdiffp-f{x+u}{\Vect u}+\Fdiffp+fx{\Vect u}\,.
  \end{align*}
\end{lemma}
\begin{proof}
Let $\Vect v=(\Vect u,u)\in\Cuball P^{n+1}$. For $\epsilon\in\{+,-\}$, we have
  \begin{align*}
    \Fdiffp\epsilon f{x}{\Vect u,u}
    &=\sum_{I\in\Cocard\epsilon{n+1}}f(x+\sum_{i\in I}v_i)\\
    &=\sum_{\Biind{I\in\Cocard\epsilon{n+1}}{n+1\in I}}f(x+\sum_{i\in I}v_i)
      +\sum_{\Biind{I\in\Cocard\epsilon{n+1}}{n+1\notin I}}f(x+\sum_{i\in
      I}v_i)\\ 
    &=\sum_{J\in\Cocard\epsilon{n}}f(x+u+\sum_{j\in J}u_j)
      +\sum_{J\in\Cocard{-\epsilon}{n}}f(x+\sum_{j\in J}u_j)\\
    &=\Fdiffp\epsilon f{x+u}{\Vect u}+\Fdiffp{-\epsilon}f{x}{\Vect u}\,.
  \end{align*}
  where $-\epsilon$ is the sign opposite to $\epsilon$. 
\end{proof}

\begin{proof}[Proof of Lemma~\ref{lemma:fdiff-iter}]
  The proof is by induction on $n\in\Nat$. For $n=0$ the equation
  holds trivially. Assume that the property holds for $n$ and let us
  prove it for $n+1$. Let $\Vect v=(\List u1{n})$ and $u=u_{n+1}$. we
  have
\begin{align*}
  f_{n+1}(x) &= f_n(x+u)-f_n(x)\\
             &= \Fdiffp+f{x+u}{\Vect v}-\Fdiffp-f{x+u}{\Vect v}
               - (\Fdiffp+f{x}{\Vect v}-\Fdiffp-f{x}{\Vect v})\\
             &\hspace{2cm}\text{by inductive hypothesis}\\
             &=\Fdiffp+f{x}{\Vect v,u}-\Fdiffp-f{x}{\Vect v}
               -\Fdiffp-fx{\Vect v,u}+\Fdiffp+fx{\Vect v}\\
             &\quad-\Fdiffp+f{x}{\Vect v}+\Fdiffp-f{x}{\Vect v}
               \text{ by Lemma~\ref{lemma:fdiffp-induction}}\\
             &=\Fdiffp+f{x}{\Vect v,u}-\Fdiffp-fx{\Vect v,u}
\end{align*}
as contended. The linearity statment is an easy consequence.
\end{proof}

\paragraph{Lemma~\ref{lemma:fdiff-iter-pre-stable}.}
  Let $f:\Cuball P\to\Ccarrier Q$ be a pre-stable function from $P$ to $Q$.
  For all $\Vect u\in\Cuball P^n$, the functions
  $\Fdiffp- f\_{\Vect u}$, $\Fdiffp+ f\_{\Vect u}$ and
  $\Fdiff f\_{\Vect u}$ are pre-stable from $\Crel P{\Vect u}$ to $Q$.
\begin{proof}
  For $\Fdiff f\_{\Vect u}$, this is an immediate consequence of
  Theorem~\ref{th:non-decreasing-class-equiv} and of the definition of
  an $n$-non-decreasing function. For $\Fdiffp\epsilon f\_{\Vect u}$,
  it results from the fact that pre-stable functions are closed under
  addition and from the fact that, for all $u\in\Cuball P$, the
  function $x\mapsto f(x+u)$ is pre-stable from $\Crel Pu$ to $Q$.
\end{proof}




\paragraph{Lemma~\ref{lemma:fdiff-sommes1}.}
  Let $f:\Cuball P\to\Ccarrier Q$ be a pre-stable function from $P$ to $Q$. Let
  $n\in\Nat$, $x,u,v\in\Cuball P$ and $\Vect u\in\Cuball P^n$, and
  assume that $x+u+v+\sum_{i=1}^nu_i\in\Cuball P$. Then
\begin{align*}
  \Fdiff f{x+u}{\Vect u}&=\Fdiff fx{\Vect u}+\Fdiff fx{u,\Vect u}\\
  \Fdiff fx{u+v,\Vect u}&=\Fdiff fx{u,\Vect u}+\Fdiff f{x+u}{v,\Vect u}\,.
\end{align*}
\begin{proof}
  The equations clearly hold for $n=0$: $f(x+u)=f(x)+\Fdiff fxu$ and
  $\Fdiff fx{u+v}=\Fdiff fxu+\Fdiff f{x+u}v$. The general case follows
  by applying these two latter equations to the function
  $g_\epsilon=\Fdiffp\epsilon f{\_}{\Vect u}$ for $\epsilon\in\{+,-\}$
  as we show now (the function $g_\epsilon$ is pre-stable by
  Lemma~\ref{lemma:fdiff-iter-pre-stable}).  For the first
  equation
   we have
  $g_\epsilon(x+u)=g_\epsilon(x)+\Fdiff{g_\epsilon}xu$ and remember that $\Fdiff f{\_}{\Vect u}=g_+-g_-$. Therefore we have 
\begin{align*}
  \Fdiff fx{\Vect u}+\Fdiff fx{u,\Vect u}
  &=(g_+-g_-)(x)+\Fdiff{(g_+-g_-)}xu\\
    &=g_+(x)-g_-(x)+\Fdiff{g_+}xu-\Fdiff{g_-}xu\\
    &=g_+(x+u)-g_-(x+u)\\
&=\Fdiff f{x+u}{\Vect u}
\end{align*}
using Lemma~\ref{lemma:fdiff-iter}. For the second equation we have similarly
\begin{align*}
  \Fdiff{f}{x}{u,\Vect u}+\Fdiff{f}{x+u}{v,\Vect u}
  &= \Fdiff{(g_+-g_-)}{x}{u}+\Fdiff{(g_+-g_-)}{x+u}{v}\\
  &= \Fdiff{g_+}{x}{u}+\Fdiff{g_+}{x+u}{v}
    -\Fdiff{g_-}{x}{u}-\Fdiff{g_-}{x+u}{v}\\
  &= \Fdiff{g_+}x{u+v}-\Fdiff{g_-}x{u+v}\\
  &= \Fdiff{(g_+-g_-)}{x}{u+v}\\
  &= \Fdiff fx{u+v,\Vect u}
\end{align*}
\end{proof}

\paragraph{Lemma~\ref{lemma:fdiff-pre-stable-gen}.}   Let $f:\Cuball P\to\Ccarrier Q$ be a pre-stable function from $P$ to
  $Q$. Then the map $g:\Cuball{\Scone pP}\to\Ccarrier Q$ defined by
  $g(x,\Vect u)=\Fdiff fx{\Vect u}$ is non-decreasing, for all
  $p\in\Natnz$.
\begin{proof}

  We have $g(x,\Vect u)=\Fdiff fx{\List u1{p}}$. It suffices to prove
  that this function is non-decreasing wrt.~all parameters
  separately. Wrt.~$x$, it results from
  Lemma~\ref{lemma:fdiff-iter-pre-stable}. Wrt.~$u_i$, it results from the
  fact that
\begin{align*}
  g(x,\Vect u)=\Fdiff f{x+u_i}{\Vect v}-\Fdiff f{x}{\Vect v}
\end{align*}
where $\Vect v=(u_1,\dots,u_{i-1},u_{i+1},\dots,u_p)$ and from the
fact that $\Fdiff fy{\Vect v}$ is non-decreasing wrt.~$y$,
which results from Lemma~\ref{lemma:fdiff-iter-pre-stable}.
\end{proof}

\subsection{Proofs of Section~\ref{sect:measurability}}

\paragraph{Lemma~\ref{lemma:mes-path-cst}.}   For any $x\in P$ and $n\in\Nat$, the function $\gamma:\Real^n\to P$
  defined by $\gamma(\Vect r)=x$ belongs to $\Mpath nP$.  If
  $\gamma\in\Mpath nP$ and $h:\Real^p\to\Real^n$ is measurable then
  $\gamma\circ h\in\Mpath pP$.
\begin{proof}
  The fact that all constant functions are measurable paths results
  from the last condition on measurability tests.  For closure under
  precomposition with measurable functions, observe that
  $\Mtestapp l{(\gamma\Comp h)}=(\Mtestapp l\gamma)\Comp(\Id\times h)$
  is measurable because $\Mtestapp l\gamma$ is.
  \todott{Ok. J'utilise $l$ pour les tests de $P$ et $m$ pour les tests
  de $Q$. Ok?}
\end{proof}

\subsection{Proofs of Section~\ref{sect:ccc}}

\paragraph{Lemma~\ref{lemma:sep-lin-cont}.}   Let $f:\Ccarrierb P\times\Cuball Q\to \Ccarrierb R$ be a function such that
  \begin{itemize}
  \item for each $y\in\Cuball Q$, the function
    $f^{(1)}_y:\Ccarrierb P\to\Ccarrierb R$
    defined by $f^{(1)}_y(x)=f(x,y)$ is
    linear (resp.~linear and Scott-continuous)
  \item and for each $x\in\Ccarrierb P$, the function
    $f^{(2)}_x:\Cuball Q\to\Ccarrierb R$ defined by $f^{(2)}_x(y)=f(x,y)$ is
    pre-stable (resp.~pre-stable and Scott-continuous).
  \end{itemize}
  Then the restriction $f:\Cuball P\times\Cuball Q\to\Ccarrierb R$ is
  pre-stable (resp.~pre-stable and Scott-continuous, that is, stable)
  from $P\times Q$ to $R$.
\begin{proof}
Let $n\in\Nat$, $\Vect u\in\Cuball P^n$, $\Vect v\in\Cuball Q^n$. We define
$\Vect w=((u_1,v_1),\dots,(u_n,v_n))\in\Cuball{(P\times Q)}^n$. Let
$(x,y)\in\Cuball{\Crel{(P\times Q)}{\Vect w}}=\Cuball{\Crel P{\Vect
    u}}\times\Cuball{\Crel Q{\Vect v}}$. We must prove that
\begin{align*}
  \Fdiffp-f{(x,y)}{\Vect w}\leq\Fdiffp+f{(x,y)}{\Vect w}\,.
\end{align*}
For $\epsilon\in\{+,-\}$, we have
\begin{align*}
  \Fdiffp\epsilon f{(x,y)}{\Vect w}
  &=\sum_{I\in\Cocard\epsilon n}f(x+\sum_{i\in I}u_i,y+\sum_{i\in I}v_i)\\
  &=\sum_{I\in\Cocard\epsilon n}f(x,y+\sum_{i\in I}v_i)
    +\sum_{i=1}^n\sum_{\Biind{I\in\Cocard\epsilon n}{i\in I}}f(u_i,y+\sum_{j\in
    I}v_j)\quad\text{by linearity on the left}\\
  &=\sum_{I\in\Cocard\epsilon n}f(x,y+\sum_{i\in I}v_i)
    +\sum_{i=1}^n\sum_{J\in\Cocard\epsilon{n-1}}
    f(u_i,y+v_i+\sum_{j\in J}v{(i)}_j)
\end{align*}
where $\Vect{v{(i)}}\in\Cuball Q^{n-1}$ is defined by $v{(i)}_j=
\begin{cases}
  v_j&\text{if }j<i\\
  v_{j+1}&\text{if }j\geq i
\end{cases}
$.

So we have proven that
\begin{align*}
  \Fdiffp\epsilon f{(x,y)}{\Vect w}=
  \Fdiffp\epsilon{f^{(2)}_x}y{\Vect v}+\sum_{i=1}^n\Fdiffp\epsilon{f^{(2)}_{u_i}}{y+v_i}{\Vect{v^{(i)}}}
\end{align*}
and we obtain the required inequation because $f^{(2)}_x$ as well as each of
the function $f^{(2)}_{u_1}$,\dots, $f^{(2)}_{u_n}$ is pre-stable.

For the ``continuity'' part of the statement, let $(w_n)_{n\in\Nat}$ be an
non-decreasing sequence in $\Cuball P\times\Cuball Q$, with $w_n=(u_n,v_n)$ where
$(u_n)_{n\in\Nat}$ and $(v_n)_{n\in\Nat}$ are non-decreasing sequences in
$\Cuball P$ and $\Cuball Q$ respectively. Then
$f(\sup_{i\in\Nat}u_i,\sup_{i\in\Nat}v_i)=\sup_{i,j\in\Nat}f(u_i,v_j)$ by
separate Scott-continuity of $f$. By montonicity of $f$ we get
$f(\sup_{i\in\Nat}u_i,\sup_{i\in\Nat}v_i)=\sup_{i\in \Nat}f(u_i,v_i)$.
\end{proof}

\paragraph{Lemma~\ref{lemma:fun-order}.}   Let $f,g\in\Ccarrierb{\Implm PQ}$. We have $f\leq g$ in $\Implm PQ$
  iff the following condition holds
  \begin{align*}
    \forall n\in\Nat\,\forall\Vect u\in P^n\
    \sum_{i=1}^nu_i\in\Cuball P\Implies\forall
    x\in\Cuballr P{\Vect u}\ \Fdiffp+fx{\Vect u}+\Fdiffp-gx{\Vect
    u}\leq\Fdiffp+gx{\Vect u}+\Fdiffp-fx{\Vect u}\,.
  \end{align*}
\begin{proof}
  Indeed, $f\leq g$ means that there is $h\in\Ccarrierb{\Implm PQ}$ such that
  $g=f+h$, but then we must have $f(x)\leq g(x)$ for all $x$ (which is just the
  condition above for $n=0$) and $h$ is given pointwise by
  $h(x)=g(x)-f(x)$. The condition above 
  coincides with pre-stability of $h$. One concludes the proof by observing
  that when $h$ so defined is non-decreasing, it is automatically
  Scott-continuous and measurable. The second property readily follows from the
  linearity of measurability tests (linear maps commute with subtraction) and
  from the closure properties of measurable functions so let us check the first
  one. Let $(x_n)_{n\in\Nat}$ be a non-decreasing sequence in $\Cuball P$ and
  let $x$ be its lub. Because $h$ is non-decreasing, it is sufficient to prove
  that $h(x)\leq\sup_{n\in\Nat}h(x_n)$, that is
  $g(x)\leq f(x)+\sup_{n\in\Nat}h(x_n)$.  By Scott-continuity of $g$ and by the
  fact that $f$ is non-decreasing, it suffices to prove that, for each
  $k\in\Nat$, $g(x_k)\leq f(x_k)+\sup_{n\in\Nat}h(x_n)$ which is clear since
  $g(x_k)=f(x_k)+h(x_k)\leq f(x_k)+\sup_{n\in\Nat}h(x_n)$.
\end{proof}

\paragraph{Lemma~\ref{lemma:function_cone_complete}.}   The cone $\Implm PQ$ is complete and the lubs in
  $\Cuballp{\Implm PQ}$ are computed pointwise.

\begin{proof}
Let $(f_n)_{n\in\Nat}$ be a non-decreasing sequence in
$\Cuballp{\Implm PQ}$. For any $x\in\Cuball P$ the sequence
$(f_n(x))_{n\in\Nat}$ is non-decreasing in $\Cuball Q$ and we set
$f(x)=\sup_{n\in\Nat}f_n(x)$. Since each $f_n$ is non-decreasing and
Scott-continuous, $f$ has the same properties. To prove that $f$ is
pre-stable, observe that
\(
  \Fdiffp\epsilon fx{\Vect u}=\sup_{n\in\Nat}\Fdiffp\epsilon{f_n}{x}{\Vect u}
\)
by Scott-continuity of $+$ in $Q$. So far we have proven that
$f$ is Scott-continuous and pre-stable. Let us check that $f$ is measurable: let
$\gamma\in\Mpathu nP$, we must prove that $f\Comp\gamma\in\Mpath nQ$. Let
$m\in\Mtest kQ$, we must prove that the function
$h=\Mtestapp m{(f\Comp\gamma)}:\Real^{k+n}\to\Realp$ is measurable. By Scott
continuity of the linear function $m(\Vect r)$, we have $h(\Vect r,\Vect
s)=\sup_{n\in\Nat}h_n(\Vect r,\Vect s)$ where $h_n=\Mtestapp
m{(f_n\Comp\gamma)}$ and conclude that $h$ is measurable by the monotone
convergence theorem. So we have proven that $f\in\Cuballp{\Implm PQ}$. 

Let $n\in\Nat$ and let us prove that $f_n\leq f$. By
Lemma~\ref{lemma:fun-order} it suffices to prove (with the usual
assumptions) that
\begin{align*}
  \Fdiffp+{f_n}x{\Vect u}+\Fdiffp-fx{\Vect u}
  \leq\Fdiffp+{f}x{\Vect u}+\Fdiffp-{f_n}x{\Vect u}\
\end{align*}
which results from the Scott-continuity of $+$, from the fact that
$\Fdiffp\epsilon fx{\Vect u}=\sup_{k\geq
  n}\Fdiffp\epsilon{f_k}{x}{\Vect u}$
and from the fact that the sequence $(f_k)_{k\geq n}$ is
non-decreasing. Last let $g\in\Cuball{(\Implm PQ)}$ be such that
$f_n\leq g$ for all $n$, we must prove that $f\leq g$. Again we apply
straightforwardly Lemma~\ref{lemma:fun-order} and the Scott-continuity
of $+$.
\end{proof}

\paragraph{Lemma~\ref{lemma:stab-eval}.}   The evaluation function $\Ev$ is stable and measurable, that is
  $\Ev\in\INFSCOTTM((\Implm PQ)\times P,Q)$.
\begin{proof}
  Stability results from
  Lemma~\ref{lemma:sep-lin-cont}, observing that $\Ev$ is linear and
  Scott-continuous in its first argument.
  We must prove now that $\Ev$ is measurable. Let $n\in\Nat$,
  $\phi\in\Mpathu n{\Implm PQ}$ and $\gamma\in\Mpathu nP$. We must prove that
  $\Ev\Comp\Pair{\phi}{\gamma}\in\Mpath nQ$. So let $q\in\Nat$ and
  $m\in\Mtest qQ$. Given $(\Vect r,\Vect s)\in\Real^{q+n}$, we have
\begin{align*}
  \Mtestapp{m}{(\Ev\Comp\Pair\phi\gamma)}(\Vect r,\Vect s)
  &= m(\Vect r)(\phi(\Vect s)(\gamma(\Vect s)))\\
  &= (\Mtestapp{(\Mtestfun{\gamma\Comp\pi_1}{m\Comp\pi_2})}{\phi})
    ((\Vect s,\Vect r),\Vect s)
\end{align*}
where $\pi_1:\Real^{n+q}\to\Real^n$ and $\pi_2:\Real^{n+q}\to\Real^q$ are the
projections, which are measurable functions. We know that
$\gamma\Comp\pi_1\in\Mpath{n+q}P$ and $m\Comp\pi_2\in\Mtest{n+q}{\Implm PQ}$
hence $\Mtestfun{\gamma\Comp\pi_1}{m\Comp\pi_2}\in\Mtest{n+q}{\Implm PQ}$ and
therefore
$\Mtestapp{(\Mtestfun{\gamma\Comp\pi_1}{m\Comp\pi_2})}{\phi}\in\Mfun{n+q+n}$
because we know that $\phi\in\Mpath n{\Implm PQ}$. It follows that
$\Mtestapp{m}{(\Ev\Comp\Pair\phi\gamma)}\in\Mfun{q+n}$ because the function
$\Real^{q+n}\to\Real^{n+q+n}$ defined by $(\Vect r,\Vect s)\to((\Vect s,\Vect
r),\Vect s))$ is measurable.
\end{proof}

\paragraph{Lemma~\ref{lemma:int-lin-cont-mes}.}   Let $X$ be a measurable space and let $f:X\to\Realp$ be a measurable and
  bounded function. Then the function $F:\Bmes X\to\Realp$ defined by
  \begin{align*}
    F(\mu)=\int f(x)\mu(dx)
  \end{align*}
  is linear and Scott-continuous.
\begin{proof}
  The proof is straightforward when $f$ is simple. Then one chooses a
  non-decreasing sequence of simple measurable functions
  $f_n:X\to\Realp$ which converges simply to $f$, that is
  $f(x)=\sup_{n\in\Nat}f_n(x)$. We have
\begin{align*}
  F(\mu)=\sup_{n\in\Nat}\int f_n(x)\mu(dx)
\end{align*}
from which the statement follows.
\end{proof}

\paragraph{Lemma~\ref{lemma:inf-cont-to-mes}.}   Let $Q$ be a cone and $X$ be a measurable space. A function
  $f:\Cuball Q\to\Bmes X$ is stable iff for all $U\in\Tribu X$, the
  function $f_U:\Cuball Q\to\Realp$ defined by $f_U(y)=f(y)(U)$ is stable.
\begin{proof}
  The condition is necessary because, for each $U\in\Tribu X$, the
  function $e_U:\mu\mapsto\mu(U)$ is linear and Scott-continuous (and
  hence stable) from $\Bmes X$ to $\Realp$. Conversely let us assume
  that $f_U$ is stable for each $U\in\Tribu X$. We prove that $f$ is
  pre-stable. Let $\Vect v\in\Cuball Q^n$ be such that
  $\sum_{i=1}^nv_i\in\Cuball Q$ and let $y\in\Cuballr Q{\Vect v}$. We
  must prove that $\Fdiffp-fy{\Vect v}\leq\Fdiffp+fy{\Vect v}$ in
  $\Bmes X$, that is, we must prove that
  $e_U(\Fdiffp-fy{\Vect v})\leq e_U(\Fdiffp+fy{\Vect v})$ in $\Realp$,
  for each $U\in\Tribu X$. This results from our assumption and from
  the fact that
  $e_U(\Fdiffp\epsilon fy{\Vect v})=\Fdiffp\epsilon{f_U}y{\Vect
    v}$. Scott-continuity of $f$ is proven similarly.
\end{proof}

\subsection{Proofs of Section~\ref{sect:sound_adeq}}

\paragraph{Lemma~\ref{lemma:interp_kernel}.}
Let $\Gamma\vdash\terma:\typea$, with $\typea=\typeb_1\rightarrow\dots\typeb_k\rightarrow\treal$. For all $i\leq k$, let $b_i\in\semtype{\typeb_i}$ and $\vec g\in\semtype{\Gamma}$, then the map $\terma\mapsto\semterm[\Gamma\vdash\typea]{\terma}\vec gb_1\dots b_k$ is a stochastic kernel from $\Terms[\Gamma\vdash\typea]$ to $\real$.
\todom{I changed the proof of this lemma (old proofs with thomas's issues commented), is it ok?}

\begin{proof}
Let us write $\semterm{\terma}\vec g\vec b$ for $\semterm[\Gamma\vdash\typea]{\terma}\vec gb_1\dots b_k$.
 Since $\Terms[\Gamma\vdash\typea]$ is the coproduct \eqref{eq:sigma_algebra_terms}, it is enough to prove that, for any $n$ and any $\fterma\in\TermsFree[\Gamma\vdash\typea]n$, the restriction ${\semterm{\_}}_\fterma\vec g\vec b$
 of ${\semterm{\_}}\vec g\vec b$ to $\TermsFree[\Gamma\vdash\typea]\fterma$ is a kernel. 
 
 For every $\terma\in\TermsFree[\Gamma\vdash\typea]\fterma$,
 ${\semterm{\terma}}_\fterma\vec g\vec b={\semterm{\fterma\vec r}}_\fterma\vec
 g\vec b$,
 for a suitable $\vec r\in\Real^n$. By the substitution property
 (Lemma~\ref{lemma:substitution}), we have:
 ${\semterm{\fterma\vec r}}_\fterma\vec g\vec
 b=\semterm[\varreal_1:\treal,\dots,\varreal_n:\treal,\Gamma\vdash\typea]{\fterma}(\Dirac{r_1},\dots,\Dirac{r_n},\vec
 g)b_1\dots b_k$. 
 This latter being equal to $h(\vec r)\vec b$, for 
 $h=\semterm{\fterma}\circ(\Diracf^n\times\vec g)$ a map from $\real^n$ to $\semtype\typea$.  Notice that $h\in\Mpathu n{\semtype\typea}$, since $\Diracf^n\times\vec g\in\Mpathu n{\Bmes\real^n\times\semtype\Gamma}$ by Lemma~\ref{lemma:mes-path-cst} and the fact that $\Diracf\in \Mpathu 1{\Bmes\real}$. This implies that the map $\vec r\mapsto h(\vec r)\vec b$ is in $\Mpathu n{\Bmes\real}$\todom{should I argue more here?}, so it is a stochastic 
kernel from $\real^n$ to $\real$ (Example~\ref{ex:meas_cone_with_meas}). We have then the statement because $\Real^n$ and $\TermsFree[\Gamma\vdash\typea]\fterma$  are isomorphic as measurable spaces.
\end{proof}

\paragraph{Proposition~\ref{prop:soundmess}.}
Let $\typea=\typeb_1\rightarrow\dots\typeb_k\rightarrow\treal$, for all $i\leq k$, $b_i\in\semtype{\typeb_i}$, and let $\vec g\in\semtype{\Gamma}$, then $(\semterm[\Gamma\vdash\typea]{\_}\vec gb_1\dots b_k) \circ \Red=\semterm[\Gamma\vdash\typea]{\_}\vec gb_1\dots b_k$, i.e.~Equation~\eqref{eq:soundness_explicit} holds for any $\terma\in\Terms[\Gamma\vdash\typea]$.
\begin{proof}
If $\terma$ is a normal form, then the statement is trivial, as $\Red(\terma,\_)=\Dirac\terma$ which is the identity in $\Kern$. Otherwise, let $\terma=E[R]$ with $R$ a redex  (Lemma~\ref{lemma:evaluation_decomposition}). 

If $R\neq\text{\oracle}$, let $R\rightarrow \termb$, so $\Red(E[R],\_)=\Dirac{E[\termb]}$, and $\left((\semterm[\Gamma\vdash\typea]{\_}\vec gb_1\dots b_k) \circ \Red\right)(E[R])=\int_{\Terms[\Gamma\vdash\typea]}\semterm[\Gamma\vdash\typea]{t}\Dirac{E[\termb]}(dt)=\semterm[\Gamma\vdash\typea]{E[\termb]}\vec gb_1\dots b_k$. By the substitution property (Lemma~\ref{lemma:substitution})) it is sufficient to prove $\semterm R=\semterm N$ to conclude. This is done by cases, depending on the type of the redex. The cases $R$ is a $\beta$ or $\fix$ redex follow the standard reasoning proving the soundness of  a cpo-enriched cartesian closed category. 

In case $R=\ifz{\num 0}\termc\termb$ then, by applying the definition in Figure~\ref{fig:interpretation}, $\semterm{R}\vec g=(\semterm{\num 0}\vec g\{0\})\semterm{\termb}\vec g+(\semterm{\num 0}\vec g(\real\!\setminus\!\{0\}))\semterm{\termc}\vec g=\semterm{\termb}\vec g$. The case for  a numeral different from $\num 0$ is analogous.

In case $R=\num f(\num{r_1},\dots,\num{r_n})$ and so $\termb=\num{f(r_1,\dots,r_n)}$ we can conclude since $\semterm[\Gamma\vdash\treal]R\vec g= (\Dirac{r_1}\otimes\dots\otimes\Dirac{r_n})\circ f^{-1}= \Dirac{f(r_1,\dots,r_n)}=\semterm[\Gamma\vdash\treal]\termb\vec g$. 

In case $R=\letterm{x}{\num r}{\termc}$ and so $\termb=\termc\{\num r/x\}$, we have: $\semterm[\Gamma\vdash\treal]R\vec g=\int_\Real(\semterm[\Gamma,x:\treal\vdash\treal]\termc\vec g\circ\Diracf)(p)(\Dirac r(dp))=\semterm[\Gamma,x:\treal\vdash\treal]\termc\vec g \Dirac r$. This latter is equal to $\semterm[\Gamma\vdash\treal]\termb\vec g$ by the substitution property. 

The last case is the sampling redex: $\terma=E[\oracle]$. Then: 
\begin{align*}
&\left((\semterm[\Gamma\vdash\typea]{\_}\vec gb_1\dots b_k) \circ \Red\right)(\terma)\\
&=\int_{\Terms[\Gamma\vdash\typea]}\semterm[\Gamma\vdash\typea]{t}\vec gb_1\dots b_k\Red(E[\oracle],dt)\\
&=\int_{\Real}\semterm[\Gamma\vdash\typea]{E[\num r]}\vec gb_1\dots b_k\lambda_{[0,1]}(dr)&\text{By definition of $\Red$}\\
&=\int_{\Real}\semterm[y:\treal,\Gamma\vdash\typea]{E[y]}(\semterm{\num r}\vec g)\vec gb_1\dots b_k\lambda_{[0,1]}(dr)&\text{By substitution (Lemma~\ref{lemma:substitution}), with $y$ fresh}\\
&=\semterm[y:\treal,\Gamma\vdash\typea]{E[y]}\left(\int_{\Real}\semterm{\num r}\vec g\lambda_{[0,1]}(dr)\right)\vec gb_1\dots b_k&\text{By linearity (Lemma~\ref{lemma:linearity_eval}) and Scott-continuity}\\
&=\semterm[y:\treal,\Gamma\vdash\typea]{E[y]}\vec g\lambda_{[0,1]}b_1\dots b_k\\
&=\semterm[\Gamma\vdash\typea]{\terma}\vec gb_1\dots b_k&\text{By substitution (Lemma~\ref{lemma:substitution})}
\end{align*}
\todot{C'est quoi ce $y$ qui apparaît à la 4ème ligne?}
\todom{explique'. Ok?}
\end{proof}

\paragraph{Lemma~\ref{lemma:key-lemma}.} Let $x_1:\typeb_1,\dots,x_n:\typeb_n\vdash\terma:\typea$ and $\forall i\leq n, u_i\prec^{\typeb_i}\termb_i$, then: $\semterm\terma\vec u\prec^{\typea}\terma\{{\vec\termb/\vec x}\}$.

The proof of this lemma uses some two auxiliary lemmata. 

\begin{lemma}\label{lemma:red_f}
Given a $k$-ary functional identifier $\num f\in\Const$ and  $\List \terma1k$ such that $\Tseq{}{M_i}{\treal}$ for
  each $i\leq k$, then we have: $(\Red^{\infty}(\terma_1,\_)\otimes\dots\otimes\Red^{\infty}(\terma_k,\_))(\num{f^{-1}(U)})\leq\Red^{\infty}(\num f(\terma_1,\dots,\terma_k),\num U)$, for every measurable subset $U$ of $\Real$. 
\end{lemma}
\begin{proof}
We prove that for all $\num f$ of arity $k$, for all  $\terma_1$,\dots,$\terma_k$ closed terms of type $\treal$, for all 
$U\in\Sigma_{\Real}$, for all $n_1,\dots,n_k\in\mathbb N$, there exists $m\in\mathbb
N$ such that:
\begin{itemize}
	\item[($\star$)] $\Red^{n_1}(\terma_1,\_)\otimes\dots\otimes\Red^{n_k}(\terma_k,\_))(\num{f^{-1}(U)})\leq \Red^{m}(\num f(\terma_1,\dots,\terma_k),\num U)$.
\end{itemize}
The statement follows by the definition of $\Red^\infty$ as a lub. The proof is by induction on $\sum_in_i$.

If $\terma_i$ is a numeral $\num{r_i}$ for every $i\leq k$, then:
$\Red^{n_1}(\terma_1,\_)\otimes\dots\otimes\Red^{n_k}(\terma_k,\_))(\num{f^{-1}(U)})=\Dirac{\num{f(r_1,\dots,r_n)}}(\num
U)=\Red(\num f(\num{r_1},\dots,\num{r_k}),\num U)$ and we are done. Otherwise
there must be one $\terma_i$ which is reducible (notice that the term $\num
f(\terma_1,\dots,\terma_k)$ is closed by hypothesis, so no $\terma_i$ can be a
variable). So let us prove ($\star$) supposing that $i$ is minimal such that
  $M_i$ is reducible.
%
If $n_i=0$, then $\Red^{n_i}(\terma_i,\_)=\Dirac{\terma_i}$ and since
$\num{f^{-1}(U)}\subseteq\num{\Real^k}$, we have that the left-hand side
expression of ($\star$) vanishes
and the equality trivially holds for any $m$. Otherwise, writing
$\mu_{n_j,\terma_j}$ for the measure $\Red^{n_j}(\terma_j,\_)$, we have:
\begin{align*}
&\mu_{n_1,\num{r_1}}\otimes\dots\otimes\mu_{n_{i-1},\num{r_{i-1}}}\otimes\mu_{n_i,\terma_i}\otimes\mu_{n_{i+1},\terma_{i+1}}\otimes\dots\otimes\mu_{n_{k},\terma_{k}}(\num{f^{-1}(U)})\\
&=\Dirac{\num{r_1}}\otimes\dots\otimes\Dirac{\num{r_{(i-1)}}}\otimes
	\left(
		\int_{\Terms[\vdash\treal]}\mu_{n_i-1,t}\Red(\terma_i,dt)
	\right)
\otimes\mu_{n_{i+1},\terma_{i+1}}\otimes\dots\otimes\mu_{n_{k},\terma_{k}}(\num{f^{-1}(U)})\\
&=\int_{\Terms[\vdash\treal]}
			\left(
			\Dirac{\num{r_1}}\otimes\dots\otimes\Dirac{\num{r_{(i-1)}}}
			\otimes\mu_{n_i-1,t}\otimes\mu_{n_{i+1},\terma_{i+1}}\otimes\dots\otimes\mu_{n_{k},\terma_{k}}(\num{f^{-1}(U)})
	\right)\Red(\terma_i,dt)\\
&\leq\int_{\Terms[\vdash\treal]}
	\Red^m(\num f(\num{r_1},\dots,\num{r_{i-1}},t,\terma_{i+1},\dots,\terma_{k}), \num U)\Red(\terma_i,dt)\\
&=\Red^{m+1}(\num f(\num{r_1},\dots,\num{r_{i-1}},\terma_i,\terma_{i+1},\dots,\terma_{k}), \num U)
\end{align*}
\noindent where the inequality between the third and fourth lines is an application of the induction hypothesis. 
\end{proof}
\begin{lemma}\label{lemma:red_if}
Let $\termc$, $\terma'$ and $\terma''$ be closed terms of type $\treal$, then:
\[
	(\Red^{\infty}(\termc,\{\num 0\}))\Red^{\infty}(\terma',\_)+(\Red^{\infty}(\termc,\num{\Real\setminus\{0\}}))\Red^{\infty}(\terma'', \_)\leq\Red^{\infty}(\ifz\termc{\terma'}{\terma''},\_).
\]
\end{lemma}
\begin{proof}
Similar to the proof of Lemma~\ref{lemma:red_f}. We prove that for every $\termc$, $\terma'$, $\terma''$ closed terms of type $\treal$, for every $U\in\Sigma_{\Real}$, for every $n_1,n_2,n_3\in\mathbb N$, there is $m\in\mathbb N$ such that:
\begin{itemize}
\item[($\star$)] $(\Red^{n_1}(\termc,\{\num 0\}))\Red^{n_2}(\terma',U)+(\Red^{n_1}(\termc,\num\Real\setminus\{\num 0\}))\Red^{n_3}(\terma'',U)\leq\Red^{m}(\ifz\termc{\terma'}{\terma''},U)$.
\end{itemize}
The proof is by induction on $n_1$.\todot[inline]{Bizarre
  que tu dises que c'est par induction sur $n_1$
  et que ça ne commence pas pas ``if $n_1=0$''\dots}\todom{c'est parce que dans cette facon j'evite de repeter le cas L un numeral\dots}
If $\termc=\num
0$, then the left-hand side expression in
($\star$)
is equal to
$\Red^{n_2}(\terma',U)=\Red^{n_2+1}(\ifz\termc{\terma'}{\terma''},U)$
and we are done. The case $\termc$
is a numeral different from $\num 0$ is symmetric.

Let us then suppose $\termc$ reducible. If $n_1=0$ then the inequality trivially holds because the left-hand side expression in ($\star$) is zero. If $n_1>0$, then 
$\Red^{n_1}(\termc,\_)=\int_{\Terms[\vdash\treal]}\Red^{n_1-1}(t,\_)\Red(\termc,dt)$, hence the left-hand side expression in ($\star$) is equal to:
\begin{align*}
&\int_{\Terms[\vdash\treal]}\Red^{n_1-1}(t,\{\num 0\})\Red^{n_2}(\terma',U)\Red(\termc,dt)+\int_{\Terms[\vdash\treal]}\Red^{n_1-1}(t,\num\Real\setminus\{\num 0\})\Red^{n_2}(\terma'',U)\Red(\termc,dt)\\
&=\int_{{\Terms[\vdash\treal]}}\left(\Red^{n_1-1}(t,\{\num 0\})\Red^{n_2}(\terma',U)+\Red^{n_1-1}(t,\num\Real\setminus\{\num 0\})\Red^{n_2}(\terma'',U)\right)\Red(\termc,dt)\\
&\leq\int_{{\Terms[\vdash\treal]}}\Red^{m}(\ifz t{\terma'}{\terma''}\Red(\termc,dt)=\Red^{m+1}(\ifz \termc{\terma'}{\terma''}
\end{align*}
\end{proof}

\begin{lemma}\label{lemma:let}
Given $\vdash\terma':\treal$ and $x:\treal\vdash\terma'':\treal$, we have:
\[
	\int_{\num\Real}\Red^{\infty}(\terma''\{t/x\},\_)\Red^{\infty}(\terma',dt)\leq\Red^{\infty}(\letterm x{\terma'}{\terma''},\_),
\]
\noindent where recall that $\num\Real$ is the set of all numerals, which is a sub-measurable space of $\Terms[\vdash\treal]$ isomorphic to $\Real$.
\end{lemma}
\begin{proof}
First of all, notice that the integral is meaningful because $\Red^{\infty}(\terma''\{t/x\},\_)$ is the stochastic kernel resulting from the composition of $\Red^{\infty}$ and $\subst{x,\terma''}$, this latter being a measurable function by Lemma~\ref{lemma:subst_measurable}\todom{here we are using that composing kernel with measurable functions gives kernel (to check and state in the appendix)}. Then the proof follows the reasoning of the proof of Lemma~\ref{lemma:red_f}. 

We prove that for every $\vdash \terma':\treal$ and $x:\treal\vdash \terma'':\treal$, for every $U\in\Sigma_{\Real}$, $n_1,n_2\in\mathbb N$, there exists $m\in\mathbb N$ such that:
\begin{itemize}
\item[($\star$)] $\int_{\num\Real}\Red^{n_2}(\terma''\{t/x\},\num U)\Red^{n_1}(\terma',dt)\leq\Red^{m}(\letterm x{\terma'}{\terma''},\num U)$.
\end{itemize}
The proof is by induction on $n_1$. If $\terma'$  is a numeral $\num r$, then the left-hand side expression in ($\star$) is equal to $\Red^{n_2}(\terma'\{\num r/x\},U)=\Red^{n_2+1}(\letterm x{\terma'}{\terma''},U)$ and we are done. So let us suppose that $\terma'$ is reducible. Under this hypothesis, if $n_1=0$, the left-hand side expression of ($\star$) is zero and so the inequality holds. Otherwise:
\begin{align*}
&\int_{\num\Real}\Red^{n_2}(\terma''\{t/x\},\num U)\Red^{n_1}(\terma',dt)\\
&=\int_{\num\Real}\Red^{n_2}(\terma''\{t/x\},\num U)\left(\int_{\Terms[\vdash\treal]}\Red^{n_1-1}(u,dt)\Red(\terma',du)\right)&\text{by def $\Red^{n_1}$}\\
&=\int_{\Terms[\vdash\treal]}\left(\int_{\num\Real}\Red^{n_2}(\terma''\{t/x\},\num U)\Red^{n_1-1}(u,dt)\right)\Red(\terma',du)&\text{by assoc. $\Kern$ composition}\\
&\leq\int_{\Terms[\vdash\treal]}\Red^m(\letterm xu{\terma''},\num U)\Red(\terma',du)&\text{by induction hypothesis}\\
&=\Red^{m+1}(\letterm x{\terma'}{\terma''},\num U)
\end{align*}
\todot[inline]{Je ne comprends pas la 2ème formule avec $dt$ et $du$ dans la
  même intégrale.}
 \todom[inline]{C'est une notation que j'ai retrouvé dans Panangaden. Precisement il faudrait dire:
 $
 \int_{\num\Real}\Red^{n_2}(\terma''\{t/x\},\num U)\left[\left(W\mapsto\int_{\Terms[\vdash\treal]}\Red^{n_1-1}(u,W)\Red(\terma',du)\right)(dt)\right]
 $. Dois je changer?
 }
\end{proof}

\begin{lemma}\label{lemma:anti-reduction}
Let $\vdash E[R]:\typea$ with $R\rightarrow \termb$ and $R\neq\oracle$. Then $f\prec^\typea E[N]$ implies $f\prec^\typea E[R]$.
\end{lemma}
\todot[inline]{Pourquoi on peut écarter le cas de sample?}
\todom[inline]{c'est juste que pour la preuve du key-lemma dans le cas des abstractions et de point-fixes, on utilise la cloture pour la beta-expansion et la expansion du point-fixe. Tous les autres cas n'utilisent pas cette cloture.}
\begin{proof}
Let $\typea=\typeb_1\rightarrow\dots\rightarrow\typeb_n\rightarrow\treal$, for
all $i\leq n$, take $u_i\prec^{\typeb_i}\termc_i$ and a measurable $U$:
we should prove $f\vec u U\leq \Red^{\infty}(E[R]\vec\termc,\num U)$. From the
hypothesis we get $f\vec u U\leq \Red^{\infty}(E[\termb]\vec\termc,\num U)$ and
we are done since $\Red^{\infty}(E[R]\vec\termc,\num U)=\int\Red^{\infty}(t,\num U)\Red(E[R]\vec\termc,dt)=\Red^{\infty}(E[N]\vec\termc,\num U)$.
\end{proof}

\begin{lemma}\label{lemma:fix-point}
For any $\vdash \terma:\typea$, we have: (i) $0\prec^\typea\terma$, and (ii) for any increasing family $(f_n)\subseteq\Cuball{\semtype\typea}$, $\sup_n f_n\prec^\typea\terma$, whenever $f_n\prec^\typea\terma$ for every $n$. 
\end{lemma}
\begin{proof}
Let $\typea=\typeb_1\rightarrow\dots\rightarrow\typeb_k\rightarrow\treal$, for all $i\leq n$, take $u_i\prec^{\typeb_i}\termc_i$. We clearly have $0\vec uU=0\leq\Red^\infty(\terma\vec\termc,\num U)$, so (i). Item (ii) follows from the fact that $(\sup_nf_n)\vec u=\sup_n(f_n\vec u)$ and the hypothesis that $f_n\vec uU\leq\Red^{\infty}(\terma,\num U)$ for every $n$. 
\end{proof}

\begin{proof}[Proof of Lemma~\ref{lemma:key-lemma}]
By structural induction on $\terma$. Variables are immediate from the hypothesis. The case of a constant of type $\treal$ is trivial because $\semterm{\num r}U=\Dirac{r}(U)=\Red^{\infty}(\num r,\num U)$, as well as $\semterm{\oracle}U=\lambda_{[0,1]}U=\Red^{\infty}(\oracle,\num U)$. Let $\terma=\num f(\terma_1,\dots, \terma_k)$, by induction hypothesis we have that, for every $i\leq k$, $\semterm{\terma_i}\vec u\prec^{\treal}\terma_i\{\vec\termb/\vec x\}$. We then have, for every measurable $U\subseteq\Real$: 
$\semterm\terma\vec u U
=(\semterm{\terma_1}\vec u\otimes\dots\otimes\semterm{\terma_k}\vec u)(f^{-1}(U))
\leq(\Red^{\infty}(\terma_1\vec{\{\termb/x\}},\_)\otimes\dots\otimes\Red^{\infty}(\terma_k\vec{\{\termb/x\}},\_))(\num{f^{-1}(U)})
\leq\Red^{\infty}(\terma\vec{\{\termb/x\}},U)$, where the latter inequality follows from Lemma~\ref{lemma:red_f} .

In case $\terma=\ifz\termc{\terma'}{\terma''}$, we have to prove that $\semterm{\ifz\termc{\terma'}{\terma''}}\vec u\prec^{\treal}\ifz{\overline\termc}{\overline\terma'}{\overline\terma''}$, with the overline denoting the result of applying the substitution $\{\vec\termb/\vec x\}$ to the corresponding term. Take a measurable $U$, by using the induction hypothesis on $\termc,\terma',\terma''$, we have:
$\semterm\terma\vec u U
=(\semterm\termc\vec u\{\num0\})\semterm{\terma'}\vec u U+(\semterm\termc\vec u(\num\real\setminus\{\num0\})\semterm{\terma''}\vec u U
\leq(\Red^{\infty}(\overline\termc,\{\num0\}))\Red^{\infty}(\overline\terma', U)+(\Red^{\infty}(\overline\termc,\num\real\setminus\{\num0\}))\Red^{\infty}(\overline\terma'', U)
\leq\Red^{\infty}(\overline\terma,U)
$, where the latter inequality follows from Lemma~\ref{lemma:red_if}.

In case  $\terma=\letterm{x}{\terma'}{\terma''}$, then, take a measurable $U$:
$\semterm\terma\vec u U
=\int_\Real\semterm{\terma''}\vec u\Dirac rU\semterm{\terma'}\vec u(dr)
\leq\int_{\num\Real}\Red^{\infty}(\overline\terma''\{t/x\},\num U)\Red^{\infty}(\overline\terma',dt)
\leq\Red^{\infty}(\terma,\num U)$, where the last inequality is Lemma~\ref{lemma:let}.

The other cases are standard. If $\terma=\lambda x^\typec.\terma'$, with $\typea=\typec\rightarrow\typec'$, then we have to prove for every $w\prec^{\typec}\termc$ that 
$\semterm{\terma}w\vec u\prec^{\typec'}\overline\terma\termc$. By IH we have $\semterm{\terma'}w\vec u\prec^{\typec'}\overline\terma'\{\termc/x\}$ and we conclude by Lemma~\ref{lemma:anti-reduction}. 
If $\terma=\terma'\terma''$ for  $x_1:\typeb_1,\dots,x_n:\typeb_n\vdash\terma':\typec\rightarrow\typea$ and $x_1:\typeb_1,\dots,x_n:\typeb_n\vdash\terma'':\typec$, we can immediate conclude by induction hypothesis on $\terma'$ and $\terma''$. Finally, if $\terma=\fix\termc$, then by hypothesis we have $\semterm\termc\vec u\prec^{\typea\rightarrow\typea}\overline\termc$. Then by Lemma~\ref{lemma:fix-point}.(i) $0\prec^{\typea}\fix{\overline\termc}$, hence $(\semterm\termc\vec u)0\prec^{\typea}\overline\termc(\fix{\overline\termc})$. By Lemma~\ref{lemma:anti-reduction}, we have: $(\semterm\termc\vec u)0\prec^{\typea}\fix{\overline\termc}$. By iterating the same reasoning, we get: $(\semterm\termc\vec u)^n0\prec^{\typea}\fix{\overline\termc}$ for any $n$, so that by Lemma~\ref{lemma:fix-point}.(ii) $\semterm{\terma}\vec u=\sup_n((\semterm\termc\vec u)^n0)\prec^{\typea}\fix{\overline\termc}$.
\end{proof}

\paragraph{Theorem~\ref{th:adequacy}.} Let $\vdash\terma:\treal$, then for every measurable set $U\subseteq\Real$, we have:
\[
	\semterm[\vdash \treal]\terma(U) = \Red^{\infty}(\terma,\num U)
\]
where $\num U$ is the set of numerals corresponding to the real numbers in $U$.
\begin{proof}
By iterating the soundness property (Proposition~\ref{prop:soundmess}), we have that $(\semterm[\vdash \treal]{\_}\circ\Red^n)\terma=\semterm[\vdash \treal]\terma$ for every $n$. Hence, taking a measurable $U\subseteq\Real$:
\begin{align*}
\semterm\terma U 
&= \int_{\Terms[\vdash\treal]}\semterm[\vdash \treal]{t}U\Red^n(\terma,dt)\\
&\geq\int_{\{\num r\text{ s.t. } r\in\Real\}}\semterm[\vdash \treal]{t}U\Red^n(\terma,dt)&\text{because $\{\num r\text{ s.t. } r\in\Real\}\subset\Terms[\vdash\treal]$}\\
&=\int_{\Real}\chi_U(r)\Red^n(\terma,d\num r)=\Red^n(\terma,\num U)
\end{align*} 
We conclude $\Red^{\infty}(\terma,\num U) = \sup_n\Red^{n}(\terma,\num U) \leq \semterm[\vdash \treal]\terma(U)$. 
The other inequality is a consequence of Lemma~\ref{lemma:key-lemma} and Definition~\ref{def:logical relation}.
\end{proof}

\end{document}